\documentclass[11pt,a4paper]{amsart}

\usepackage[a4paper,margin=3cm,footskip=.5cm]{geometry}

\usepackage{yfonts}
\usepackage{hyperref}
\usepackage{amsmath}
\usepackage{amssymb}
\usepackage{amsthm}
\usepackage{mathtools}
\usepackage{enumerate}
\usepackage{cite}
\usepackage{caption}
\usepackage{subcaption}
\usepackage{graphicx}
\usepackage{bm}
\usepackage{bbm}
\usepackage{float}
\usepackage{multirow}
\usepackage{makecell}
\usepackage[table]{xcolor}
\usepackage{xcolor}
\usepackage{arydshln} 
\usepackage{boldline}
\usepackage{tabu}
\usepackage{array}
\usepackage{booktabs}
\usepackage{hhline}
\usepackage{epstopdf}
\usepackage{amsaddr}  
\usepackage{mathrsfs} % https://www.ctan.org/pkg/mathrsfs
\usepackage{multirow}
\usepackage{array}
\usepackage{booktabs}
\usepackage[skip=2pt, indent=20pt]{parskip}
\usepackage{accents}
\usepackage{stmaryrd}  % Pour les doubles crochets

\newcolumntype{C}[1]{>{\centering\arraybackslash}m{#1}}

\newcommand{\EQ}{\begin{equation}}
\newcommand{\EN}{\end{equation}}
\newcommand{\EQS}{\begin{equation*}}
\newcommand{\ENS}{\end{equation*}}

\newcommand{\bea}{\bed\begin{array}{rl}}
\newcommand{\eea}{\end{array}\eed}

\newcommand{\rr}{{\hbox{{\rm I}{\kern -0.22em}{\rm R}}}}

\newcommand{\Real}{\mathfrak{R}}

\definecolor{darkgray}{rgb}{0.3, 0.3, 0.3} % Ajustez les valeurs si nécessaire

\newcommand{\mi}{\mathrm{i}}

\theoremstyle{plain}
\newtheorem{remark}{Remark}

\newtheorem{cor}{Corollary}
\newtheorem{assum}{Assumption}

\newtheorem{example}{Example}

\def\({\left(}
\def\){\right)}
\def \Q {{\mathbb Q}}

\def \C  {{\mathbb C}}

\def \AN {{\mathrm{AN}}}
\def \CN {{\mathrm{CN}}}
\def \EUR {{\mathrm{EUR}}}

\newcommand{\blue}[1]{{\color{black}#1}}

\newcommand{\TS}{\mathrm{TS}}
\newcommand{\R}{\mathbb{R}}

\newcolumntype{L}[1]{>{\raggedright\let\newline\\\arraybackslash\hspace{0pt}}m{#1}}
\newcolumntype{M}[1]{>{\centering\let\newline\\\arraybackslash\hspace{0pt}}m{#1}}

\newcommand{\BG}{\mathrm{BG}}

\newcommand{\GG}{\bm{\Gamma}^{\pi}}
\newcommand{\bs}{\bm{s}}

\newcommand{\bu}{\bm{u}}
\newcommand{\bd}{\bm{d}}

\newcommand{\GP}[2]{\bm{\Gamma}^{\bm{\pi}}\left[#1|~#2\right]}

\newcommand{\N}{\mathbb{N}}

\newcommand{\D}{\mathrm{d}}

\newtheorem{theorem}{Theorem}[section]
\newtheorem{proposition}{Proposition}[section]

\newtheorem{lemma}[theorem]{Lemma}

\newcommand{\intinterv}[2]{\left\llbracket #1, #2 \right\rrbracket}

\newcommand{\ubar}[1]{\underaccent{\bar}{#1}}

\definecolor{darkgreen}{RGB}{0,100,0}      % le "dark green" classique
\newcommand*{\green}{\textcolor{black}}
\title[Option pricing under TS processes]{Fast and explicit European option pricing under Tempered Stable processes}

\author{Gaetano Agazzotti}
\address{CEREMADE, Université Paris Dauphine-PSL
}
\email{agazzotti[AT]ceremade[DOT]dauphine[DOT]fr}

\author{Jean-Philippe Aguilar}
\address{Soci\'{e}t\'{e} G\'{e}n\'{e}rale
 }
\email{jean-philippe.aguilar[AT]sgcib[DOT]com}
\date{September 4, 2026.}

\begin{document}

\begin{abstract}
We provide series expansions for the tempered stable densities and for the price of European-style contracts in the exponential L\'evy model driven by the tempered stable process. These formulas recover several popular option pricing models, and become particularly simple in some specific cases such as bilateral Gamma process and one-sided TS process. When compared to traditional Fourier pricing, our method has the advantage of \blue{not involving delicate choices of hyperparameters, but only a truncation order for the resulting analytical pricing formulas}. We also provide a detailed numerical analysis and show that our technique is competitive with state-of-the-art pricing methods.
\end{abstract}

\maketitle

% \begin{keywords}
%       Tempered Stable processes,
%       Tempered Stable distributions, 
%       Lévy processes,
%       Mellin transform,
%       Series expansions,
%       Option pricing
% \end{keywords}
% \keywords{
% }

% \tableofcontents
%%%%%%%%%%%%%%%%%%%%%%%%%%%%%%%
%%%%%%%%%%%%%%%%%%%%%%%%%%%%%%%
%%%%%%%%%%%%%%%%%%%%%%%%%%%%%%%
\section{Introduction}
% {\color{red} Define the characteristic function on the complex strip $\mathscr{S}_\TS = \{u\in\C|~\mathrm{Im}[z]\in(-\lambda_+,\lambda_-)\}$}

%{\color{red} Préciser que $\N = \{0,1,...\}$}
The absence of normality in stock returns has been observed since at least the pioneering essay by Bachelier \cite{Bachelier14} in the early twentieth century and has motivated the introduction of stochastic processes that go beyond the Brownian motion\blue{,} notably in terms of \blue{skewness and kurtosis,} to model the logarithm of stock prices. Among these processes, L\'evy processes (see the classical monographs \cite{Bertoin96,Cont04,ken1999levy} among many other references) play a leading role as they capture many stylized facts such asymmetry in the distributions and fat tails, while preserving the independence and stationarity of financial returns. Among the most popular L\'evy processes in quantitative finance, one finds the Variance Gamma (VG) process of \cite{Madan98}, the bilateral Gamma (BG) process of \cite{kuchlerBG}, the CGMY process of \cite{Carr02} or the KoBoL process of \cite{BoyarOption00}; more recent extensions also include additive processes \blue{whose applications to financial modeling have been discussed for instance in} \cite{Azzone25,Amici25} and that relax the stationarity hypothesis, or convolutions of generalized inverse Gaussian subordinators \blue{such as introduced} in \cite{AgAg24} for instance.

\blue{The} VG, BG, CGMY and KoBoL processes actually belong to a broader class of processes, known as tempered stable (TS) processes. \blue{They are related to} the so-called truncated L\'evy flights introduced in physics in \cite{Mantegna94} \blue{and} \cite{Koponen95} with the purpose to capture two different behaviors: a central part that is approximately L\'evy stable, and tails that are approximately exponential in order to ensure finiteness of the second moment. \blue{The terminology of TS process, however, is not uniquely defined in the literature and has evolved over time; nowadays, the consensus is probably that TS processes form a six parameter class of L\'evy processes such as defined and extensively studied in \cite{rosinski2007, KuchlerTS}, and that have previously been named generalized tempered stable (GTS) processes for instance in \cite{Cont04,rachev2011financial}. 

Although initially introduced in phyics,}
% \blue{or, sometimes, as generalized tempered stable (GTS) processes, and whose properties have been extensively studied in \cite{KuchlerTS}}; they \blue{are related to} the so-called truncated L\'evy flights introduced in physics in \cite{Mantegna94} \blue{and} \cite{Koponen95} with the purpose to capture two different behaviors: a central part that is approximately L\'evy stable, and tails that are approximately exponential in order to ensure finiteness of the second moment. Since then, 
TS processes have also proved popular in finance notably in the framework of exponential L\'evy pricing models as described in \cite{exp-ts-kuchler} (see also more information on the change of measure and Esscher transform for TS processes in \cite{grabchak2016tempered}), thanks to their tractability and the good empirical results they provide. Among the many financial applications of TS processes, one can mention energy market modeling and pricing of associated derivatives in \cite{Sabino22,Sabino23}, incorporation of stochastic volatility in \cite{zaevski2014tsstochvol, elkathib2025pricingtsvolsto} or multivariate extensions and basket option pricing in \cite{Xia11022022, xia2024pricing}. Equity indices calibration has also been conducted in \cite{fallahgoul2021modelling}, and simulation algorithms for TS processes have been provided in \cite{Chakrabarty2011tsrw, devroye2009simul, baeumer2010simul, grabchak2021simul}.

% More information on the change of measure and the Esscher transform for these processes is detailed in \cite{grabchak2016tempered}[Chap 7.1].
% In \cite{Sabino22,Sabino23}, the authors model the energy market with TS processes and are then able to price associated derivatives products. 
% Refinements of TS processes with the addition 
% of stochastic volatility components and associated pricing method have been proposed in \cite{zaevski2014tsstochvol, elkathib2025pricingtsvolsto}. Algorithms to simulate these processes can be found in \cite{Chakrabarty2011tsrw, devroye2009simul, baeumer2010simul, grabchak2021simul} and the multivariate case is handled in \cite{Xia11022022}. In \cite{xia2024pricing}, a pricing method is proposed for the multi-asset options case.
% A great effort has also been dedicated to study the numerical aspects of TS processes. In \cite{fallahgoul2021modelling}, the authors provide an empirical study of TS modeling for market data.

% An empirical study of TS modelling for market data is provided in \cite{fallahgoul2021modelling}. of  A detailed study of the TS processes in finance can be found in \cite{bianchi}.
% In \cite{xia2024pricing} 

One of the reasons for the popularity of TS processes\blue{,} and of L\'evy processes in general\blue{,} in option pricing applications is the availability of their characteristic function in closed form, allowing for the usage of Fourier pricing methods, that are easy to implement and benefit from many refinements that have been achieved to increase their performance. In particular, the introduction of the Fast Fourier Transform (FFT) in Finance in \cite{Carr99} paved the way to the current state-of-the-art Fourier methods such as the COS method \blue{of} \cite{fang2009novel} or the PROJ method \blue{of} \cite{kirkby2015efficient}, and extensive research has been conducted to extend the methods to exotic options \blue{in} \cite{kirkby2017bermudan, kirkby2018american, kirkby2020asian}. However, Fourier related methods are known to be subject to numerical inaccuracies stemming from multiple valued function such as complex square roots or logarithms (see \cite{AlbrecherTrap,abi25}); moreover, at least one prefixed parameter has to be selected to perform the computations. For the FFT, the option price has to be modified\blue{, or} \textit{dampened}\blue{,} to ensure square integrability, and the convergence of the method is highly sensitive to the choice of the damping parameter, in particular in the multivariate case\blue{;} see \blue{in particular} \cite{samet2024}, where optimization of damping factors is discussed to accelerate the convergence of numerical quadrature methods in the context of Fourier pricing. % is the so-called damping factor. Recent results provided in \cite{samet2024} show that the choice of this parameter has an important impact on the convergence of the Carr-Madan pricing formula. The authors show that the optimal factor is the solution to a minimization problem. 
For the COS method, no damping parameters are required, however a truncation range has to be selected, into which the log returns density is approximated by a finite cosine expansion; this truncation is generally chosen through a ``rule of thumb" based on the first four cumulants of the distribution, that can actually lead to serious mispricings, see \cite{Junike22}; in some cases, the numerical inaccuracies of the FFT and COS methods and \blue{of other popular Fourier techniques such as the Lewis-Lipton formula} \cite{Lewis01} have also been observed to compensate the model error itself, giving birth to what has been called in \cite{Boyar2015,Boyar24} a ``ghost calibration" and to spurious volatility smiles and skews. Similarly, the PROJ method also does not require the selection of a control parameter and displays great efficiency and robustness \blue{ as desmontrated in} \cite{kirkby2020asian,kirkby2017bermudan}, however, like in the COS case, it necessitates the introduction of a truncation for the density support, that is also determined through a cumulant-based rule and needs to be adapted when it comes to very precise option pricing; additionally, and like in the COS case, PROJ pricing also includes a series truncation, corresponding to the discretization of the integration interval and controlling the precision of the characteristic function's projection over the Riesz basis. 

The purpose of this paper is therefore to develop an option pricing method under TS processes that minimizes the numerical disadvantages described above, and allows for an arbitrary degree of precision up to machine capabilities within a competitive computational time. 
The solution we propose is based on series expansions generated by multidimensional residues associated to Mellin-Barnes integrals. We note that the Mellin transform has already been used for option pricing purposes, notably via transformations of partial differential equations, for instance in \cite{panini2004option,gzyl2017,choi2022mellin,yoon2015}. \blue{However, its} combination with residue calculus, \green{is less frequently} \blue{employed in the financial literature. Yet, as shown in \cite{gerhold2016, AguilarKirkbyResidue}, it is a powerful technique to get precise asymptotics of several relevant quantities such as the slope of the implied volatility smile around the money or the price of several time independent contracts.} 
% is more recent and dates back to \cite{AguilarKirkbyResidue} in the context of path-independent contracts under exponential L\'evy models. 
Refining this technique, we are able to derive new and fast converging series representations for TS distributions and also for digital and European \blue{call} option prices, and we recover and improve the formulas previously obtained in \cite{AguilarKirkbyResidue} and \cite{aguilarkirkby2022robust} in the VG, BG and one-sided cases. Moreover, 
our \blue{formulas being analytical}, the precision of the computed prices depending only on the chosen summation order.

% As evidenced in our numerical experiments, the choice of this parameter can be crucial if one needs to reach a given precision (a too wide integration domain leads to a too high summation order (and an important computational time), whereas a too narrow integration domain leads to a lack of precision). 

The paper is organized as follows: we first recall in section \ref{sec:TS_series} the definition of TS distributions, their associated L\'evy processes as well as their particular cases. Then, in section \ref{sec:serie-ts-distributions}, we provide series expansions of the TS distributions by means of multidimensional Mellin transform inversion; these formulas are new to the literature (they only existed in the single-sided case so far). In section \ref{sec:opt-pricing-ts-processes}, we derive series expansions for digital and European \blue{call} option under TS processes and, in section \ref{sec:particular-cases}, we focus on the particular cases covered by the TS model, and derive simpler \blue{expansions notably} for the BG and VG models. In section \ref{sec:opt-pricing-one-sided}, we adapt our methodology to the one-sided case and provide particularly simple pricing formulas, taking the form of series expansions; finally, we conduct in section \ref{sec:num-results} detailed numerical analysis and provide comparison with state-of-the-art pricing methods. 
To further aid in research and reproducibility, the implementation code is provided in the publicly available Github repository {\sffamily TS-Pricing}\footnote{\url{https://github.com/gagazzotti/TS-Pricing}} written in {\sffamily Python}. \blue{Following the acknowledgment section, }
section \ref{sec:concl} is dedicated to conclusive remarks.

% We want to price options of the form:
% \begin{equation}
%     C = \E [ p(S_T)].
% \end{equation}

% Modelling $S$ as an exp-Lévy process, yields:
% \begin{equation}
%     C = \int_\R S_0e^xf(x)\D x
% \end{equation}
% with $f$ the density of $\log(S_T/S_0)$.

% {
% \color{red} 
% TODO
% \begin{itemize}
%     \item define the payoffs for AN, CN, European
%     \item use "generalized TS" instead of TS
% \end{itemize}
% }

%%%%%%%%%%%%%%%%%%%%%%%%%%%%%%%
%%%%%%%%%%%%%%%%%%%%%%%%%%%%%%%
%%%%%%%%%%%%%%%%%%%%%%%%%%%%%%%

%%%%%%%%%%%%%%%%%%%%%%%%%%%%%
\section{TS distributions}
% {\color{black}
\label{sec:TS_series}
%%%%%%%%%%%%%%%%%%%%%%%%%%%%%

%%%%%%%
\subsection{Notation}
%%%%%%%%

In all of the following, we will consider a probability space $(\Omega,\mathcal{F},\mathbb{P})$ and, unless otherwise mentioned, expectations will be taken under $\mathbb P$, that is $\mathbb E [.] = \mathbb E ^{\mathbb P} [.]$. Given a random variable $X:\Omega \rightarrow \mathbb R$, we define its characteristic function (Fourier transform) as
\begin{equation}
    \forall u\in\R,\quad
    \varphi_X (u) 
    =
    \mathbb E \left[ e^{\mi u X} \right]
    =
    \int_{\mathbb R} e^{\mi u x} f_X(x) \D x
    ,
\end{equation}
where the second equality holds if we assume that $X$ possesses a $\mathbb P$-density $f_X$. In the following, we will also denote by $\varphi_X$ the analytic continuation of the characteristic function whenever it exists. 

%%%%%%%%
\subsection{One-sided and double-sided tempered stable distributions}
%%%%%%%%
We start by recalling some essential facts about $\TS$ distributions. The reader can find further details in \cite{KuchlerTS} (as well as multiple other references therein).

%%%
\subsubsection{One-sided TS distributions}
%%%%

One-sided $\TS$ distributions are probability distributions supported by $\mathbb R_+$ and whose characteristic function $\varphi_{\TS_+}$ is of the form:
\begin{equation}
    \label{eq:chf-levy-measure}
    \forall u\in\R,\quad\varphi_{\TS_+}(u) 
    =\exp\left(
    \int_\R\left(e^{\mi u x}-1\right)\nu_{\TS_+}(x)\D x
    \right),
\end{equation}
where the Lévy measure $\nu_{\TS_+}$ is given by:
\begin{equation}\label{Levy_measure_TS+}
    \forall x\in\R, \quad \nu_{\TS_+}(x) = \alpha\frac{e^{-\lambda x}}{x^{1+\beta}}\mathbbm{1}_{(0,\infty)}(x)
    ,
\end{equation}
for $\alpha,\lambda\in(0,\infty)$ and $\beta\in [0,1)$ \blue{(some authors, such as \cite{rosinski2007}, \green{exclude the case $\beta=0$})}. We will denote such distributions by $\TS_+(\alpha,\lambda,\beta)$, and their density by $f_{\TS_+}(x;\alpha,\beta,\lambda)$; when there is no ambiguity on the parameters and to simplify the notation, we will simply write $f_{\TS_+}(x)$. 
\green{Existence of densities is a direct consequence of the fact that tempered stable distributions are self-decomposable (see \cite{KuchlerTS}), and hence absolutely continuous (see \cite[Example 27.8]{ken1999levy}).} 

\blue{The case $\beta=0$ corresponds to the classical gamma distribution; the characteristic function is well known in this case, and writes $\varphi_{\mathrm{Gamma}}(u) = ( \lambda / (\lambda- \mi u))^\alpha$.}
If $\beta\neq 0$, then the characteristic function of a random variable $X\sim \TS_+(\alpha,\beta,\lambda)$ is known in closed form \green{as well}, and can be written explicitly as:
\begin{equation}\label{char_TS+}
    \forall u\in\R,
    \quad\varphi_{\TS_+}(u) 
    =
    \exp\left({\alpha\Gamma(-\beta)\left((\lambda-\mi u)^{\beta}-\lambda^{\beta}\right)}\right).
\end{equation}
As $\TS_+(\alpha,\beta,\lambda)$ distributions are infinitely divisible, they generate a Lévy process \blue{denoted by} $X^+=(X_t^+)_{t\geq 0}$; \blue{the characteristic function of $X_t^+$ is then, \green{with a slight abuse of notation}:}
\begin{equation}
    \forall (u,t)\in\R\times(0,\infty),\quad\varphi_{\TS_+}(u,t) 
    =
    \green{(\varphi_{\TS_+}(u))^t}
    =
    \exp\left(t\alpha\Gamma(-\beta)\left((\lambda-\mi u)^{\beta}-\lambda^{\beta}\right)\right).
\end{equation}
The process $X^+$ is an almost surely non decreasing process (i.e., it is a subordinator), whose increments are independent and distributed according to:
\begin{equation}
\forall t\geq s\geq 0, \quad X_t^+-X_s^+\sim \TS_+(\alpha (t-s),\lambda,\beta). 
\end{equation}

\begin{remark}
    Under the same assumptions on the parameters, one sided TS distributions can also be defined on the negative axis in an evident way. The characteristic function in this case is given by:
\green{
    \begin{equation}
    \label{eq:chf-levy-measure-ts-}
    \forall u\in\R,\quad\varphi_{\TS_-}(u) 
    =\exp\left(
    \int_\R\left(e^{\mi u x}-1\right)\nu_{\TS_-}(x)\D x
    \right),
\end{equation}
where the Lévy measure is:
\begin{equation}
\label{eq:levy-measure-ts-}
\forall x\in\R,\quad
    \nu_{\TS_-}(x) = \alpha\frac{e^{-\lambda |x|}}{|x|^{1+\beta}}\mathbbm{1}_{(-\infty,0)}(x).
\end{equation}
\green{Like in the positive case,} when $\beta\neq 0$, the characteristic function can be re-written as:
}
    \begin{equation}\label{char_TS-}
    \forall u\in\R,
    \quad\varphi_{\TS_-}(u) 
    =
    \exp\left({\alpha\Gamma(-\beta)\left((\lambda+\mi u)^{\beta}-\lambda^{\beta}\right)}\right)
    .
    \end{equation}
    Last, the density function satisfies
    \begin{equation}
        \forall x \in (-\infty,0),\quad
        f_{\TS_-}(x;\alpha,\beta,\lambda) = f_{\TS_+}(-x;\alpha,\beta,\lambda).
    \end{equation}
    Such distributions will be denoted by $\TS_-(\alpha,\beta,\lambda)$. 
    \green{We note that $X\sim \TS_-(\alpha,\beta,\lambda) $ if and only if $-X\sim \TS_+(\alpha,\beta,\lambda) $.}
\end{remark}

% A well known procedure of bilateralization can be applied to TS distributions (see kuchler).
% To this aim, we consider a ``positive" and a ``negative" TS distributions respectively denoted  $\TS_+(\alpha_+,\lambda_+,\beta_+)$ and $\TS_-(\alpha_-,\lambda_-,\beta_-)$ with Lévy measures:
% \begin{equation}
% \forall x\in\R, \quad 
%     \begin{cases}
%         \displaystyle
%         \nu_{\TS_+}(x) = \alpha_+\frac{e^{-\lambda_+ x}}{x^{1+\beta_+}}\mathbbm{1}_{(0,\infty)}(x),\\
%         \displaystyle
%         \nu_{\TS_-}(x) = \alpha_-\frac{e^{-\lambda_- |x|}}{x^{1+\beta_-}}\mathbbm{1}_{(-\infty,0)}(x),
%     \end{cases}
% \end{equation}
% where $\alpha_\pm,\lambda_\pm>0$ and $\beta_\pm\in[0,1]$.

%%%
\subsubsection{Double-sided TS distributions}
%%%%

We now define the double-sided TS distribution or, simply, the TS distribution, by the following convolution:
\begin{equation}\label{TS_def}
    \TS(\alpha_+,\beta_+,\lambda_+,\alpha_-,\beta_-,\lambda_-) 
    =
    \TS_+(\alpha_+,\beta_+,\lambda_+) \star \TS_-(\alpha_-,\beta_-,\lambda_-) 
\end{equation}
where $\alpha_\pm, \lambda_\pm \in \blue{(0,\infty)}$ and $\beta_\pm \in [0,1)$. \green{We note that if $Y\sim \TS_+(\alpha_+,\beta_+,\lambda_+)$ and $Z\sim \TS_+(\alpha_-,\beta_-,\lambda_-)$, the random variable $X:= Y-Z\sim \TS(\alpha_+,\beta_+,\lambda_+,\alpha_-,\beta_-,\lambda_-)$.} We note that some authors \blue{use the terminology of} generalized tempered stable distributions (GTS) \blue{instead of TS distributions}, see for instance \cite{Cont04}.
% \textit{i.e.}, $\mathrm{TS}(\alpha_+,\lambda_+,\beta_+,\alpha_-,\lambda_-,\beta_-) = \TS_+(\alpha_+,\lambda_+,\beta_+) \star \TS_-(\alpha_-,\lambda_-,\beta_-)$.
As for one sided TS distributions, TS distributions are infinitely divisible and, as such, they generate a Lévy process; it follows from definition \eqref{TS_def} and from \eqref{char_TS+} and \eqref{char_TS-} that its characteristic function is
\begin{multline}
\label{eq:char_TS}
% \begin{aligned}
    \varphi_{\mathrm{TS}}(u,t)      =
    \green{(\varphi_{\TS_+}(u)\varphi_{\TS_-}(u))^t}
    =
     \exp\left(t\alpha_+\Gamma(-\beta_+)\left((\lambda_+-\mi u)^{\beta_+}-\lambda_+^{\beta_+}\right)\right.\\
     +\left.t\alpha_-\Gamma(-\beta_-)\left((\lambda_-+\mi u)^{\beta_-}-\lambda_-^{\beta_-}\right)
     \right)
% \end{aligned}
\end{multline}
\green{for $\beta_\pm \in (0,1)$ and
for all $(u,t)\in\R\times(0,\infty)$.}
Extension to the case $1\green{\leq} \beta_\pm < 2$ has been considered (see discussion in \cite{KuchlerTS}), however we will stick to the initial hypothesis $\beta_\pm\in[0,1)$ because it covers \blue{many} cases of financial significance \blue{notably for equity markets} (see section \ref{subsec:particular-cases} thereafter). Last, it follows from definition \eqref{TS_def} that the density of a $\TS$ distribution admits the \blue{following representation as a densities convolution:}
\begin{equation}\label{TSdensity_conv}              f_\TS(x;\alpha_+,\beta_+,\lambda_+,\alpha_-,\beta_-,\lambda_-) = \int_0^\infty f_{\TS_+}(x+y;\alpha_+,\beta_+,\lambda_+)f_{\TS_+}(y;\alpha_-,\beta_-,\lambda_-)\D y
\end{equation}
for all $x\in\R$. \blue{In the following, where there is no ambiguity on the parameters and to simplify the notation, we will simply write $f_{\TS}(x)$.}

\begin{remark}
Let us note that, if $Y$ and $Z$ are two independent TS subordinators whose generating distributions are respectively $\TS_+(\alpha_+,\lambda_+,\beta_+)$ and $\TS_+(\alpha_-,\lambda_-,\beta_-)$, then the difference process $X:=Y-Z$ is a $\TS$ process whose increments are independent and satisfy:
\begin{equation}
\forall t > s, \quad
X_t-X_s \sim \mathrm{TS}(\alpha_+(t-s),\lambda_+,\beta_+,\alpha_-(t-s),\lambda_-,\beta_-).
\end{equation}
\end{remark}

%%%%%%%%%%%%%%%%%%%%%%%%
%%%%%%%%%%%%%%%%%%%%%%%%
\subsection{A continuity property}
\label{subsec:continuity}
Let us derive an interesting continuity property that will be useful in the option pricing context\blue{.} \green{Weak convergence of sequences of TS distributions was already studied in \cite[Prop. 3.1(1)]{KuchlerTS}; in the following proposition, we build on this result to establish pointwise convergence of the corresponding sequences of TS densities.} 

\begin{proposition}
    \label{prop:pointwise-convergence}
    Let $(\alpha_+,\lambda_+,\alpha_-,\lambda_-)\in (0,+\infty)^4$ 
    and $(\beta_n^\pm)_{n\in\N}\in (0,1)^\N$
    be two sequences such that 
    $\beta_n^\pm \underset{n\to\infty}{\longrightarrow}\beta_\pm\in(0,1)$. 
    Let $f$ \green{(resp. $f_n$ for $n\in\N$) be} the density function 
    % and 
    % the characteristic function 
    of the distribution
     $\TS(\alpha_+, \beta_+,\lambda_+, \alpha_-,\beta_-,\lambda_-)$ \green{(resp. 
    %  and for each $n\in\N$, let $f_n$ and $\varphi_n$ being respectively the density function 
    % and the characteristic function of the distribution
     $\TS(\alpha_+, \beta_n^+,\lambda_+, \alpha_-,\beta_n^-,\lambda_-).$ )}
     We have the following pointwise convergence:
    \begin{equation}
        \forall x \in\R, \quad \lim_{n\to+\infty}f_n(x)
        =f(x).
    \end{equation}
\end{proposition}

\blue{\begin{proof}
\green{Let $\varphi$ (resp. $\varphi_n$ for $n\in\N)$) be respectively the characteristic function 
of the
     $\TS(\alpha_+, \beta_+,\lambda_+, \alpha_-,\beta_-,\lambda_-)$ (resp.  
    %  and for each $n\in\N$, let $f_n$ and $\varphi_n$ being respectively the density function 
    % and the characteristic function of the distribution
     $\TS(\alpha_+, \beta_n^+,\lambda_+, \alpha_-,\beta_n^-,\lambda_-)$ distribution.}
    Let us first show that there exists an integrable function $g$ such that $|\varphi_n|\leq g$. From definition \eqref{TS_def} and from the \green{Lévy measures \eqref{Levy_measure_TS+} and \eqref{eq:levy-measure-ts-}}, we have:
    % We can write $\varphi_n$ with the Lévy-Kinhchine representation, \textit{i.e.}:
    \begin{equation}
        \forall u\in\R,\quad\varphi_{n}(u) 
        =\exp\left(
        \int_\R\left(e^{\mi u x}-1\right)\frac{k_{n}(x)}{x}\D x
        \right)
    \end{equation}
    where $k_n(x) = k_n^+(x)\mathbbm{1}_{(0,\infty)}(x) - k_n^-(x)\mathbbm{1}_{(-\infty,0)}(x)$ with $k_n^\pm(x)= \alpha_\pm |x|^{-\beta_n^\pm}e^{-\lambda_\pm |x|}$. Writing $e^{\mi u x} = \cos (ux) + \mi \sin(ux)$, we get:
    \begin{equation}
        \label{eq:bound-levy-measure}
        \begin{aligned}
        \forall u\in\R,~|\varphi_{n}(u)| &\leq 
        \exp\left(
        \int_\R \left(\cos(ux)-1\right)\frac{k_{n}(x)}{x}\D x
        \right)\\
        &=\exp\left(
        \int_0^\infty \left(\cos(ux)-1\right)\frac{(k_{n}^+(x)+k_n^-(x))}{x}\D x
        \right)
        \end{aligned}
    \end{equation}
    where the equality follows from the change of variable $x\gets -x$ on $(-\infty, 0)$. 
    By lemma \ref{lem:bound-integral}, we have:
    \begin{equation}
        \int_0^\infty \left(\cos(ux)-1\right)\frac{k_n^\pm(x)}{x}\D x \leq 
        \alpha_\pm 
        \green{g(u;\beta_n^\pm,\lambda_\pm)}
    \end{equation}
    where \green{$g(\cdot;\beta_n^\pm,\lambda_\pm)$} 
    % do not 
    % depend on the sequences $(\beta_n^\pm)_{n\in\N}$
     is defined as:
    \green{\begin{equation}
        g(u;\beta_n^\pm,\lambda_\pm)
        =
        \begin{cases}
            \displaystyle
            0 \quad &\text{if}~|u|\leq 1,\\
            \displaystyle
            -\frac{11}{24(2-\beta)}e^{-\lambda_\pm}|u|^{\beta_n^\pm} \quad &\text{if}~|u|\geq 1.
        \end{cases}
    \end{equation}
    Since $\beta_n^\pm\to\beta^\pm$, there exits $n_0\in\N$ such that for all $n\geq n_0$, we have $\beta_n^\pm\geq \beta^\pm/2$ and thus, for all $n\geq n_0$ we can write:
    \begin{equation}
        \forall u\in\R,\quad 
        g(u;\beta_n^\pm,\lambda_\pm)\leq 
        g(u;\beta^\pm/2,\lambda_\pm)
    \end{equation}
    since $|u|^{\beta_n^\pm}\geq |u|^{\beta^\pm/2}$ for all $|u|\geq 1.$
    }
    Combined with \eqref{eq:bound-levy-measure}, we have:
    \begin{equation}
        \label{eq:bound-chf}
        \forall u\in\R,\quad |\varphi_{n}(u)| 
        \leq \exp\left(
        \alpha_+\green{g(u;\beta^+/2,\lambda_+)}+\alpha_-\green{g(u;\beta^-/2,\lambda_-)}
        \right)
    \end{equation}
    where the right hand side of \eqref{eq:bound-chf} is integrable on $\R$. Since it follows from the Fourier inversion formula that
        \begin{equation}
        \forall x\in\R, 
        \quad 
        f_n(x) = \frac{1}{2\pi}\int_\R e^{-\mi ux}\varphi_n(u)\D u,
    \end{equation}
    \green{and that $\varphi_n(u)\to_{n\to \infty}\varphi(u)$ for all $u\in\R$ by \cite[Prop. 3.1 (1)]{KuchlerTS}},
    we conclude, by \eqref{eq:bound-chf} and the dominated convergence theorem, that
    \begin{equation}
        \forall x\in\R,\quad \lim_{n\to\infty} f_n(x) = f(x).
    \end{equation}
    % Since we have the pointwise convergence of the characteristic 
    % functions $\varphi_n\to \varphi$ and the Fourier inversion theorem:
    % \begin{equation}
    %     \forall x\in\R, 
    %     \quad 
    %     f_n(x) = \frac{1}{2\pi}\int_\R e^{-iux}\varphi_n(u)\D u
    % \end{equation}
    % we have, since all the $|\varphi_n|$ are 
    % dominated by an integrable function independent of $n$, that:
    % \begin{equation}
    %     \forall x\in\R, 
    %     \quad 
    %     \lim_{n\to\infty}\int_\R e^{-iux}\varphi_n(u)\D u
    %     = \int_\R e^{-iux}\lim_{n\to\infty}\varphi_n(u)\D u
    %     =\int_\R e^{-iux} \varphi(u)\D u.
    % \end{equation}
    % The same Fourier inversion theorem applied to $\varphi$ gives the pointwise convergence:
    % \begin{equation}
    %     \forall x\in\R,\quad \lim_{n\to\infty} f_n(x) = f(x).
    % \end{equation}
\end{proof}}

\subsection{Particular cases}
\label{subsec:particular-cases}
% \red{to do: give examples}
% }

TS distributions recover many well-known distributions that are popular within the option pricing and financial engineering space; below are some examples.

\begin{itemize}
    \item \blue{$\beta_+=\beta_-=:\beta>0$} in \eqref{eq:char_TS} yields the characteristic function of the KoBoL \green{(Boyarchenko-
Koponen-Levendorskii)} distribution (see \cite{BoyarOption00,BoyarLevin02}).
    \item If $\alpha_+=\alpha_-\blue{=:\alpha>0}$ and $\beta_+=\beta_-\blue{=:\beta <2}$ then the TS distribution becomes a CGMY (Carr-Geman-Madan-Yor) distribution, where $C=\alpha$, $G=\lambda_-$, $M=\lambda_+$, $Y=\beta$. We note, moreover, that the case $Y\in[0,1)$ \blue{to which we restrict ourselves} is the case where the associated L\'evy process is of infinite activity and finite variation, a situation \blue{that allows for greater flexibility and differences between the historical and risk-neutral processes}, and is typical for risk-neutral processes of most equity indices for instance (see \cite{Carr02}). \blue{We note however that the situation can be different on other markets, where a calibrated $\beta>1$ can occur more frequently (see for instance \cite{kirkbyreturn}) for the FX case.}
    \item If $\beta_\pm=0$, \blue{it reduces to} a bilateral Gamma (BG) distribution (see \cite{kuchlerBG}) and in that case the characteristic function becomes:
\begin{equation}\label{eq:char_BG}
    \forall u\in\R,\quad\varphi_{\mathrm{BG}}(u)      
    =
    \left(
    \frac{\lambda_+}{\lambda_+ - \mi u}
    \right)^{\alpha_+}
    \left(
    \frac{\lambda_-}{\lambda_+ + \mi u}
    \right)^{\alpha_-}
    .
 \end{equation}
 \item When $\alpha_+=\alpha_-:=\alpha$ and $\beta_+=\beta_-:=0$, one recovers the Variance Gamma (VG) distribution of \cite{Madan98}; when furthermore $\lambda_+=\lambda_-$, the distribution is symmetric, and was introduced in \cite{Madan90}.
\end{itemize}

%%%%%%%%%%%%%%%%%%%%%%%%
%%%%%%%%%%%%%%%%%%%%%%%%
% % \input{sections/3_series_ts_distributions}
\section{Series expansion of TS distributions}
\label{sec:serie-ts-distributions}

The purpose of this section is to derive series representations for single and double-sided TS distributions. Our approach is the following: first, we derive a Mellin-Barnes (MB) integral representation for the corresponding density function, \green{and next}, we apply residue calculus to express these integrals under the form of residues series expansions. Of course, series (and asymptotic) representations for L\'evy densities have already been discussed in the literature but they often focus on Gram-Charlier expansions (see for instance \cite{Chateau17,Asmussen22} and references therein) or on expansions for the cumulative distribution function (see \cite{Navas25}). Mellin residue calculus for TS densities has been utilized in \cite[eq. 3.13]{gupta2021densitiesinversetemperedstable} to get a density expansion in the one-side case, that we will recover in proposition \ref{prop:series_density_TS}; the general double-sided series expansion we will obtain in proposition \ref{prop:series_density_GTS} is the first of its kind (to the best of our knowledge).  
\green{In sections \ref{sec:serie-ts-distributions} and \ref{sec:opt-pricing-ts-processes}, we will assume $\beta^\pm\neq 0$, the particular case $\beta^\pm=0$ (BG and VG distributions) being addressed independently in section \ref{sec:particular-cases}.}

%Even though this procedure is an intermediary step to derive series for option prices, 
% In the one-sided TS case, we show that our expansion recovers the already known series expansion of \cite[eq. 3.13]{gupta2021densitiesinversetemperedstable}. In the general double-sided case, the expansion we will derive in proposition \ref{prop:series_density_GTS} for is the first of its kind (to the best of our knowledge).
% derive a new series representation for the double-sided TS distribution (the first of its kind to the best of our knowledge). 

\subsection{Notation}

% This section introduces the notations that will ease
% the use Mellin-Barnes integral. These integrals are of the form:
% \begin{equation}
%     \int_{\bc+\mi \R^{n}} 
%     \GG\left(U\bs_{i:j}^{(n)}+ \bu, D\bs_{k:\ell}^{(n)}+\bd\right)
%     \prod_{i=1}^n a_i^{\bm{s}_i^{(n)}}
%     \frac{\D \bs^{(n)}}{(\mi2\pi)^n}
% \end{equation}
% where $\bc\in\C^n$, the integration variable is $\bs^{(n)} 
% := \left(\bs_1 ,...,\bs_n\right)^\intercal
% \in \bc+\mi\R^n$ and the differential form 
% $\D \bs^{(n)}$ is defined as
% $\D \bs^{(n)} := \D \bs_1\wedge\dots\wedge \D \bs_n$. 
% For $1\leq i\leq j\leq n$, the sub-vector $\bs_{i:j}^{(n)}$ is defined as $\bs_{i:j}^{(n)} = \left(\bs_i ,\dots,\bs_j\right)^\intercal$.
% The integrand is defined as:
% \begin{equation}
%     \GG\left(U\bs_{i:j}^{(n)}, D\bs_{k:\ell}^{(n)}\right)
%     := 
%     \frac{
%     \prod_{p=1}^{m_u}\Gamma\left(\left(U\bs_{i:j}^{(n)}\right)_p+\bu_p\right)}
%     {\prod_{p=1}^{m_d}\Gamma\left(\left(D\bs_{k:\ell}^{(n)}\right)_p+ \bd_p\right)}
% \end{equation}
% where $(U,D)\in\mathcal{M}_{m_u,j-i+1}(\R)\times \mathcal{M}_{m_d,\ell-k+1}(\R)$ with $(m_u,m_d)\in\N_{>0}^2$. The vectors $\bu$ and $\bd$ are assumed to be such that $(\bu, \bd)\in \R^{m_u}\times\R^{m_d}$.

As a preamble, we first introduce notation that will be used extensively in the article, especially for expressing multiple MB integrals, that is, contour integrals involving ratios of products of gamma functions of linear arguments - see \cite{Abramowitz72} or any monograph on special functions - in a \blue{more compact form}. 

To this aim, let us first introduce some quantities. 
Let $n\in\mathbb{N}$ be the dimension of the MB integral ($n=2$ for a double integral, $n=3$ for a triple integral etc.), and let $(m_u,m_d)\in\mathbb{N}^2$ denote the number of gamma functions in the numerator and in the denominator respectively. We also define two other integers \blue{$\ell_u,\ell_d\in\intinterv{1}{n}:= \{1,\dots,n\}$} denoting the dimension of the linear span generated by the complex variables involved in the numerator and denominator respectively. %the number of complex variables really needed as argument of these gamma functions (lengths of the subvectors of the complex integration variable). 
% We also need the vectors  $(\bm{c},\bm{a})\in(\R^n)^2$ indicating respectively the integration complex lines and the term that will be powered to the complex integration variables. 
Let us finally consider the two matrices $(U,D)\in\mathcal{M}_{m_u,\ell_u}(\R)\times\mathcal{M}_{m_d,\ell_d}(\R)$\blue{, where $\mathcal{M}_{p,q}(\R)$ denotes the set of $p\times q$ real matrices,} and the two vectors $(\bm{u},\bm{d})\in\R^{m_u}\times\R^{m_d}$
containing
the coefficients of the arguments of the gamma functions, and let us define the integration subset
\begin{equation}
    \bm{c}+\mi\R^n
    :=
    \{
    \bm{z}\in\C^n,
    z_1 = c_1 + \mi y_1,
    \dots,
    z_n = c_n + \mi y_n
    , \quad
    (y_1,\dots,y_n) \in \R^n\blue{\}}.
\end{equation}
With this notation, an $n$-dimensional MB integral can be written as:
\begin{equation}
    \label{eq:int-mb-def}
    \int\limits_{\bm{c}+\mi\R^n}
    \GP{U\bm{s}_{\leqslant\ell_u}+\bm{u}}{D\bm{s}_{\leqslant\ell_d}+\bm{d}}
    \prod_{i=1}^n q_i^{s_i}\frac{\D \bm{s}}{(\mi 2\pi)^n}
\end{equation}
where $\bm{q}\in\R^n$,  $\bm{s}_{\leqslant\ell_u}:=(s_1,...,s_{\ell_u})$, $\bm{s}_{\leqslant\ell_d}:=(s_1,...,s_{\ell_d})$, $\D \bm{s} := \D s_1 \wedge ... \wedge \D s_n$, and where $\GP{\cdot}{\cdot}$ is defined as:
\begin{equation}
    \forall (\bm{v},\bm{w})\in \R^{m_u}\times\R^{m_d},
    \quad 
    \GP{\bm{v}}{\bm{w}}
    :=
    \frac{\prod_{i=1}^{m_u} \Gamma\left(v_i\right)}{\prod_{i=1}^{m_d} \Gamma\left(w_i\right)}.
\end{equation}
\blue{To the matrices $U$ and $D$, we associate the characteristic vector $\bm{\Delta}$ defined as 
\begin{equation}\label{eq:def_Delta}
    \forall k\in\intinterv{1}{\max(\ell_u,\ell_d)},~\bm{\Delta}_k := \sum_{\ell=1}^{m_u} U_{\ell,k} - 
    \sum_{\ell=1}^{m_d} D_{\ell,k} 
\end{equation}
where by convention: $\sum_{\ell=1}^{m_u} U_{\ell,k}=0$
if $k>\ell_u$ and $\sum_{\ell=1}^{m_d} D_{\ell,k}=0$
if $k>\ell_d$.}
Finally, for a matrix $M$ with $n_M$ lines, and for a set of $k$ indices $\{i_m\}_{m\in\intinterv{1}{k}}\in\intinterv{1}{n_M}^{k}$,  we will denote $M_{\{i_1,...,i_k\}}$ the sub-matrix composed of the \blue{rows} 
$\{i_1,...,i_k\}$.

In order to clarify this formalism, we provide an explanatory illustration in example \ref{example:notations}. 
% \red{A definir après par exemple au début de la sous section 3.3 puisque c'est là qu'on commence à faire des séries}
\begin{example}
\label{example:notations}
    Consider the following 3-dimensional MB integral:
    \begin{equation}
        \int\limits_{\bm{c}+\mi \R^3}\frac{\Gamma(s_1)\Gamma(2+3s_2)\Gamma(1+s_1-s_2)}{\Gamma(1-2s_1)\Gamma(s_1+s_2+5s_3)}
        x_1^{-s_1}x_2^{-s_2}x_3^{-s_3}
        \frac{\D \bm{s}}{(\mi 2\pi)^3}
        ,
    \end{equation}
    where we assume that $\bm{c}\in\R^3$ and $(x_1,x_2,x_3)\in\R^3$ are chosen so that the integral exists (see \cite{zhdanov1998} for more details). In our framework, we have $n=3$, $m_u=3$ and $m_d=2$. Since we only use $(s_1,s_2)$ in the numerator and all the integration variables in the denominator, 
    we have $\ell_u=2$ and $\ell_d=3$. To get the correspondence with the definition, one can easily establish the following identifications:
    \begin{equation}
        \bm{q}
        =(x_1^{-1},x_2^{-1},x_3^{-1}),
        ~
        \bm{u}= (0,2,1),
        ~
        \bm{d}= (1,0),
        ~ 
        U=
        \begin{pmatrix}
            1 &0\\
            0&3\\
            1&-1
        \end{pmatrix},
        ~
        D=
        \begin{pmatrix}
            -2 &0&0\\
            1&1&5
        \end{pmatrix}.
    \end{equation}
\blue{
In this case, 
the characteristic vector 
is $\bm \Delta = (3,2,-5)^\top$.
}
\end{example}

\blue{Let us conclude this section by providing a key extension of the classical residue theorem in the context of $n$-dimensional MB integrals; this class of integrals has been extensively studied in \cite{passare1994multidimensional,zhdanov1998} and all proofs can be found in these works. We introduce the admissible subspace associated to the MB integral \eqref{eq:int-mb-def} as
\begin{equation}\label{eq:Pi_def}
    \Pi_{\bm \Delta} := 
    \{
    \bm{s} \in \mathbb{C}^n,
    \Re [ \langle \bm{\Delta} , \bm{s} \rangle ]
    <
    \Re [ \langle \bf{\Delta} , \bm{c} \rangle ]
    \}
\end{equation}
where $\langle \cdot,\cdot \rangle $ denotes the Euclidean inner product and where $\bm{\Delta}$ is the characteristic vector defined in \eqref{eq:def_Delta}. As shown in \cite{passare1994multidimensional,zhdanov1998}, if $\bm \Delta \neq 0$, the integrand of \eqref{eq:int-mb-def} decreases exponentially at infinity in $\Pi_{\bm \Delta}$, which allows to express this integral as a sum of the residues of its integrand, 
that we denote by $\bm \omega$, in this subspace:
\begin{equation}\label{eq:residue_thm}
    % \int\limits_{\bm{c}+\mi\R^n}
    % \GP{U\bm{s}_{\leqslant\ell_u}+\bm{u}}{D\bm{s}_{\leqslant\ell_d}+\bm{d}}
    % \prod_{i=1}^n q_i^{s_i}\frac{\D \bm{s}}{(\mi 2\pi)^n}
    % =
    % \sum\limits_{\Pi_\Delta}
    % \mathrm{Res}
    % \GP{U\bm{s}_{\leqslant\ell_u}+\bm{u}}{D\bm{s}_{\leqslant\ell_d}+\bm{d}}
    % \prod_{i=1}^n q_i^{s_i}
    \int\limits_{\bm c + \mi \R^n} \bm{\omega}
    =
    \sum_{\Pi_{\bm \Delta}} \mathrm{Res[\bm\omega]}.
\end{equation}
% where $\bm \omega$ is the integrand of \eqref{eq:int-mb-def}.  
In particular, if $n=1$ then the characteristic vector is simply a scalar $\Delta$; moreover if $\Delta>0$ (resp. $\Delta<0$) the admissible subspace \eqref{eq:Pi_def} is simply the half-plane located left (resp. right) to the vertical line $\Re(s) = c$.
% \eqref{eq:residue_thm} 

% +MODIFY IN THE PROOFS (PRECISE DELTA diff 0 ?
}

%%%%%%%%%%%%%%%%%%%%%%%%
\subsection{Mellin-Barnes representation of TS densities}
%%%%%%%%%%%%%%%%%%%%%%%%
We give here the MB representations for the densities of the single and double-sided TS distributions.
% as a first step to get series expansion of the density functions. 
We start with the single-sided case, for which as MB representation was already obtained in \cite{gupta2021densitiesinversetemperedstable} with a slightly different approach.

\begin{proposition}
\label{prop:density-Mellin-TS}
    % \color{black}
    Let $f_{\TS_+}$ be the density function of the $\TS_+(\alpha, \beta, \lambda)$-distribution, then the MB representation holds:
    \begin{equation}
        \label{eq:density-Mellin-TS}
        \forall x\in(0,+\infty),\quad f_{\TS_+}(x) = \frac{e^{-\lambda x + a \lambda^\beta}}{\beta}
        \int\limits_{c+\mi \R} 
         \frac{\Gamma\left(-\frac{s}{\beta}\right)}{\Gamma(-s)}a^{\frac{s}{\beta}}x^{-s-1}
        \frac{\D s}{\mi 2 \pi}
    \end{equation}
    where $c<0$ and $a:=-\alpha \Gamma(-\beta)$.
\end{proposition}

\begin{proof}
% \red{to do (Gaetano): raccourcir la preuve stp}
    % \color{black}
    % Noticing that the Laplace transform of $f_{\TS_+}$ is $\varphi_{\TS_+}(\mi \times \cdot)$ and 
    \blue{Given a random variable $X \sim \TS_+(\alpha,\beta,\gamma)$, its moment generating function is obtained by taking $ u=\mi p$ in \eqref{char_TS+}; using the Laplace inversion formula, we get:
    }
    % Using \eqref{char_TS+} and the Laplace inversion formula yields: 
    \begin{equation}
        \forall x \in(0,\infty),~f_{\TS_+}(x) = \int\limits_{c_1+\mi \R} e^{px+\alpha \Gamma(-\beta)\left((\lambda+p)^\beta-\lambda^\beta\right)} \frac{\D p}{\mi 2 \pi}
    \end{equation}
    with $c_1 > - \lambda$ and $a:=-\alpha \Gamma(-\beta)$ (note that $a>0$). 
    The change of variable \blue{$p \leftarrow p/x-\lambda$ gives:
    \begin{equation}
    \label{eq:ts-var-change}
        \forall x \in(0,\infty),~f_{\TS_+}(x) = 
        \frac{e^{a\lambda^\beta-\lambda x}}{x}\int\limits_{c_1+\mi \R} e^{p-ap^{\beta}x^{-\beta}} \frac{\D p}{\mi 2 \pi}
    \end{equation}
    with now $c_1>0$. Recall the inverse Mellin transform \cite[eq. 6.3.16 of Vol. 1]{Bateman}:
    % u>-lambda
    % u = u_tilde-lambda
    % u_tilde = u + lambda>0
    \begin{equation}
        \forall x> 0,~
        e^{-\rho x^{-h}}
        =
    \int\limits_{c_2+\mi\R} x^{-s} h^{-1}\rho^{-s/h}\Gamma\left(-\frac{s}{h}\right)\frac{\D s}{\mi 2 \pi}
    \end{equation}
    which holds for any $h,\Re(\rho)>0$ and $c_2<0$; choosing $h=\beta$ (positive by hypothesis) and $\rho = ap^\beta$ (whose real part is positive given $\Re (p)>0$ on the integration line and $\beta\in(0,1)$), we can rewrite
    \eqref{eq:ts-var-change} as:
    % \begin{equation}
    %     f_{\TS_+}(x) = \frac{\exp\left[-\lambda x + a \lambda^\beta\right]}{x}
    %     \int_{c_2+\mi \R} \exp\left[u\right]\exp\left[-au^\beta x^{-\beta}\right] \frac{\D u}{\mi 2 \pi}.
    % \end{equation}
    % where $\Real(c_2) > 0$. Applying \cite[6.3.16 of Vol 1]{Bateman}, we have:
    \begin{equation}
        f_{\TS_+}(x) = \frac{e^{-\lambda x + a \lambda^\beta}}{\beta x}
        \int\limits_{c_1+\mi \R} e^p
        \int\limits_{c_2+\mi \R} \Gamma\left(-\frac{s}{\beta}\right)a^{\frac{s}{\beta}}p^{s}x^{-s}
        \frac{\D s}{\mi 2 \pi}\frac{\D p}{\mi 2 \pi}
    \end{equation}
    where $c_1>0$ and $c_2<0$.}
    % Rearranging the terms, we have:
    % \begin{equation}
    %     f_{\TS}(x) = \frac{\exp\left[-\lambda x + a \lambda^\beta\right]}{\beta x}
    %     \int_{c_3+\mi \R} 
    %     \left(\int_{c_2+\mi \R}\exp\left[u\right]u^{s}\frac{\D u}{\mi 2 \pi}\right)
    %      \Gamma\left(-\frac{s}{\beta}\right)a^{s/\beta}x^{-s}
    %     \frac{\D s}{\mi 2 \pi}.
    % \end{equation}
    Changing the order of integration and noting that $\Real(s) = c_2 < 0 $, the arising $u$-integral is the Laplace inversion integral of $u^{s}$ evaluated in 1, which equals $1/\Gamma(-s)$ since \blue{(see \cite[eq. 4.3.1 of Vol. 1]{Bateman}):
    \begin{equation}
        % \forall x\in\R,~e^{-\alpha x} = \int_0^\infty x^{s-1}\alpha^{-s}\Gamma(s)
        \forall t>0,~
        \frac{1}{\mi 2\pi}\int_{\gamma+\mi\R}e^{pt}(\Gamma(\nu+1)u^{-\nu-1})\D p
        =
        t^{\nu}
    \end{equation}
    for any $\Re(\gamma)>0$ and $\Re(\nu)>-1$. The choice $\gamma = c_1$, $t=1$ and $\nu = -s-1$ satisfies these conditions. Last, renaming $c_2 :=c$ yields \eqref{eq:density-Mellin-TS}}
    % \begin{equation}
    %     f_{\TS}(x) = \frac{\exp\left[-\lambda x + a \lambda^\beta\right]}{\beta x}
    %     \int_{c_3+\mi \R} 
    %      \frac{\Gamma\left(-\frac{s}{\beta}\right)}{\Gamma(-s)}a^{s/\beta}x^{-s}
    %     \frac{\D s}{\mi 2 \pi}.
    % \end{equation}
\end{proof}

\blue{From now on, inequalities involving matrices and vectors have to be understood componentwise; in addition,} we associate to a $\TS(\alpha_+,\beta_+,\lambda_+,\alpha_-,\beta_-,\lambda_-)$ distribution the quantities $a_\pm$, $\ubar{\lambda}$ and $\gamma$ defined as:
\begin{equation}
\label{eq:constants-ts}
    \begin{cases}
        a_+ := -\alpha_+ \Gamma(-\beta_+),\\
        a_- := -\alpha_- \Gamma(-\beta_-),\\
        \ubar{\lambda} := \lambda_+ + \lambda_-,\\
        \gamma := a_+\lambda_+^{\beta_+}+a_-\lambda_-^{\beta_-}.
    \end{cases}
\end{equation}
% \blue{In the rest of paper, in}
\begin{proposition}
\label{prop:density-Mellin-GTS}
Let $f_\TS$ be the density function of the $\TS(\alpha_+,\beta_+,\lambda_+,\alpha_-,\beta_-,\lambda_-)$ distribution, then the following MB representation holds:
\begin{multline}
\label{eq:density-Mellin-GTS}
% \color{blue}
    \forall x \in (0,+\infty),\quad 
    f_\TS(x) = \frac{e^{\gamma-\lambda_+ x}}{\beta_+\beta_-}   
    \int\limits_{\bm{c}+\mi \R^3 } 
    \GP{U_\TS\bs + \bu_\TS}{D_\TS\bs_{\leqslant 2} + \bd_\TS}
    \\
% \color{blue}
    \times  \ubar{\lambda}^{\frac{s_1+s_2}{2}-s_3}x^{-1-\frac{s_1+s_2}{2}-s_3}a_+^{\frac{s_1}{\beta_+}}a_-^{\frac{s_2}{\beta_-}}
    \frac{\D \bs}{(\mi 2 \pi)^3}
\end{multline}
where:
\begin{equation}
    U_\TS
    :=
    \begin{pmatrix}
         -1/2 & -1/2 & 1\\
         1/2 &  1/2 & 1 \\
         1/2  & -1/2 & -1 \\
         -\beta_+^{-1} &0 &0\\
         0 &-\beta_-^{-1} &0\\
    \end{pmatrix}
    ,
    \quad 
    \bm{u}_\TS
    :=
    \begin{pmatrix}
        0\\
        1\\
        0\\
        0\\
        0
    \end{pmatrix}
    ,\quad
    D_\TS
    :=
    \begin{pmatrix}
         1 & 0\\
         -1&  0 \\
         0 &  -1 \\
    \end{pmatrix}
    ~\mathrm{and}~
    \bm{d}_\TS
    :=
    \begin{pmatrix}
        1\\
        0\\
        0
    \end{pmatrix}
\end{equation}
and  $\bm{c}\in\{\bm{x}\in\R^3|~U_\TS x+\bm{u}_\TS> 0\}$.
\end{proposition}

\begin{proof}
    From \eqref{TSdensity_conv} and proposition \ref{prop:density-Mellin-TS}, we have for $x\in(0,+\infty)$:
    % \begin{multline}
    %     f_\TS(x)= \frac{e^{\gamma}}{\beta_+\beta_-}
    %     \int_0^\infty\int_{\ubar{c}+\mi \R }
    %     \left(
    %     e^{-\lambda_+(x+y) - \lambda_- y} \frac{\Gamma(-s_1/\beta_+)\Gamma(-s_2/\beta_-)a_+^{s_1/\beta_+}a_-^{s_2/\beta_-}}{\Gamma(-s_1)\Gamma(-s_2)}\right.\\
    %     \times \left.(x+y)^{-s_1-1}y^{-s_2-1}\right)\frac{\D s_1}{\mi 2 \pi}\frac{\D s_2}{\mi 2 \pi} \D y
    % \end{multline}
    \begin{multline}
        f_\TS(x)= \frac{e^{\gamma}}{\beta_+\beta_-}
        \int_0^\infty\int\limits_{\bm{c}+\mi \R^2 }
        \left(
        e^{-\lambda_+(x+y) - \lambda_- y} 
        \GP{U_1\bs}{ D_1\bs}\right.\\
        \times \left.a_+^{\frac{s_1}{\beta_+}}a_-^{\frac{s_2}{\beta_-}}(x+y)^{-s_1-1}y^{-s_2-1}\right)\frac{\D \bs}{(\mi 2 \pi)^2} \D y
    \end{multline}
    where $U_1 = -\mathrm{diag}((\beta_+^{-1},\beta_-^{-1}))$, 
    $D_1 = -\mathrm{Id}_2$ and $\bm{c}\in (-\infty, 0)^2$.
    Changing the order of integration
    % by Fubini's theorem
    yields:
    \begin{equation}
         \label{eq:whit-intermediaire-1}
         f_\TS(x)
        =\frac{e^{\gamma-\lambda_+ x}}{\beta_+\beta_-}
        \int\limits_{\bm{c}+\mi \R^2 }
        \left(
        \int_0^\infty
        e^{-\ubar{\lambda}y} 
        (x+y)^{-s_1-1}y^{-s_2-1}\D y\right)
        % g\left(\bm{s}\right)
        \GP{U_1\bs}{D_1\bs}a_+^{\frac{s_1}{\beta_+}}a_-^{\frac{s_2}{\beta_-}}
        \frac{\D \bs}{(\mi 2 \pi)^2}.
    \end{equation}
    % where $g\left(\bs\right) 
    %         :=
    %         \GP{U_1\bs}{D_1\bs}a_+^{\frac{s_1}{\beta_+}}a_-^{\frac{s_2}{\beta_-}}.$ 
    Since $\Real\left(-(1+s_1+s_2)/2-(s_2-s_1)/2\right) > -1/2$ on the integration subset and $x>0$, 
    % the identity \cite[eq 9.222.2]{gradshteyn2014table} holds and 
    the integral over $y$ 
    can be written in terms of the Whittaker function \blue{(\cite[eq 9.222.2]{gradshteyn2014table})}:
    \begin{multline}
         \label{eq:whit-intermediaire-2}
        \int_0^\infty
        e^{-\ubar{\lambda}y} 
        (x+y)^{-s_1-1}y^{-s_2-1}\D y\\
        = 
        \Gamma(-s_2)\ubar{\lambda}^{\frac{s_1+s_2}{2}}x^{-\frac{s_1+s_2}{2}-1}
        e^{\frac{\ubar{\lambda}x}{2}}W_{\left((s_2-s_1)/2, -(1+s_1+s_2)/2\right)}\left(\ubar{\lambda} x\right) 
        .
    \end{multline}
    % where  (existence condition on Whittaker function).
% We have then:
% \begin{multline}
%     \label{eq:ts-mb}
%         f_\TS(x) = \frac{e^{\gamma-\lambda_+ x}}{\beta_+\beta_-}
%         \int_{\bm{c}+\mi \R^2}
%         \left(
%         \ubar{\lambda}^{(s_1+s_2)/2}x^{-(s_1+s_2)/2-1}\Gamma(-s_2)
%          e^{\ubar{\lambda}x/2}\right.\\
%         \left.\times
%         W_{\left((s_2-s_1)/2, -(1+s_1+s_2)/2\right)}(\ubar{\lambda x})
%         g\left(\bs\right)\right)
%         \frac{\D \bs}{(\mi 2 \pi)^2}.
%     \end{multline}
\blue{
The MB representation for the Whittaker function with parameters $v$ and $w$ such that $1/2\pm v-w\notin -\N$ can be written as (\cite[eq. 13.16.11]{DLMF}):
\begin{equation}
\begin{multlined}
    \forall z\in\C ~\mathrm{s.t.}~\mathrm{\lvert ph(z)\rvert<\frac{3\pi}{2}},\\
    W_{v,w}(z)
    =
    e^{-z/2}
    \int\limits_{c_3+\mi \R}
    \frac{\Gamma(1/2+w+s)\Gamma(1/2-w+s)\Gamma(-v-s)}{\Gamma(1/2+w-v)\Gamma(1/2-w-v)}z^{-t}\frac{\D s}{\mi 2 \pi}
\end{multlined}
\end{equation}
where $c_3$ is chosen so that, on the integration line, all the arguments of the gamma functions in the numerator are positive. 
Hence,
}
% yields:
\begin{equation}
    \label{eq:whit-mb}
         e^{\ubar{\lambda}x/2}
        W_{\left((s_2-s_1)/2, -(1+s_1+s_2)/2\right)} (\ubar{\lambda}x)=
        \int\limits_{c_3 + \mi \R}
        \GP{U_2\bs +\bu_2}{ D_2\bs_{\leqslant 2} +\bd_2}
        (\ubar{\lambda}x)^{-s_3}\frac{\D s_3}{\mi 2 \pi}.
\end{equation}
where now $\bs:=(s_1,s_2,s_3)$, and
\begin{equation}
        U_2 :=
        \begin{pmatrix}
            -1/2 & -1/2 & 1\\
            1/2 &   1/2 & 1\\
             1/2 & -1/2 & -1 
        \end{pmatrix},
        \quad
        \bm{u}_2:=
        \begin{pmatrix}
            0\\
            1\\
            0
        \end{pmatrix},
        \quad
        D_2 := 
        \begin{pmatrix}
            1 &0 \\
            0 &-1
        \end{pmatrix}
        ~\mathrm{and}~
        \bm{d}_2
        :=
        \begin{pmatrix}
            1\\
            0
        \end{pmatrix}.
    \end{equation}
% The density is then:
% \begin{multline}
%     f_\TS(x) = \frac{e^{a_+\lambda_+^{\beta_+}+a_-\lambda_-^{\beta_-}-\lambda_+ x}}{\beta_+\beta_-} \times \\
%     \int_{\ubar{c}+\mi \R } \frac{\Gamma\left(-\frac{\bs_1+\bs_2}{2}+s_3\right)\Gamma\left(1+\frac{\bs_1+\bs_2}{2}+s_3\right)\Gamma\left(\frac{\bs_1-\bs_2}{2}-s_3\right)\Gamma(-\bs_1/\beta_+)\Gamma(-\bs_2/\beta_-)}{\Gamma\left(1+\bs_1\right)\Gamma\left(-\bs_1\right)\Gamma\left(-\bs_2\right)}\\
%     \times  \ubar{\lambda}^{(\bs_1+\bs_2)/2-s3}x^{-1-\frac{\bs_1+\bs_2}{2}-s_3}a_+^{\bs_1/\beta_+}a_-^{\bs_2/\beta_-}\frac{\D \bs_1}{\mi 2 \pi}\frac{\D \bs_2}{\mi 2 \pi}\frac{\D s_3}{\mi 2 \pi}\\
% \end{multline}
and
$\bm{c}\in\{\bm{x}\in\R^3|~U_2 \bm{x}+\bm{u}_2> 0\}$. Replacing successively \eqref{eq:whit-mb} in \eqref{eq:whit-intermediaire-2} and then in  \eqref{eq:whit-intermediaire-1} leads to the desired result.
\end{proof}

\green{
\begin{remark} We note that MB representation for the convolution of stable densities with different index has been obtained in \cite{mainardi2008mellin}; the MB representation we get in proposition \ref{prop:density-Mellin-GTS} can be seen as its generalization to the TS case. 
\end{remark}
}

%%%%%%%%%%%%%%%%%%%%%%%%%%%%
\subsection{Series expansion of TS densities}
%%%%%%%%%%%%%%%%%%%%%%%%%%%%

Thanks to the MB representations that we have established in propositions \ref{prop:density-Mellin-TS} and \ref{prop:density-Mellin-GTS}, 
we will now be able to derive series representations for the single and double-sided TS distributions using residue calculus. We start with the one-sided case, \blue{for which the result we will get in  proposition \ref{prop:series_density_TS} has already been obtained in \cite[p.7]{barabesi2016new} based on representations from \cite[p.88] {ken1999levy}; it was also used in subsequent works such as \cite[eq. 2.15]{torricelli2022tempered}
% for which we are able to recover a more compact version of the expression already obtained in and in 
and \cite[prop. 4.1]{gupta2021densitiesinversetemperedstable}. }%We give a proof of this already known result to show the reasoning that will be used in the rest of the article to obtain more complex series expansions. }

\begin{proposition}\label{prop:series_density_TS}
    Let $f_{\TS_+}$ be the density function of the $\TS_+(\alpha, \beta, \lambda)$ distribution, we have the expansion: 
    \begin{equation}\label{eq:f_TS_series}
        \forall x\in(0,+\infty),\quad f_{\TS_+}(x) = e^{-\lambda x + a \lambda^\beta} \sum_{k=1}^\infty \frac{(-a)^k x^{-k\beta-1}}{ k!\Gamma(-k\beta)}
    \end{equation}
    where $a$ is defined in proposition \ref{prop:density-Mellin-TS}.
\end{proposition}

\begin{proof}
    % Using the formalism of \cite{passare1994multidimensional,zhdanov1998}, we can write 
    The characteristic vector \eqref{eq:def_Delta} associated to the MB integral \eqref{eq:density-Mellin-TS} is $\Delta = 1-\beta^{-1}$. 
    % Following the Jordan lemma for one-dimensional MB integrals (see \cite{passare1994multidimensional}), 
    % \blue{Following \eqref{eq:residue_thm}, we can express} the integral \eqref{eq:density-Mellin-TS} as the sum of the residues associated to the singularities of the integrand located in the half-plane $\{ s\in \mathbb{C},~ \Re[\Delta s] < \Delta c\}$. 
    Since $\beta\in(0,1)$, we have $\Delta < 0$ and therefore it follows from the residue theorem \eqref{eq:residue_thm} that the integral \eqref{eq:density-Mellin-TS} is equal to the sum of the residues associated to the singularities of the integrand located in the right half-plane. They are inducted by the $\Gamma(-s/\beta)$ function, are located at $\{k\beta|~ k\in\mathbb{N}\}$ and are equal to $\beta(-1)^k/k!$. Summing all residues completes the proof.
    % only the poles located on the positive semi-axis contribute to the series expansion. More precisely, the expansion is generated by the poles located in $s=n\beta$ and the result immediately follows.
    % Using residue calculus for Mellin-Barnes integrals, the integral in \eqref{eq:density-Mellin-TS} can be expressed as a series generated by the poles located in $s=n\beta$ for $n\in\N$. 
    % Remembering that $\mathrm{Res}(\Gamma,-n) = (-1)^n/\Gamma(-n)$, the result immediately follows.  
    % Using the residue theorem, the proof is straightforward. 
    % Note that the same formula has already been found in REF with a slightly different approach.
    % \red{A justifier (Tsikh?)}
\end{proof}

\blue{
\begin{remark}
We can also express \eqref{eq:f_TS_series} in terms of the Mainardi-Wright function $\mathrm{F}_\beta$ %(or \textit{auxiliary functions of the Wright function}) 
as introduced in \cite[eq. 2.8]{mainardi1994initial}:
\begin{equation}
    \forall x\in(0,+\infty),
    \quad
    f_{\TS_+}(x) = 
    e^{-\lambda x + a\lambda^\beta}\mathrm{F}_\beta(ax^{-\beta}).
\end{equation}
\green{
This is a simple consequence of proposition \ref{prop:series_density_TS} and definition \cite[eq.9]{mainardi2020wright}.
Note that representation for stable densities has been studied in \cite[eq.38]{mainardi2020wright}.}
\end{remark}}
Let us now extend proposition \ref{prop:series_density_TS} to the more general double-sided TS case. 
To that extent, and whenever double-sided TS distributions are considered throughout the rest of the paper, \blue{we will impose a restriction on the $\beta_\pm$ parameters.
% we will impose the additional condition $(\beta_+,\beta_-)\in ((0,1)\cap(\R\backslash\Q))^2$. 
\begin{assum}
    \label{assum:beta-irra}
    $\beta_+,\beta_-\in (0,1)\cap\Q^c$.
\end{assum}
All subsequent propositions will be stated under assumption \ref{assum:beta-irra}}. The irrationality of the $\beta_\pm$ parameters is not mandatory in theory, but it allows to avoid multiple poles in the arising MB integrals, and thus allows for simpler expressions for the residue series. Furthermore, we note that, with the pointwise convergence established in proposition \ref{prop:pointwise-convergence}, this condition appears to be absolutely not restrictive thanks to the density of $\R\backslash\Q$ in $\R$.

\begin{proposition}\label{prop:series_density_GTS}
    Let $f_\TS$ be the density of the  $\TS(\alpha_+,\beta_+,\lambda_+,\alpha_-,\beta_-,\lambda_-)$ distribution, we have the expansion: 
    \begin{equation}
    \label{eq:serie-gts-density}
    \forall x\in(0,+\infty),\quad
        f_\TS(x) = 
        e^{\gamma-\lambda_+ x} 
        \sum_{\bm{n}\in\N^3}
        \frac{(-1)^{n_{1}+n_2+n_3}a_+^{n_2}a_-^{n_3}}{n_1!n_2!n_3!}
        % \times
        \left(
        d_{\bm{n}}^{(1)}(x) +d_{\bm{n}}^{(2)}(x)\right)
    \end{equation}
    where the coefficients $d_{\bm{n}}^{(1)}(x) $ and $d_{\bm{n}}^{(2)}(x)$ are defined as:
    \begin{equation}
    \begin{cases}
    \displaystyle
        d_{\bm{n}}^{(1)}(x)  := \frac{\Gamma(1-n_1+\beta_+n_2+\beta_- n_3)\Gamma(n_1-\beta_-n_3)}
        {\Gamma(1+\beta_+n_2)\Gamma(-\beta_+n_2)\Gamma(-\beta_-n_3)} 
        \ubar{\lambda}^{n_1} x^{-1+n_1-\beta_+ n_2-\beta_- n_3},\\
    \displaystyle
        d_{\bm{n}}^{(2)}(x) :=
        \frac{\Gamma(-1-n_1-\beta_+n_2-\beta_-n_3)\Gamma(1+n_1+\beta_+n_2)}
        {\Gamma(-\beta_+n_2)\Gamma(1+\beta_+n_2)\Gamma(-\beta_-n_3)} 
        \ubar{\lambda}^{1+n_1+\beta_+n_2+\beta_- n_3} x^{n_1}.
    \end{cases}
    \end{equation}
\end{proposition}

\begin{proof}
    % Using the formalism of \cite{passare1994multidimensional, zhdanov1998} in the $\mathbb{C}^3$ case, we denote 
    The characteristic vector $\bm{\Delta}$ associated to the 
    MB integral \eqref{eq:density-Mellin-GTS} is $\bm{\Delta}=(1/2-\beta_+^{-1}, 1/2-\beta_-^{-1},1)$.
    % The multidimensional Jordan residue lemma \cite{passare1994multidimensional} 
    % allows to express 
    \blue{Following \eqref{eq:residue_thm}, we can express} the integral
    \eqref{eq:density-Mellin-GTS} as a sum of residues associated to the poles of the integrand located in the admissible subspace $\Pi_{\bm{\Delta}}$ as defined in \eqref{eq:Pi_def}.
    % and $\bm{c}$ belonging to the polyhedron generated by the positivity constraint given in proposition \ref{prop:density-Mellin-GTS}. 
    In this subspace, there are two sets of poles, that are given by the solutions of the two linear systems
    % that are the solution There are two sets of admissible poles that generate the series expansion.
    % These sets of poles are solutions, for $\bm{n\in\N^3}$, of the two systems:
    \begin{equation}
    % \begin{enumerate}
    \begin{cases}
        (U_\TS \bm{s} + \bm{u}_\TS)_{\{1,4,5\}} = -\bm{n},\\
        (U_\TS \bm{s} + \bm{u}_\TS)_{\{2,4,5\}} = -\bm{n},
    \end{cases}
    % \end{enumerate}
    \end{equation}
    where $U_\TS$ and $\bm{u}_\TS$ are defined in \ref{prop:density-Mellin-GTS} and where $\bm{n\in\N^3}$.
    These systems have unique solutions, respectively given by:
    \begin{equation}
        \begin{cases}
            \bm{s}=(\beta_+n_2,\beta_-n_3,-n_1+\beta_+n_2+\beta_-n_3),\\
            \bm{s}=(\beta_+n_2,\beta_-n_3,-1-n_1-\beta_+n_2-\beta_-n_3).
        \end{cases}
    \end{equation}
    \blue{Recalling that for $\ell\in\N$, the residue of $\Gamma(z)$ in $z=-\ell$ is equal to $(-1)^\ell/\ell!$), 
    and summing the residues at the above locations 
    yields the series \eqref{eq:serie-gts-density}.}
    % {\color{blue}
    % \red{TO DO, using that the pole 0,3,4 and 1,3,4.}
    % \red{pas de problème pcq vaut 0}
    % We use two sets of poles. The coefficients $d_{\bm{n}}^{(1)}$ come from the poles of the gamma function corresponding to the first, fourth and fifth components of the vector $U\bm{s}+\bm{u}$ set to zero. In other words, the poles are found by solving the linear system:
    % \begin{equation}
    %     \begin{cases}
    %         -s_1/2-s_2/2+s_3 = -n_1,\\
    %         -\beta_+^{-1}s_1 = -n_2,\\
    %         -\beta_-^{-1}s_2 = -n_3.
    %     \end{cases}
    % \end{equation}
    % The poles are then located in $\bm{s} = (\beta_+n_2,\beta_-n_3,-n_1+\beta_+n_2/2+\beta_-n_3/2)$.
    % }
    % The coefficients $c_{\bm{n}}^{(2)}$ are obtained by solving the linear system given by equating the second, fourth and fifth components of the vector $U\bm{s}+\bm{u}$  to 0.
    % \red{TODO: checker tsikh, vecteur delta etc...}
\end{proof}

In figure \ref{fig:ts-density-series}, we compare the series expansion from \ref{prop:series_density_GTS} (truncated at $\bm{n}=(60,60,60)$  with the numerical Fourier inversion of the TS characteristic function \eqref{eq:char_TS} with $t=1$; we can observe that \blue{the two approaches agree}.

% An example comparing a tempered distribution via Fourier inversion and with our series expansion is shown in Figure \ref{fig:ts-density-series}.

\begin{figure}
    \centering
    \includegraphics[width=0.7\linewidth]{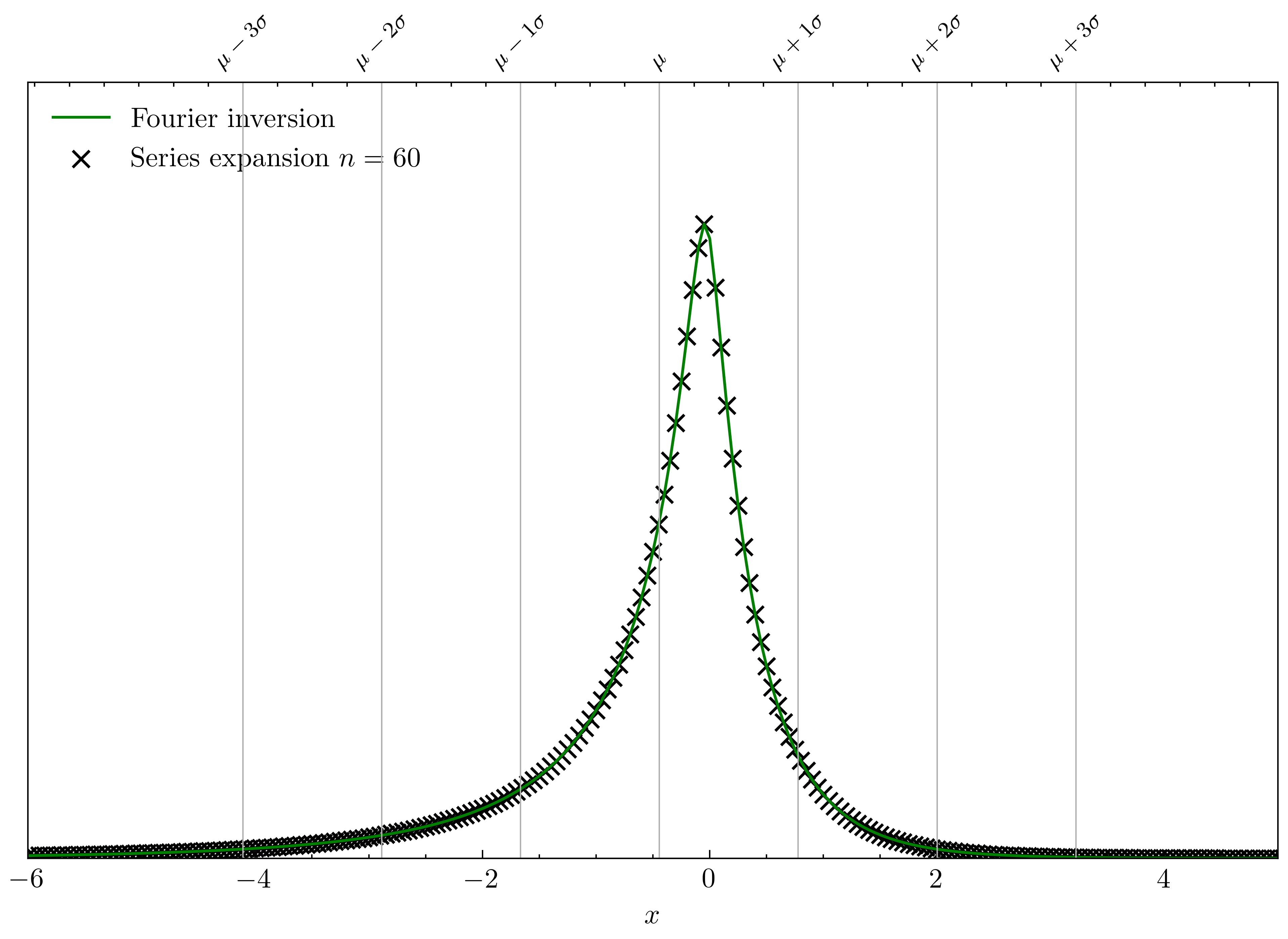}
    \caption{ Series expansion \eqref{eq:serie-gts-density} of the double-sided TS density truncated at  $\bm{n}=(60,60,60)$ \textit{vs.} numerical Fourier inversion of the characteristic function \eqref{eq:char_TS} for $t=1$. TS parameters from \eqref{eq:numerical-params}.}
    \label{fig:ts-density-series}
\end{figure}

% \begin{remark}
% Proposition \ref{prop:series_density_GTS} is valid for $x>0$, however it can easily be extended to the negative semi-axis thanks to the symmetry relation
% \begin{equation}
%     \label{eq:dens-symetry}
%     \forall x\in (-\infty, 0),\quad f_{\TS(\alpha_+,\beta_+,\lambda_+,\alpha_-,\beta_-,\lambda_-)}(x) = f_{\TS(\alpha_-,\beta_-,\lambda_-,\alpha_+,\beta_+,\lambda_+)}(-x)
% \end{equation}
% that follows immediately from the characteristic function \eqref{eq:char_TS}.
% \end{remark}
\blue{
\begin{cor}
    % Proposition \ref{prop:series_density_GTS} can be extended to the negative semi-axis thanks to the symmetry relation:
    For any $x<0$, we have:
    \begin{equation}
        \label{eq:dens-symetry}
        % \forall x\in (-\infty, 0),\quad 
        f_{\TS}(x;\alpha_+,\beta_+,\lambda_+,\alpha_-,\beta_-,\lambda_-)
        =
        f_{\TS}(-x;\alpha_-,\beta_-,\lambda_-,\alpha_+,\beta_+,\lambda_+,)
    \end{equation}
    where the right hand-side is given by the series \eqref{eq:serie-gts-density}.
    % that immediately follows from the characteristic function \eqref{eq:char_TS}
\end{cor}
\begin{proof}
    Immediate from the characteristic function \eqref{eq:char_TS}.
\end{proof}
}

\section{Option pricing under TS processes}
\label{sec:opt-pricing-ts-processes}

%%%%%%%%%%%%%%%%%%%%%%%%%%%%%%%%%%
%%%%%%%%%%%%%%%%%%%%%%%%%%%%%%%%%%
% \red{noter les valeurs d'option $C(T,K)$ au lieu de $C$ ?, explication irrationnels}

TS distributions have already been 
widely used in financial modeling context (see for example \cite{poirot2006monte,KuchlerTS,bianchi, xia2024pricing}). The availability of the characteristic function in closed form allows to use Fourier pricing techniques such as Lewis \cite{Lewis01} or Carr-Madan \cite{Carr99}, as well as  Hilbert \cite{phelan2019hilbert}, COS \cite{fang2009novel} and PROJ \cite{kirkby2015efficient}, the latter being known to be particularly fast and efficient. All of these methods feature one (or several) hyperparameters to be optimally selected before any pricing procedure, a drawback that is of course undesirable in practice. In this section, we will therefore focus on providing
\blue{analytical pricing formulas} for digital and European options based on residue series expansion; to that end, we will first derive MB representations for the option prices and then compute them using residue calculus. 
% In particular, we will see that the case of finally show that our method can be applied to bilateral Gamma and one-sided tempered distributions to obtain very compact pricing formula.

%%%%%%%%%%%
\subsection{Option pricing framework}
%%%%%%%%%%%%
Consider a filtered probability space 
$(\Omega, \mathcal{F},\left(\mathcal{F}_t\right)_{t\geqslant 0}, \mathbb{P})$ 
where $\left(\mathcal{F}_t\right)_{t\geqslant 0}$ 
denotes the natural filtration generated by the TS process $ \left(X_t\right)_{t\geqslant 0}$. We assume that the price at time $t\geqslant 0$ of some financial asset can be written as:
\begin{equation}
\label{St_dynamics}
    % \forall t\geqslant 0,\quad 
    S_t := S_0e^{(r-q+\zeta)t + X_t}
    ,
\end{equation}
where $S_0>0$, $r$ and $q$ are respectively the risk-free rate and the dividend yield (both assumed to be continuously compounded) and $\zeta$ is the martingale adjustment (or compensator) chosen in such a way that the discounted process $(e^{-(r-q)t}S_t)_{t\geqslant 0}$ is a martingale.
\blue{
\begin{assum}
    \label{assum:lambda}
    $\lambda_+>1$.
\end{assum}
}
\blue{Assumption \ref{assum:lambda} is necessary because we know that,} for any real valued L\'evy process $Y$ \blue{started at zero}, $e^Y$ is a martingale if and only if $\mathbb E \left[e^{Y_1}\right]=1$ (see \cite[Prop. 3.17]{Cont04}); it is easy to see from \eqref{eq:char_TS} that the discounted price process is therefore a martingale if and only if $\lambda_+ > 1$. \blue{ Furthermore, it also follows from \eqref{eq:char_TS} that:}
\begin{equation}
\label{eq:def-zeta}
    \zeta = 
    -\ln \varphi_\TS (-\mi)
    =
    a_+ \left( (\lambda_+-1)^{\beta_+} - \lambda_+^{\beta_+} \right) 
    +
    a_- \left( (\lambda_-+1)^{\beta_-} - \lambda_-^{\beta_-} \right)
    .
\end{equation}
Last, we recall that at time $t\geqslant 0$, the price $C_t$ of a contingent claim delivering a (integrable) payoff $\mathcal{P}(S_T)$ at maturity $T\geqslant t$ is given by the conditional expectation
\begin{equation}\label{pricing}
    C_t
    =
    \mathbb E \left[ e^{-r(T-t)} \mathcal{P} (S_T)|\mathcal{F}_t \right].
\end{equation}
% where $(\mathcal{F}_t)_{t\geqslant 0 }$ is the natural filtration of the price process $(S_t)_{t\geqslant 0 }$.
% given by $\zeta := -\ln \varphi_\TS (-\mi)$. It is straightforward that the discounted process $(e^{-(r-q)t}S_t)_{t\geqslant 0}$ is a martingale and can be used for pricing purposes.
% As proved in \cite{KuchlerTS}, it is necessary to have $\lambda_+>1$.

In the following, we will assume without loss of generality that we conduct pricing at $t=0$, and we will focus on the two most important path-independent claims, namely digital options (that can be either cash-or-nothing (CN) or asset-or-nothing (AN) options) and European (EUR) options. 
% For a sake of clarity, we now set $t=0$, the maturity $T\geqslant 0$ and the strike $K\geqslant 0$. 
The payoffs of the digital options of maturity $T$ and strike price $K$ are given by
\begin{equation}
\begin{cases}
    \mathcal{P}_{\CN}(S_T) = \mathbbm{1}_{\{S_T>K\}}
    \\
    \mathcal{P}_{\AN}(S_T) = S_T \mathbbm{1}_{\{S_T>K\}},
    ,
    % \mathcal{P}_{\EUR}(S_T) := (S_T-K)\mathbbm{1}_{S_T>K}.
\end{cases}
\end{equation}
and the payoff of a European \blue{call} option of maturity $T$ and strike $K$ is: 
\begin{equation}\label{eq:payoff_Eureopean}
    \mathcal{P}_{\EUR}(S_T)
    =
    \mathcal{P}_{\AN}(S_T) - K \mathcal{P}_{\CN}(S_T)
    =
    (S_T-K)\mathbbm{1}_{\{S_T>K\}}.
\end{equation}
% of maturity $T$ and strike $K$ call options, whose terminal payoffs are given by:

% \begin{equation}
% \begin{cases}
%     \mathcal{P}_{\AN}(S_T) := S_T \mathbbm{1}_{S_T>K},\\
%     \mathcal{P}_{\CN}(S_T) := \mathbbm{1}_{S_T>K},\\
%     \mathcal{P}_{\EUR}(S_T) := (S_T-K)\mathbbm{1}_{S_T>K}.
% \end{cases}
% \end{equation}
For the two digital options, it is easy to see that the prices \eqref{pricing} can be carried out via the following integrals:
\begin{equation}
    \label{eq:digital-option-def}
    \begin{cases}
        \displaystyle C_\CN = e^{-rT}\int_{-k}^\infty f_{X_T}(x) \D x,\\
        \displaystyle C_\AN = 
        Ke^{k-rT}\int_{-k}^\infty e^xf_{X_T}(x) \D x,
        %S_0e^{(\zeta-q)T}\int_{-k}^\infty e^xf_{X_T}(x) \D x,
    \end{cases}
\end{equation}
where $f_{X_T}$ is the density function of $X_T$ and $k:= \ln(S_0/K) + (r-q+\zeta)T$ is the log-forward moneyness. Our strategy is the following: we first derive MB integral representations for the two digital calls with $k<0$, then we express these integral as infinite series of residues. We extend the formulas for $k>0$ thanks to the symmetry relations of the TS density \eqref{eq:dens-symetry}, and, last, we get the European call price by difference using \eqref{eq:payoff_Eureopean}.
% One is finally able to get the put prices via the call-put parity.
% In the following, for a $\TS(\alpha_+,\beta_+,\lambda_+,\alpha_-,\beta_-,\lambda_-)$-dsitribution and $T>0$, we define the constant $a_\pm(T)$ by $a_\pm(T):=-\alpha_\pm T\Gamma(-\beta_\pm)$.

\blue{
\begin{remark}
    We note that our definition of the log-forward moneyness $k$ includes the martingale adjustment $\zeta$, which is because the L\'evy process driving the asset prices in \eqref{St_dynamics} is not itself compensated. This definition of $k$ is common in the exponential L\'evy pricing literature (see \cite{kirkby2015efficient,huang2004specification} for instance), but should not be mistaken for the model-independent definition of the moneyness $k_0 = \ln S_0/K + (r-q)T$. In the rest of the paper, all references to the (log-forward) moneyness and to related market situations (out-of-the money, at-the-money and in-the-money prices) will be with respect to $k$ and not to $k_0$.
\end{remark}
}

%%%%%%%%%%%
\subsection{Digital options}
%%%%%%%%%%%%

As mentioned before, we give the MB representations of digital options and derive their associated series expansions. 
%%%%%%%%%%%
\subsubsection{Mellin-Barnes integral representations}
 We start by deriving the MB representation for the CN option in proposition \ref{prop:CN-Mellin-GTS}; extension to the AN case will be handled in the subsequent proposition \ref{prop:AN-Mellin-GTS}. The matrices and vectors $U_{\TS}$, $u_{\TS}$, $D_{\TS}$, and $d_{\TS}$ are the ones defined earlier in proposition. \ref{prop:density-Mellin-GTS}
% The cash-or-nothing integral representation is first handled in proposition \ref{prop:CN-Mellin-GTS} and its immediate extension to asset-or-nothing option is given in proposition \ref{prop:AN-Mellin-GTS}.
% }
\begin{proposition}
    \label{prop:CN-Mellin-GTS}
    The value of a CN call with moneyness $k<0$ admits the MB representation:
    \begin{multline}
    \label{eq:CN-Mellin-GTS}
    C_\CN = 
    \frac{e^{(\gamma-r)
    T}}{\beta_+\beta_-}  
    \int\limits_{\bm{c}+\mi \R^4 }
    %\\
     \frac{\GP{U\bs + \bu}{ D_\TS\bs_{\leqslant 2}+ \bd_\TS}}{\frac{s_1+s_2}{2}+s_3+s_4}\\
     % \left(\frac{s_1+s_2}{2}+s_3+s_4\right)
    \times  
    \ubar{\lambda}^{(s_1+s_2)/2-s_3}
    \lambda_+^{-s_4}
    (a_+T)^{s_1/\beta_+}
    (a_-T)^{s_2/\beta_-}(-k)^{-\frac{s_1+s_2}{2}-s_3-s_4}
    \frac{\D \bs}{(\mi 2 \pi)^{4}}
\end{multline}
where:
\begin{equation}
        U :=
        % \begin{pmatrix}
        %  -1/2 & -1/2 & 1&0\\
        %  1/2 &  1/2 & 1 &0\\
        %  1/2  & -1/2 & -1&0 \\
        %  -\beta_+^{-1} &0 &0&0\\
        %  0 &-\beta_-^{-1} &0&0\\
        %  0 & 0 &0 &1
        % \end{pmatrix}, 
        \left(
        \begin{array}{c|c}
                 U_\TS & \bm{0}_4 \\
                 \hline
                 \bm{0}_4^\intercal & 1
        \end{array}
        \right),
        ~ 
        \bm{u} := 
        \left(
        \begin{array}{c}
                 \bm{u}_\TS \\
                 \hline
                 0
        \end{array}
        \right),
    \end{equation}
and $\bm{c}\in\{\bm{x}\in\R^4|~ U \bm{x} + \bm{u} > 0,~\langle\bm{x},(1/2,1/2,1,1)\rangle >0\}$.
\end{proposition}

\begin{proof}
    Combining \eqref{eq:digital-option-def} and proposition \ref{prop:density-Mellin-GTS}, we can write: 
    %with the same notations:
    \begin{multline}
        \label{eq:proof-cn-exponential}
        C_\CN  
        = 
        \frac{e^{(\gamma-r)T}}{\beta_+\beta_-} \int_{-k}^\infty\int\limits_{\bm{c}+\mi \R^3 }  e^{-\lambda_+ x} \GG
        \left(U_\TS\bs + \bu_\TS, D_\TS\bs_{\leqslant 2} + \bd_\TS\right)\\
        \times 
        \ubar{\lambda}^{(s_1+s_2)/2-s_3}
        x^{-1-\frac{s_1+s_2}{2}-s_3}
        (a_+T)^{s_1/\beta_+}(a_-T)^{s_2/\beta_-}
        \frac{\D \bs}{(\mi 2 \pi)^3}\D x.
    \end{multline}
    For the exponential term, since $\lambda_+>0$, we use the following one-dimensional MB representation (see \cite[eq. 6.3.1]{Bateman}):
    \begin{equation}\label{Cahen}
       \forall x>0, 
       \quad
       e^{-\lambda_+x} = \int\limits_{c_4+\mi\R} \Gamma(s_4) (\lambda_+x)^{-s_4}\frac{\D s_4}{\mi 2 \pi}
    \end{equation}
    which holds for any $c_4>0$. We thus obtain:
    \begin{multline}
        \label{eq:proof-an-int}
        C_\CN  
        =
        \frac{e^{(\gamma-r)T}}{\beta_+\beta_-} \int_{\bm{c}+\mi \R^4 }
        \left(\int_{-k}^\infty x^{-1-\frac{s_1+s_2}{2}-s_3-s_4} \D x\right)
         \GP{U\bs + \bu}
         {D\bs_{\leqslant 2} + \bd}
         \Gamma(s_4)\\
        \times  \ubar{\lambda}^{(s_1+s_2)/2-s3}\lambda_+^{-s_4}
        % x^{-1-\frac{s_1+s_2}{2}-s_3}
        (a_+T)^{s_1/\beta_+}(a_-T)^{s_2/\beta_-}
        \frac{\D \bs}{(\mi 2 \pi)^4}
        ,
    \end{multline}
    where $\bm{c}$ is now a 4-vector belonging to set $\{ \bm{x}\in\R^4|~U\bm{x} +\bm{u}> 0  \}$. Last, the integral w.r.t. $x$ can be carried out as:
    \begin{equation}
    \label{eq:integral-x}
        \int_{-k}^\infty x^{-1-\frac{s_1+s_2}{2}-s_3-s_4} \D x 
        =
       \frac{(-k)^{\frac{s_1+s_2}{2}+s_3+s_4}}{\frac{s_1+s_2}{2}+ s_3+ s_4}
    \end{equation}
    provided that $\Real(\frac{s_1+s_2}{2}+s_3+s_4)>0$ and \blue{$k<0$}. Substituting \eqref{eq:integral-x} in \eqref{eq:proof-an-int} yields the desired result. 
    \end{proof}
    
\begin{proposition}
    \label{prop:AN-Mellin-GTS}
    The value of an AN call with moneyness $k<0$ admits the MB representation:
    \begin{multline}
    \label{eq:AN-Mellin-GTS}
    C_\AN = 
    \frac{Ke^{k+(\gamma-r)
    T}}{\beta_+\beta_-}  
    \int_{\bm{c}+\mi \R^4 }
    %\\
     \frac{\GP{U\bs + \bu}{ D_\TS\bs_{\leqslant 2}+ \bd_\TS}}{\frac{s_1+s_2}{2}+s_3+s_4}\\
     % \left(\frac{s_1+s_2}{2}+s_3+s_4\right)
    \times  
    \ubar{\lambda}^{(s_1+s_2)/2-s_3}
    (\lambda_+-1)^{-s_4}
    (a_+T)^{s_1/\beta_+}
    (a_-T)^{s_2/\beta_-}(-k)^{-\frac{s_1+s_2}{2}-s_3-s_4}
    \frac{\D \bs}{(\mi 2 \pi)^4}
\end{multline}
with the same notation and conditions as in proposition \ref{prop:CN-Mellin-GTS}.
\end{proposition}

\begin{proof}
    The proof is a straightforward adaptation of the proof of proposition \ref{prop:CN-Mellin-GTS}, except that the arising exponential term is now $e^{-(\lambda_+-1)x}$ instead of $e^{-\lambda_+ x }$ in \eqref{eq:proof-cn-exponential}. From \blue{assumption \ref{assum:lambda}}, $\lambda_+ - 1>0$ and therefore we can write:
    \begin{equation}\label{Cahen-lam-1}
       \forall x>0, 
       \quad
       e^{-(\lambda_+-1)x} = \int\limits_{c_4+\mi\R} \Gamma(s_4) ((\lambda_+-1)x)^{-s_4}\frac{\D s_4}{\mi 2 \pi}.
    \end{equation} 
    The rest of the proof is the same as for proposition \ref{prop:CN-Mellin-GTS}.
    % Adjusting the multiplicative coefficients of \eqref{eq:digital-option-def} gives the desired result.
\end{proof}
    
\subsubsection{Series expansions}
    Let us now express the MB integrals of propositions \ref{prop:CN-Mellin-GTS} and \ref{prop:AN-Mellin-GTS} as  residue series. To that extent, given $x_1<0$, $x_2>0$ and an integer vector $\bm{n} := (n_1,n_2,n_3)\in\N^3$, let us define the coefficients: 
    \begin{align}
    \label{eq:coefficients-ts}
    \begin{cases} 
        \begin{aligned}
            c_{\bm{n}}^{(1)}(x_1;x_2) := &\ (-\beta_-n_3)_{n_1}
            \frac{\Gamma(1-n_1+\beta_+n_2+\beta_-n_3)}{\Gamma(1+\beta_+n_2)\Gamma(-\beta_+n_2)} \\
            &\times
            \ubar{\lambda}^{n_1} 
            x_2^{-n_1+\beta_+n_2+\beta_-n_3}
            \gamma_{\text{inc}}(n_1-\beta_+n_2-\beta_-n_3,-x_1 x_2), \\
            c_{\bm{n}}^{(2)}(x_1;x_2) := &\ (1+\beta_+n_2)_{n_1} 
            \frac{\Gamma(-1-n_1-\beta_+n_2-\beta_-n_3)}{\Gamma(-\beta_+n_2)\Gamma(-\beta_-n_3)} \\
            &\times
            \ubar{\lambda}^{1+n_1+\beta_+n_2+\beta_-n_3} 
            x_2^{-1-n_1}
            \gamma_{\text{inc}}(1+n_1,-x_1x_2 ), \\
            c_{\bm{n}}^{(3)}(x_2) := &
            \frac{\Gamma(n_1-\beta_+n_2-\beta_-n_3)}{\Gamma(1+n_1-\beta_+n_2)\Gamma(-\beta_-n_3)} \times
            \ubar{\lambda}^{-n_1+\beta_+ n_2+\beta_-n_3}
            x_2^{n_1},
        \end{aligned}
    \end{cases}
\end{align}
    where $(z)_n$ denotes the Pochhammer symbol (descending factorial) and $\gamma_{\text{inc}}$ is the lower incomplete gamma function \blue{defined for $a\notin \R\backslash\{-\N\}$ and $z\in\R$ as (see \cite[eq. 8.7.1]{DLMF}):
    \begin{equation}
        \label{eq:series-incomp-series}
        \gamma_{\text{inc}}(a,z) =\sum_{k=0}^\infty \frac{(-z)^k}{k!(a+k)}.
    \end{equation}}
    Furthermore, the coefficients \eqref{eq:coefficients-ts} admit the special limiting cases:
    \begin{equation}
        \label{eq:cn-coeff-boundary}
        \begin{cases}
            \displaystyle
            c_{ (n_1,0,0)}^{(1)}(x_1;x_2) :=  
            \frac{\beta_+}{\beta_-+\beta_+}\mathbbm{1}_{\{n_1 = 0\}},\\
            \displaystyle
            c_{(n_1,0,0)}^{(2)}(x_1;x_2) := 0,
            \\
            \displaystyle
            c_{(0,0,0)}^{(3)}(x_2) := \frac{\beta_-}{\beta_-+\beta_+}.
        \end{cases}
    \end{equation}
    We finally define the overall coefficient $c_{\bm{n}}$ as:
    \begin{equation}
    \label{eq:coefficients-ts-sum}
        c_{\bm{n}}(x_1,x) := \frac{(-1)^{n_1}}{n_1!}
        \left(
        c_{\bm{n}}^{(1)}(x_1;x_2)
        +c_{\bm{n}}^{(2)}(x_1;x_2)
        \right)
        +c_{\bm{n}}^{(3)}(x_2).
    \end{equation}

\begin{proposition}
     \label{prop:cn-series-gts}
    The value of a CN call with moneyness $k<0$ admits the series representation:
    \begin{equation}\label{eq:cn-series-gts}
        C_\CN 
        = 
        e^{-rT}
        -e^{(\gamma-r)T}\sum_{\bm{n}\in\N^3}
        \frac{(-1)^{n_{2}+n_3}}{n_{2}!n_3!}
        (a_+T)^{n_2}(a_-T)^{n_3}
        c_{\bm{n}}(k;\lambda_+).
    \end{equation}
\end{proposition}

\begin{proof}
% Using the formalism of \cite{passare1994multidimensional, zhdanov1998} in the $\mathbb{C}^4$ case, we denote the characteristic vector associated to the MB integral \eqref{eq:CN-Mellin-GTS} by $\bm{\Delta}=(1/2-\beta_+^{-1}, 1/2-\beta_-^{-1},1,1)$.
The characteristic vector $\bm{\Delta}$ associated to the 
MB integral \eqref{eq:CN-Mellin-GTS} is $\bm{\Delta}=(1/2-\beta_+^{-1}, 1/2-\beta_-^{-1},1,1)$.
% The multidimensional Jordan residue lemma \cite{passare1994multidimensional} 
% allows to express 
\blue{Following \eqref{eq:residue_thm}, we can express} the integral
\eqref{eq:CN-Mellin-GTS} as a sum of residues associated to the poles of the integrand located in the admissible subspace $\Pi_{\bm{\Delta}}$ as defined in \eqref{eq:Pi_def}.
% The multidimensional Jordan residue lemma \cite{passare1994multidimensional} allows to express \eqref{eq:CN-Mellin-GTS} as a sum of residues associated to the poles of the integrand located in the subspace $\{ s\in\mathbb{C}^4|~ \Re[\langle \bm{\Delta} , s \rangle] <  \langle \bm{\Delta}, \bm{c} \rangle \}$. 
In this subspace, there are four sets of poles; the first two sets are given by the solutions of the two linear systems
\begin{equation}
    \begin{cases}
        U_{\{1,4,5,6\}}\bm{s} + \bm{u}_{\{1,4,5,6\}} = -\bm{n},\\
        U_{\{2,4,5,6\}}\bm{s} + \bm{u}_{\{2,4,5,6\}} = -\bm{n},
    \end{cases}
\end{equation}
    % \end{enumerate}
for $\bm{n}\in\mathbb{N}^4$, and the other two sets are given by the solutions of the two linear systems
\begin{equation}
\begin{cases}
    U_{\{4,5,6\}}\bm{s} + \bm{u}_{\{4,5,6\}} = -\bm{n}~\text{and}~\langle \bm{s},(1/2,1/2,1,1)\rangle =0,\\
    U_{\{3,4,5\}}\bm{s} + \bm{u}_{\{3,4,5\}} = -\bm{n}~\text{and}~\langle \bm{s},(1/2,1/2,1,1)\rangle =0,
\end{cases}
\end{equation}
for $\bm{n}\in\mathbb{N}^3$. 
% For  $\bm{n}\in\N^4$, the first poles are at the parametrized locations $\bm{s}=(\beta_+n_2,\beta_-n_3,-n_1+\beta_+n_2/2+\beta_-n_3/2,-n_4)$ and $\bm{s}=(\beta_+n_2,\beta_-n_3,-1-n_1-\beta_+n_2/2-\beta_-n_3/2,-n_4)$. For $\bm{n}\in\N^3$, the last two are located in $\bm{s} = (\beta_+n_2,\beta_-n_3,-\beta_+n_2/2-\beta_-n_3/2+n_1,-n_1)$ and $\bm{s}=(\beta_+n_2,\beta_-n_3,n_1+\beta_+n_2/2-\beta_-n_3/2,-n_1-\beta_+n_2)$.
    These four systems have unique solutions, respectively given by:
    \begin{equation}
        \begin{cases}\label{Pi}
            \text{(P1)} \quad \bm{s}=(\beta_+n_2,\beta_-n_3,-n_1+\beta_+n_2/2+\beta_-n_3/2,-n_4),\\
            \text{(P2)} \quad \bm{s}=(\beta_+n_2,\beta_-n_3,-1-n_1-\beta_+n_2/2-\beta_-n_3/2,-n_4),\\
            \text{(P3)} \quad \bm{s} = (\beta_+n_2,\beta_-n_3,-\beta_+n_2/2-\beta_-n_3/2+n_1,-n_1),\\
            \text{(P4)} \quad \bm{s}=(\beta_+n_2,\beta_-n_3,n_1+\beta_+n_2/2-\beta_-n_3/2,-n_1-\beta_+n_2).
        \end{cases}
    \end{equation}
    Detailed computations for the residues associated to the sets of poles (P1), (P2) and (P4) and their simplifications (in terms of lower incomplete gamma functions for (P1) and (P2), and as a simple constant for (P4)) are detailed in appendix \ref{subsec:app-first-pole},  \ref{subsec:app-second-pole} and \ref{subsec:app-third-pole}. As for the residues associated to (P3), they are straightforward to obtain thanks to the relation:
    \begin{equation}
            \Gamma(-n_1+\beta_+n_2) = (-1)^{n_1-1}\frac{\Gamma(-\beta_+n_2)\Gamma(1+\beta_+n_2)}{\Gamma(1+n_1-\beta_+n_2)}.
\end{equation}
Summing all residues yields the series \eqref{eq:cn-series-gts}.

\end{proof}

In the following remark, we show that it is straightforward to extend proposition \ref{prop:cn-series-gts} to the case of a positive moneyness. 
\begin{remark}
    \label{rem:cn-positive-moneyness}
    Assume $k>0$, i.e. $S_0\exp((r-q+\zeta)T)<K$, and denote the corresponding option value by $\widetilde{C}_\CN$. In this situation, we simply write 
    % use that $\mathbb{E}[\mathbbm{1}_{\{S_T>K\}}] = 1- \mathbb{E}[\mathbbm{1}_{\{S_T<K\}}]$ and it comes that:
    \begin{equation}
        \mathbb{E}[\mathbbm{1}_{\{S_T>K\}}] 
        =
        1- \mathbb{E}[\mathbbm{1}_{\{S_T<K\}}]
        =
        1-\int_{-\infty}^{-k}f_{X_T}(x)\D x
        ,
    \end{equation}
    where $(-\infty, -k)\subset \mathbb{R}_-$ and therefore one has to use the TS density for negative $x$, which is given by the symmetry relation \eqref{eq:dens-symetry}. The proofs of propositions \ref{prop:CN-Mellin-GTS}
    and \ref{prop:cn-series-gts} can be straightforwardly adapted to obtain that:
    \begin{equation}
       \widetilde{C}_\CN 
        = 
        e^{(\gamma-r)T}\sum_{\bm{n}\in\N^3}
        \frac{(-1)^{n_{2}+n_3}}{n_{2}!n_3!}
        (a_+T)^{n_2}(a_-T)^{n_3}
        \widetilde{c}_{\bm{n}}(-k;\lambda_-)
    \end{equation}
    % with
    % \begin{equation}
    %     \widetilde{c}_{\CN,\bm{n}}(-k,\lambda_-) := \frac{(-1)^{n_1}}{n_1!}
    %     \left(
    %     \widetilde{c}_{\bm{n}}^{(1)}(-k;\lambda_-)
    %     +
    %     \widetilde{c}_{\bm{n}}^{(2)}(-k;\lambda_-)
    %     \right)
    %     +
    %     \widetilde{c}_{\bm{n}}^{(3)}(\lambda_-)
    % \end{equation}
    where the coefficient $\widetilde{c}_{\bm{n}}(x_1,x_2)$ is the same as the coefficient $c_{\bm{n}}(x_1,x_2)$ defined in 
    \eqref{eq:coefficients-ts-sum}, but with flipped parameters, i.e., the parameters $\{\alpha_+,\beta_+,\lambda_+,\alpha_-,\beta_-,\lambda_-\}$ become $\{\alpha_-,\beta_-,\lambda_-,\alpha_+,$ $\beta_+,\lambda_+\}$ in definition \eqref{eq:coefficients-ts}.
\end{remark}
%%%%%%%%%%%

The AN option can be treated almost exactly the same way as the CN one, the difference being that we now express the MB integral \eqref{eq:AN-Mellin-GTS} featuring a $(\lambda_+-1)^{-s_4}$ term instead of the MB integral \eqref{eq:CN-Mellin-GTS} featuring a $\lambda_+^{-s_4}$ term. We obtain:

\begin{proposition}
     \label{prop:an-series-gts}
    The value of an AN call $C_\AN$ with moneyness $k<0$ admits the series representation:
    \begin{equation}
        C_\AN 
        = 
        Ke^{k-(\zeta+r)T}
        -Ke^{k+(\gamma-r)T}\sum_{\bm{n}\in\N_0^3}
        \frac{(-1)^{n_{1}+n_2+n_3}}{n_{1}!n_2!n_3!}
        (a_+T)^{n_2}(a_-T)^{n_3}
        c_{\bm{n}}(k;\lambda_+-1).
    \end{equation}
\end{proposition}

\begin{proof}
    The proof is almost the same than in proposition \ref{prop:CN-Mellin-GTS}, given that the integrand admits the same sets of poles (P1), (P2), (P3) and (P4) defined in \eqref{Pi}. The residues are computed in a similar way, with a slight difference for the residues associated to (P4) which is detailed in 
    Appendix \ref{subsec:app-fourth-pole}. 
\end{proof}

As for CN options, the case of positive moneyness \green{is} easily handled, as shown in remark \ref{rem:an-positive-moneyness}.
\begin{remark}
\label{rem:an-positive-moneyness}
    For $k>0$, we use the identity:
    \begin{equation}
        \mathbb{E}[S_T\mathbbm{1}_{\{S_T>K\}}] 
        =
        S_0 e^{(r-q)T} - \mathbb{E}[S_T\mathbbm{1}_{\{S_T<K\}}]
        =
        S_0e^{(r-q)T}-
        Ke^{k}\int_{-\infty}^{-k}e^xf_{X_T}(x)\D x.
    \end{equation}
    It comes that the value of the asset-or-nothing
    option is given by:
    \begin{multline}
        \widetilde{C}_\AN 
        =
        S_0e^{-qT}-
        Ke^{k-(\zeta+r)T}
        \\
        +Ke^{k+(\gamma-r)T}\sum_{\bm{n}\in\N_0^3}
        \frac{(-1)^{n_{1}+n_2+n_3}}{n_{1}!n_2!n_3!}
        (a_+T)^{n_2}(a_-T)^{n_3}
        \widetilde{c}_{\bm{n}}(-k;\lambda_-+1).
    \end{multline}
    where $\widetilde{c}_{\bm{n}}$ denotes the coefficient with flipped parameters (as in remark \ref{rem:cn-positive-moneyness}). 
\end{remark}

\subsection{European options}

As mentioned before, the European call option price is simply obtained by differences of the AN (proposition \ref{prop:an-series-gts}) and CN (proposition \ref{prop:cn-series-gts}) prices. We immediately obtain: 

    \begin{proposition}
     \label{prop:eur-series-gts}
        The value of a European call $C_\EUR$ with moneyness $k<0$ admits the series representation:
        \begin{equation}\label{eur-series-gts}
            C_\EUR = Ke^{-rT}\left[e^{k-\zeta T}-1 + 
            e^{\gamma T}
            \sum_{\bm{n}\in\N^3}
            \frac{(-1)^{n_{1}+n_2+n_3}}{n_{1}!n_2!n_3!}
            (a_+T)^{n_2}(a_-T)^{n_3}
            c_{\EUR,\bm{n}}(k;\lambda_+)
            \right]
        \end{equation}
        where
        \begin{equation}
            c_{\EUR,\bm{n}}(k,\lambda_+)
            :=
            e^kc_{\bm{n}}(k,\lambda_+-1)  
            -c_{\bm{n}}(k,\lambda_+).  
            \end{equation}
    \end{proposition}
    % \begin{proof}
    %     \color{black}
    %     Immediate consequence of the payoff relation \eqref{eq:payoff_Eureopean} and propositions \ref{prop:an-series-gts} and \ref{prop:cn-series-gts}
    % \end{proof}
    
Last, the case of positive moneyness $k>0$ is again straightforward to handle, thanks to remarks \ref{rem:an-positive-moneyness} and \ref{rem:cn-positive-moneyness}.

\section{Particular cases}
\label{sec:particular-cases}
%%%%%%%%%%%%%%%%%%%%%%%%%%%%%%%%%%
%%%%%%%%%%%%%%%%%%%%%%%%%%%%%%%%%%

The pricing series \eqref{eur-series-gts} is valid under the most general TS processes; however, as we will see, it can be simplified in some particular cases that correspond to models that are popular in the literature. These simplifications are particularly remarkable in the cases of bilateral Gamma and Variance Gamma pricing.

%%%%%%%%%%%
\subsection{KoBoL}
%%%%%%%%%%%%

The KoBoL process (see \cite{BoyarOption00,BoyarLevin02}) corresponds to a TS process with the restriction $\beta_+=\beta_-\blue{=: \beta>0}$; to obtain the corresponding European call price, it therefore suffices to take $\beta_+=\beta_-$ in the coefficients \eqref{eq:coefficients-ts} and the overall coefficient \eqref{eq:coefficients-ts-sum}, the latter being denoted by $c_{\bm{n}}^{\text{KoBoL}}(k,x)$ in this case. Applying the pricing formula \eqref{eur-series-gts} (proposition \ref{prop:eur-series-gts}) thus yields the KoBoL call price:
\begin{equation}
\label{eur-series-KoBoL}
            C_\EUR^{\text{KoBoL}} = Ke^{-rT}\left[e^{k-\zeta T}-1 + 
            e^{\gamma T}
            \sum_{\bm{n}\in\N^3}
            \frac{(-1)^{n_{1}+n_2+n_3}}{n_{1}!n_2!n_3!}
            (a_+T)^{n_2}(a_-T)^{n_3}
            c_{\EUR,\bm{n}}^{\text{KoBoL}}(k;\lambda_+)
            \right]
\end{equation}
        where
\begin{equation}
            c_{\EUR,\bm{n}}^{\text{KoBoL}}(k,\lambda_+)
            :=
            e^kc_{\bm{n}}^{\text{KoBoL}}(k,\lambda_+-1)  
            -c_{\bm{n}}^{\text{KoBoL}}(k,\lambda_+).  
\end{equation}

%%%%%%%%%%%
\subsection{CGMY}
%%%%%%%%%%%%

The CGMY process of \cite{Carr02} is a KoBoL process with furthermore $\alpha_+ = \alpha_- := C$\blue{. In this context the traditional usage is to write $\beta_+=\beta_-:=Y$}. Interestingly, one can still use the KoBoL coefficients like in \eqref{eur-series-KoBoL}; the only difference actually lies in the $a_\pm$ terms that are now equal and can therefore be grouped into the same power, yielding the CGMY call price:
\begin{equation}
\label{eur-series-CGMY}
            C_\EUR^{\text{CGMY}} = Ke^{-rT}\left[e^{k-\zeta T}-1 + 
            e^{\gamma T}
            \sum_{\bm{n}\in\N^3}
            \frac{(-1)^{n_{1}}}{n_{1}!n_2!n_3!}
            (C \Gamma(-Y)T)^{n_2+n_3}
            c_{\EUR,\bm{n}}^{\text{KoBoL}}(k;M)
            \right]
\end{equation}
where we have used the traditional notation $\lambda_+=M$. As before, extension to in-the-money situation $k>0$ is straightforward for both KoBoL and CGMY processes, thanks to remarks \ref{rem:an-positive-moneyness} and \ref{rem:cn-positive-moneyness}.

%%%%%%%%%%%
\subsection{Bilateral Gamma}
%%%%%%%%%%%
\label{subsec:particular-cases-bg}

The BG process of \cite{kuchlerBG} corresponds to a TS process with the restriction $\beta_+=\beta_-=0$. As such, it is therefore a particular case of a KoBoL process, but not of a CGMY process \blue{as $\alpha_+\neq \alpha_-$ in general}. 
A series pricing formula has already been derived in \cite{kuchlerBG}, however it is valid only for at-the-money options ($k=0$); a formula valid for all $k$ has been derived in \cite{aguilarkirkby2022robust} however it involves three summation indices. In this section, we will obtain a much simpler formula, valid for all moneyness and involving only one summation index.
% The degenerate tempered stable distribution obtained by setting $\beta_+=\beta_- = 0$ is the so-called bilateral Gamma distribution introduced in \cite{kuchlerBG}. 
% \red{Dire que c'est avoir un $X_t$ un processus BG dans \eqref{St_dynamics} puis ce que ça veut dire en termes de coeffs} 
% In this section, we generalize the pricing formula given in \cite{kuchlerBG} for $k\neq 0$ and considerably simplify the formula obtained in \cite{aguilarkirkby2022robust}.
For a $\mathrm{BG}(\alpha_+,\lambda_+,\alpha_-,\lambda_-)$ process, let us define the following five quantities:
\begin{equation}
        \begin{cases}
        \displaystyle
        m_T := \frac{\lambda_+^{\alpha_+ T}\lambda_-^{\alpha_- T}}{\underline{\lambda}^{\underline{\alpha}T}\Gamma(\alpha_+ T)\Gamma(\alpha_- T)\Gamma(1-\alpha_+ T)},\\
        \displaystyle
        \underline{\lambda} := \lambda_++\lambda_-,\\
        \displaystyle
        \overline{\lambda} := \frac{1}{2}(\lambda_+-\lambda_-),\\
        \displaystyle
        \underline{\alpha} := \frac{1}{2}(\alpha_++\alpha_-),\\
        \displaystyle
        \overline{\alpha} := \frac{1}{2}(\alpha_+-\alpha_-).
        \end{cases}
\end{equation}
Moreover, the moneyness is now defined by  $k_{\BG}:=\log (S_0/K)+(r-q+\zeta_\mathrm{BG})T$ where
\begin{equation}
    \zeta_\BG := -\ln \varphi_\BG (-\mathrm{i})
    =\blue{
    -\ln\left[\left(\frac{\lambda_+}{\lambda_+-1}\right)^{\alpha_+}
    \left(\frac{\lambda_-}{\lambda_-+1}\right)^{\alpha_-}\right].}
\end{equation}
In the rest of this section, we will still denote $k_\BG$ by $k$, on order to preserve the \blue{simplicity of notation}. 

In \cite[eq. 4.8]{kuchlerBG}, the authors show that the density function $f_\BG$ of the BG distribution can be written as:
\begin{equation}
    \forall x\in(0,+\infty),\quad
    f_\BG(x)
    =m_1\Gamma(1-\alpha_+)\Gamma(\alpha_-)
    x^{\underline{\alpha}-1}
    e^{-(\lambda_+-\lambda_-)x/2}
    W_{\overline{\alpha},\underline{\alpha}-1/2}(\underline{\lambda}x).
\end{equation}
Introducing the MB representation \cite[eq. 13.16.11]{DLMF} for the Whittaker function yields
% to the density expression given in \cite{BG} to get its Mellin-Barnes representation. If $f_\mathrm{BG}$ denotes the density function of a $\mathrm{BG}(\alpha_+, \lambda_+,\alpha_-,\lambda_-)$ distribution, 
    % we have for all $x>0$:
    \begin{equation}\label{eq:MB-density-BG}
        f_\mathrm{BG}(x)=
        m_1
        e^{-\lambda_+x}
        \int\limits_{c+\mi \R}
        \Gamma\left(\underline{\alpha}+s\right)\Gamma(1-\underline{\alpha}+s)\Gamma(-\overline{\alpha}-s)\underline{\lambda}^{-s}x^{-s-1+\underline{\alpha}}
        \frac{\D s}{\mi 2\pi}
    \end{equation}
    where $c\in(-\underline{\alpha}\vee(\underline{\alpha}-1),-\overline{\alpha})$.
    % \begin{equation}
    %     \begin{cases}
    %     \displaystyle
    %     M_\mathrm{BG} := \frac{\lambda_+^{\alpha_+}\lambda_-^{\alpha_-}}{\underline{\lambda}^{\underline{\alpha}}\Gamma(\alpha_+)\Gamma(\alpha_-)\Gamma(1-\alpha_+)},\\
    %     \displaystyle
    %     \underline{\lambda} := \lambda_++\lambda_-,\\
    %     \displaystyle
    %     \underline{\alpha} := \frac{\alpha_++\alpha_-}{2},\\
    %     \displaystyle
    %     \overline{\alpha} := \frac{\alpha_+-\alpha_-}{2},\\
    %     \end{cases}
    % \end{equation}
    % and $\Real(c)\in\cdots$ \red{TODO}. The residue theorem leads to the series expansion for $x>0$:
    % \begin{equation}
    % \begin{aligned}
    %     f_\mathrm{BG}(x)=
    %     M_\mathrm{BG}&e^{-\lambda_+ x}
    %     \left(
    %     \sum_{n=0}^\infty 
    %     \frac{(-1)^n}{n!}\Gamma(\alpha_-+n)\Gamma(1-2\underline{\alpha}-n)\underline{\lambda}^{\underline{\alpha}+n}x^{-1+2\underline{\alpha}+n}
    %     \right.\\
    %     & \left. +\sum_{n=0}^\infty \frac{(-1)^n}{n!} \Gamma(1-\alpha_++n)\Gamma(-1+2\underline{\alpha}-n)\underline{\lambda}^{1-\underline{\alpha}+n}x^n\right). 
    % \end{aligned}
    % \end{equation}

    \begin{proposition}
        The values of a CN and of an AN call with moneyness $k<0$ admit the MB representations:
        \begin{equation}\label{eq:MB-digital-BG}
            \begin{cases}
                \displaystyle
                C_{\mathrm{CN}}^\BG
                =
                m_Te^{-rT} \int\limits_{\bm{c}+\mi\R}
                \frac{\mathbf{\Gamma}^\pi (U^{\mathrm{BG}}\bs+{\bu}^{\mathrm{BG}})}{s_2}
                \underline{\lambda}^{-s_1}
                \lambda_+^{s_1-s_2-\underline{\alpha}}
                (-k)^{-s_2}
                \frac{\D \bm{s}}{(\mi 2\pi)^2}\\
                \displaystyle C_{\mathrm{AN}}^\BG
                =
                Km_Te^{k-rT} \int\limits_{\bm{c}+\mi\R}
                \frac{\mathbf{\Gamma}^\pi (U^{\mathrm{BG}}\bs+\bu^{\mathrm{BG}})}{s_2}
                \underline{\lambda}^{-s_1}
                (\lambda_+-1)^{s_1-s_2-\underline{\alpha}}
                (-k)^{-s_2}
                \frac{\D \bs}{(\mi 2\pi)^2}
            \end{cases}
            .
        \end{equation}
        where:
    \begin{equation}
        U^{\mathrm{BG}} := 
        \begin{pmatrix}
            1 & 0\\
            1 &0 \\
            -1 & 0 \\
            -1 & 1 
        \end{pmatrix}
        ,\quad
        \bm{u}^{\mathrm{BG}} := 
        \begin{pmatrix}
            \underline{\alpha}T\\
            1-\underline{\alpha}T\\
            -\overline{\alpha}T\\
            \underline{\alpha}T
        \end{pmatrix}
    \end{equation}
    and $\bm{c}\in\{\bm{x}\in\R^2|~U^\BG \bm{x} + \bu^{\BG} > 0 ~\langle\bm{x},(0,1)\rangle >0  \}$.
    \end{proposition}
    \begin{proof}
    The proof is similar to the proof of proposition \blue{\ref{prop:CN-Mellin-GTS}}: we insert the MB representation for the BG density \eqref{eq:MB-density-BG} in the CN (resp. AN) payoff of \eqref{eq:digital-option-def}, then we introduce a second MB representation for the $e^{-\lambda_+ x}$ (resp. $e^{-(\lambda_+-1)x}$) term, and last we integrate over the $x$ variable to obtain \eqref{eq:MB-digital-BG}.
        % The integrals are obtained by the same computations already done in the proof of proposition \ref{prop:-cn-an-bg-mellin}. \red{détailler un peu mieux: MB pour Whittaker et exponentielle + intégration sur x}
    \end{proof}
    
    Before providing the series expansion for the option prices, let us first define for $x_1<0$, $x_2>0$, for an integer $n\in\N$, and for a maturity $T>0$, the coefficients:
    \begin{equation}
        \label{eq:bg-series-coefficients}
        \begin{cases}
            \displaystyle
             c_n^{\BG, (1)}(x_1;x_2) :=  \Gamma(1 - 2 \underline{\alpha}T - n) \Gamma(\alpha_-T   + n) \underline{\lambda}^{n + \underline{\alpha}T} x_2^{-n - 2 \underline{\alpha}} \gamma_\mathrm{inc} (n + 2 \underline{\alpha}T, -x_1 x_2), \\
            \displaystyle
            c_n^{\BG, (2)}(x_1;x_2) := 
             \Gamma(- \alpha_+T  + 1 + n) \Gamma(2 \underline{\alpha}T - 1 - n) \underline{\lambda}^{- \underline{\alpha} T+ 1 + n} x_2^{-1 - n} \gamma_\mathrm{inc}  (n+ 1, -x_1 x_2),\\
             c_n^{\BG}(x_1;x_2) :=c_n^{\BG, (1)}(x_1;x_2)+c_n^{\BG, (2)}(x_1;x_2),
        \end{cases}
    \end{equation}
    and the function $h$ defined by:
    \begin{equation}
    \label{eq:h_function}
        \forall x\in(-\underline{\lambda},\underline{\lambda}),\quad h(x) := \frac{\lambda_+^{\alpha_+T}\lambda_-^{\alpha_-T}    \Gamma(2 \underline{\alpha}T)}{\underline{\lambda}^{2\underline{\alpha}T}\Gamma(1+\alpha_+T)\Gamma(\alpha_-T)} 
    ~{_2}F_1 \left( 2\underline{\alpha}T,1 ; 1+\alpha_+T ; \frac{x}{\underline{\lambda}}\right),
    \end{equation}
    where $_2F_1$ is the Gauss hypergeometric function (see \cite{Abramowitz72} or any other classical monograph on special functions).
    \begin{proposition}
        \label{prop:series-digital-BG}
         The values of a CN and of an AN call $C_\CN^{\BG}$ and $C_\TS^{\BG}$ with moneyness $k<0$ admits the series representation:
         % The value of a CN call with moneyness k < 0 admits the series representation:
         \begin{equation}
            \label{eq:bg-series-digital}
             \begin{cases}
             \displaystyle
                 C_\CN^{\BG}  = e^{-rT}\left(1-h(\lambda_+)-m_T\sum_{n\in\N}
                 \frac{(-1)^n}{n!}c_n^{\BG}(k;\lambda_+)
                 \right),\\
             \displaystyle
                 C_\AN^{\BG}  = Ke^{k-rT}\left(e^{-\zeta_\mathrm{BG} T}-h(\lambda_+-1)-m_T\sum_{n\in\N}
                 \frac{(-1)^n}{n!}c_n^{\BG}(k;\lambda_+-1)
                 \right).
             \end{cases}
         \end{equation}
    \end{proposition}
    \begin{proof}
        %     Using the formalism of [40, 48] in the C4 case, we denote the characteristic vector associ-
        % ated to the MB integral (48) by ∆ = (1/2 − β−1
        % + , 1/2 − β−1
        % − , 1, 1) The multidimensional Jordan
        % residue lemma [40] allows to express (48) as a sum of residues associated to the poles of the
        % integrand located in the subspace {s ∈ C4| ℜ[⟨∆, s⟩] < ⟨∆, c⟩}.
        % Using the formalism of \cite{passare1994multidimensional, zhdanov1998} in the $\mathbb{C}^2$ case, we denote the characteristic vector associated to the MB integrals \eqref{eq:MB-digital-BG} by $\bm{\Delta}=(0,1)$.
        % The multidimensional Jordan residue lemma \cite{passare1994multidimensional} allows to express \eqref{eq:CN-Mellin-GTS} as a sum of residues associated to the poles of the integrand located in the subspace
        % $\{ s\in\mathbb{C}^2|~ \Re[\langle \bm{\Delta} , s \rangle] <  \langle \bm{\Delta}, \bm{c} \rangle \}$.
        The characteristic vector $\bm{\Delta}$ associated to the 
        MB integral \eqref{eq:MB-digital-BG} is $\bm{\Delta}=(0,1)$.
        % The multidimensional Jordan residue lemma \cite{passare1994multidimensional} 
        % allows to express 
        \blue{Following \eqref{eq:residue_thm}, we can express} the integral
        \eqref{eq:MB-digital-BG} as a sum of residues associated to the poles of the integrand located in the admissible subspace $\Pi_{\bm{\Delta}}$ as defined in \eqref{eq:Pi_def}.
        In this subspace, there are four sets of poles that are solutions of the four linear systems:
        % Using the formalism of
        % As done several times, we will show the result for the cash-or-nothing option only. In the formalism of \cite{zhdanov1998} and \cite[appendix C]{aguilarkirkby2022robust}, we have $\bm{\Delta} = (0,1)$.
        % The admissible sets of poles can be shown to be, for $\bm{n}=(n_1,n_2)\in\N^2$, solutions of the systems:
        %     \begin{enumerate}
        %         \item $U_{\{3\}}^{\BG}\bm{s}+ \bu_{\{3\}}^{\BG} = -n_1$ and $s_2=0$,
        %         \item $U_{\{4\}}^{\BG}\bm{s}+ \bu_{\{4\}}^{\BG} = -n_1$ and $s_2=0$,
        %         \item $U_{\{1,4\}}^{\BG}\bm{s}+ \bu_{\{1,4\}}^{\BG} = -\bm{n}$,
        % \end{enumerate}
        \begin{equation}            
            \begin{cases}
                 U_{\{3\}}^{\BG}\bm{s}+ \bu_{\{3\}}^{\BG} = -n_1~\text{and}~\langle (0,1),\bm{s}\rangle=0\\
                 U_{\{4\}}^{\BG}\bm{s}+ \bu_{\{4\}}^{\BG} = -n_1~\text{and}~\langle (0,1),\bm{s}\rangle=0\\
                 U_{\{1,4\}}^{\BG}\bm{s}+ \bu_{\{1,4\}}^{\BG} = -\bm{n}\\
                 U_{\{2,4\}}^{\BG}\bm{s}+ \bu_{\{2,4\}}^{\BG} = -\bm{n}.
            \end{cases}
        \end{equation}
        where $\bm{n}=(n_1,n_2)\in\mathbb N^2$. These four systems have unique solutions, respectively given by:
    \begin{equation}
        \begin{cases}
            \text{(P1)}_\mathrm{BG} \quad \bm{s}=(n_1-\overline{\alpha}T, 0),\\
            \text{(P2)}_\mathrm{BG} \quad \bm{s}=(n_1+\underline{\alpha}T, 0),\\
            \text{(P3)}_\mathrm{BG} \quad \bm{s} = (-n_1-\underline{\alpha}T,-2\underline{\alpha}T-n_2-n_1),\\
            \text{(P4)}_\mathrm{BG} \quad \bm{s} = (-n_1-\underline{\alpha}T-1,-1-n_2-n_1).
        \end{cases}
    \end{equation}
        % The poles are then at the parametrized locations $\bm{s} = (n_1-\overline{\alpha}, 0)$, $\bm{s} = (n_1+\underline{\alpha}, 0)$, $\bm{s} = (-n_1-\underline{\alpha},-2\underline{\alpha}-n_2-n_1)$ and $\bm{s} = (-n_1-\underline{\alpha}-1,-1-n_2-n_1)$. 
        Computing the residues at these poles yields, in the CN case:
        {\small
        \begin{equation}   
        \label{eq:bg-series-not-simplified}
        \begin{aligned}
            C_\mathrm{CN}^{\BG}
            &=m_Te^{-rT}\left(
            \sum_{n_1\in\N}
            \frac{(-1)^{n_1}}{n_1!}
            \Gamma(\alpha_-T+n_1)
            \Gamma(1-\alpha_+T+n_1)
            \Gamma(\alpha_+T-n_1)\underline{\lambda}^{\overline{\alpha}T-n_1}\lambda_+^{n_1-\alpha_+T}\right.\\
            &\quad + \sum_{n_1 \in\N}
            \frac{(-1)^{n_1}}{n_1!}
            \Gamma(n_1 + 2\underline{\alpha}T)\Gamma(1+n_1) \Gamma(-\alpha_+T-n_1)
            \underline{\lambda}^{-n_1-\underline{\alpha}T}\lambda_+^{n_1}\\
            &\quad -\sum_{\bm{n}\in\N^2}
           \frac{(-1)^{n_{1}+n_2}}{n_{1}!n_2!} \, \frac{\Gamma(1 - 2 \underline{\alpha}T - n_1) \, \Gamma(\alpha_{-}T + n_1)}{n_2 + n_1 + 2 \underline{\alpha}T} \times \underline{\lambda}^{n_1 + \underline{\alpha}T} \, \lambda_{+}^{n_2} \, (-k)^{n_2 + n_1 + 2 \underline{\alpha}T}\\
            &\left.\quad -\sum_{\bm{n}\in\N^2}\frac{(-1)^{n_{1}+n_2}}{n_{1}!n_2!} \, \frac{\Gamma(-\alpha_{+} + 1 + n_1) \, \Gamma(2 \underline{\alpha}T - 1 - n_1)}{n_2 + 1 + n_1} \times \underline{\lambda}^{-\underline{\alpha}T + 1 + n_1} \, \lambda_{+}^{n_2} \, (-k)^{n_2 + 1 + n_1}\right).
        \end{aligned}
    \end{equation}
    % The result follows from the computations done in appendix \ref{appendix:bg-first-pole} and \ref{appendix:bg-second-pole} for the first two set of poles and the use of the series expansion of the lower gamma incomplete function given by \cite[eq. 8.7.3]{DLMF} for the last two sets of poles.
    It is easy to see that the first series in \eqref{eq:bg-series-not-simplified} can be written as:
    % \begin{equation}
    %     \label{eq:appendix-interm-1}
    %     \Gamma(\alpha_-T)\Gamma(\alpha_+T)\Gamma(1-\alpha_+T)\frac{\underline{\lambda}^{\overline{\alpha }T}}{\lambda_+^{\alpha_+T}}\sum_{n\in\N}\frac{(\alpha_-T)_n}{n!}\left(\frac{\lambda_+}{\underline{\lambda}}\right)^n.
    % \end{equation}
    % which can itself be rewritten as 
    \begin{equation}
        \label{eq:appendix-interm-2}
        \Gamma(\alpha_-T)\Gamma(\alpha_+T)\Gamma(1-\alpha_+T)\frac{\underline{\lambda}^{\overline{\alpha }T}}{\lambda_+^{\alpha_+T}}~_1F_0 \left( \alpha_-T ; \frac{\lambda_+}{\underline{\lambda}}\right).
      \end{equation}
        Recalling that $_1 F_0 (a;z)= (1-z)^{-a}$ as soon as $|z|<1$, we can therefore simplify \eqref{eq:appendix-interm-2} 
        to $1/m_T$.}
    The second series in the brackets of \eqref{eq:bg-series-not-simplified} can be expressed as:
    \begin{equation}\begin{aligned}
    %     -\frac{\Gamma(\alpha_+T)\Gamma(1-\alpha_+T)\Gamma(2\underline{\alpha}T)}{\Gamma(1+\alpha_+T)}
        \frac{\Gamma(2\underline{\alpha})\Gamma(-\alpha_+)}{\underline{\lambda}^{\underline{\alpha} T}}
        \sum_{n\in\N} \frac{1}{n!}\frac{(2\underline{\alpha}T)_n}{(1+\alpha_+T)_n}\left(\frac{\lambda_+}{\underline{\lambda}}\right)^n
        &=\blue{
        \frac{\Gamma(2\underline{\alpha})\Gamma(-\alpha_+)}{\underline{\lambda}^{\underline{\alpha} T}} \!_2F_1(2\underline{\alpha},1;1+\alpha_+,\lambda_+/\underline{\lambda})}\\
        &\blue{= - \frac{h(\lambda_+)}{m_T},}
        % -\frac{1}{m_T} h(\lambda_+)
    \end{aligned}
    \end{equation}
    \blue{by the definitions of $h$ (see \eqref{eq:h_function}) and of the $_2F_1$ function (see e.g. \cite[eq. 15.2.1]{DLMF}).
     \begin{equation}
        \forall |z|\leq 1,\quad _2F_1(a,b;c;z) := \sum_{k=0}^\infty \frac{(a)_k(b)_k}{(c)_k}\frac{z^k}{k!}.
     \end{equation}
     }
    
    % where one recognize the definition of the function $_2F_1(2\underline{\alpha},1;1+\alpha_+,\lambda_+/\underline{\lambda})$. Multiplying par $m_T$ gives the desired result.
        % \cite[eq. 8.7.1]{DLMF}
    % A direct adaptation in the asset-or-nothing case leads to $e^{-\zeta}$.
    Third and fourth series in \eqref{eq:bg-series-not-simplified} can be expressed with lower incomplete gamma functions by recognizing its series expansions over $n_2$ (see \blue{\eqref{eq:series-incomp-series}}).
    For the AN case, the same computations leads to $e^{-\zeta_\mathrm{BG}T}$ instead of 1 for the first series, and $\lambda_+$ is replaced by $\lambda_+-1$ for the other series. 
    \end{proof}

As before in the general TS case (see \eqref{eq:payoff_Eureopean}), the price of the European claim $C_\EUR^{\BG}$ can be written as $C_\text{Eur}^{\BG} = C_\text{AN}^{\BG} - KC_\text{CN}^{\BG}$.

\begin{remark}
    The extension to $k>0$ is straightforward by the adaptation of remarks \ref{rem:cn-positive-moneyness} and \ref{rem:an-positive-moneyness}.
    Moreove, since $n + 2 \underline{\alpha}T$ and $n+ 1$ are positive, the two gamma incomplete functions $\gamma_\mathrm{inc} (n + 2 \underline{\alpha}T, -x_1 x_2)$ and $   \gamma_\mathrm{inc}  (n+ 1, -x_1 x_2)$ are identically null if $x_1=k=0$. Consequently, the at-the-money (ATM) price of a European \blue{call} option in the BG model reduces to:   
    % \eqref{eq:bg-series-coefficients} are positive, the series in \eqref{eq:bg-series-digital} vanish for $k=0$.
    % At the money, i.e. $k=0$, the value of an European option is then:
    \begin{equation}
        \label{eq:atm-bg}
        C_\mathrm{Eur}^\BG
        =
        Ke^{-rT}\left(e^{-\zeta_\mathrm{BG}}+h(\lambda_+)-h(\lambda_+-1)-1\right).
    \end{equation}
    In appendix \ref{appendix:atm-bg-kt}, we show that this result is the same as the one obtained in \cite[eq. 8.10]{kuchlerBG}.
\end{remark}

%%%%%%%%%%%
\subsection{Variance Gamma}
%%%%%%%%%%%%

The VG process of \cite{Madan98} is a BG process with the restriction $\alpha_+=\alpha_-=\alpha$. The pricing formulas from proposition \ref{prop:series-digital-BG} remain valid, with $\underline{\alpha}=\alpha$ and $\overline{\alpha}=0$, and with the traditional VG parametrization:
\begin{equation}\label{VG_notations}
    \left\{
    \begin{aligned}
        & \sigma^2 = \frac{2\alpha}{\lambda_+\lambda_-}, \\
        & \theta = \frac{\alpha}{\lambda_+} - \frac{\alpha}{\lambda_-}, \\
        & \nu = \frac{1}{\alpha}.
    \end{aligned}
    \right.
\end{equation}
When furthermore $\lambda_+=\lambda_-=\lambda$, then $\theta=0$ and the VG model is symmetric (see \cite{Madan90}); replacing in \eqref{eq:h_function} and using Legendre's duplication formula yields:
\begin{equation}
    h(x)
    =
    % \frac{    2^{2\alpha T -1}\Gamma(1/2 + {\alpha}T)}{\sqrt{\pi}\Gamma(1+\alpha T)} 
    \frac{ \Gamma(\frac{1}{2} + \frac{T}{\nu})}{2\sqrt{\pi}\Gamma(1+\frac{T}{\nu})}
    ~{_2}F_1 \left( \frac{2T}{\nu},1 ; 1+\frac{T}{\nu} ; \frac{x}{\lambda}\right)
\end{equation}
and the ATM price \eqref{eq:atm-bg} is:
% \begin{multline}
%         \label{eq:atm-symvg}
%         C_\mathrm{Eur}^\mathrm{sVG}
%         =
%         Ke^{-rT}
%         \times \left(e^{-\zeta_\mathrm{sVG}T}-1\right.\\
%         \left.
%         +
%         \frac{    2^{2\alpha T -1}\Gamma(1/2 + {\alpha}T)}{\sqrt{\pi}\Gamma(1+\alpha T)}
%         \left(
%         {_2}F_1 \left( 2\alpha T,1 ; 1+\alpha T ; \frac{1}{2}\right)
%         -{_2}F_1 \left( 2\alpha T,1 ; 1+\alpha T ; 1-\frac{1}{\lambda}\right)\right)
%         \right).
% \end{multline}
% where:
% \begin{equation}
%     \zeta_\mathrm{sVG} := - \alpha\ln \left[\frac{\lambda^2}{\lambda^2-1}\right].
% \end{equation}
\begin{multline}
        \label{eq:atm-symvg}
        C_\mathrm{Eur}^\mathrm{sVG}
        =
        Ke^{-rT}
        \times \left(
        \left(1-\frac{\sigma^2\nu}{2}\right)^{-\frac{T}{\nu}}-1\right.\\
        \left.
        +
        \frac{ \Gamma(1/2 + \frac{T}{\nu})}{2\sqrt{\pi}\Gamma(1+\frac{T}{\nu})}
        \left(
        {_2}F_1 \left( \frac{2T}{\nu},1 ; 1+\frac{T}{\nu} ; \frac{1}{2}\right)
        -{_2}F_1 \left( \frac{2T}{\nu},1 ; 1+\frac{T}{\nu} ; \frac{1}{2}-\frac{\sigma\sqrt{\nu}}{2^{3/2}}\right)\right)
        \right).
\end{multline}
\blue{
Furthermore, using respectively the identities \cite[eq. 15.4.28]{DLMF}:
\begin{multline}
    {_2}F_1 \left( 2T/\nu,1 ; 1+T/\nu , 1/2-\sigma\sqrt{\nu}/(2^{3/2})\right)
     =
    \frac{\Gamma(1+T/\nu)\sqrt{\pi}}{\Gamma(1/2+T/\nu)}(1-\sigma^2\nu/2)^{-T/\nu}
    \\
    -\sqrt{2}\sigma T/\sqrt{\nu}~
    {_2}F_1 \left( 1/2+T/\nu,1 ; 3/2 , \sigma^2\nu/2\right)
\end{multline}
and \cite[eq. 15.8.25]{DLMF}
\begin{equation}
    {_2}F_1 \left( \frac{2T}{\nu},1 ; 1+\frac{T}{\nu} ; \frac{1}{2}\right)
    =
    \sqrt{\pi}\frac{\Gamma(1+T/\nu)}{\Gamma(1/2+T/\nu)},
\end{equation}
formula \eqref{eq:atm-symvg} can be simplified, using simple algebra, into:
\begin{equation}\label{eq:navas2}
    C_\mathrm{Eur}^\mathrm{sVG}
        =
        \frac{Ke^{-rT}}{2\Gamma(T/\nu)}
        \sum_{n=0}^\infty
        \frac{\Gamma((n+1)/2+T/\nu)}{\Gamma(n/2+1)}
        \left(\sigma\sqrt{\nu/2}\right)^n
        ,
\end{equation}
recovering the formula already obtained in \cite{aguilar2020some}.
}
% \blue{CAN BE SIMPLIFIED to recover the series formula from my old IJTAD paper cf GNV mail and add to acknowledgments}

% %%%%%%%%%%%
% \subsection{FMLS}
% %%%%%%%%%%%%

% %%%%%%%%%%%
% \subsection{BSM?}
%%%%%%%%%%%%

%%%%%%%%%%%%%%%%%%%%%%%%%%%%%%%%
% \section{Case of one sided tempered stable processes}
%%%%%%%%%%%%%%%%%%%%%%%%%%%%%%%%%

% \subsection{Option pricing under GTS distribution}

%%%%%%%%%%%%%%%%%%%%%%%%%%%%%%%
%%%%%%%%%%%%%%%%%%%%%%%%%%%%%%%
%%%%%%%%%%%%%%%%%%%%%%%%%%%%%%%

\section{Option pricing under one-sided tempered stable processes}
\label{sec:opt-pricing-one-sided}

%%%%%%%%%%%%%%%%%%%%%%%%%%%%%%%%%%
Let us now focus on another important class of TS processes, namely one-sided TS processes, whose L\'evy measure is either supported by the positive or negative real axis only. This class is actually important  \blue{for financial applications notably in the negative jumps case, as it} contains several models used for option pricing or credit risk purpose for instance in \cite{carr2003finite,MadanSchoutens08}.
We will see that for one-sided TS distributions, option prices can be expressed as series over a single index, which constitutes an improvement when compared to the series obtained in \cite{AguilarKirkbyResidue} where 3 summation indices are needed, and will allow for an accelerated pricing convergence. 

Throughout this section, the underlying log-price process $(X_t)_{t\geqslant 0}$ is therefore assumed to be a $\blue{\TS_-}(\alpha,\beta,\lambda)$ process. We still keep the definition $a:=-\alpha\Gamma(-\beta)$ and we adapt the definition of $\gamma$ by setting now $\gamma:=a\lambda^\beta$. 
The moneyness is now defined by $k_{\blue{\TS_-}}:=\log (S_0/K)+(r-q+\zeta_{\blue{\TS_-}})T$ where
\begin{equation}
    \zeta_{\blue{\TS_-}} := -\ln \varphi_{\blue{\TS_-}} (-\mathrm{i})
    =-a((\lambda-1)^\beta - \lambda^\beta)
\end{equation}
In the rest of this section, we will still denote $k_{\blue{\TS_-}}$ by $k$ and $\zeta_{\blue{\TS_-}}$ by $\zeta$ , on order to preserve the \blue{simplicity of notation}.

%%%%%%%%%%%
\subsection{Mellin-Barnes representation}
%%%%%%%%%%%%
As done for double-sided TS distributions, we first introduce Mellin-Barnes representations for digital option prices.

\begin{proposition}
    \label{prop:-cn-an-bg-mellin}
    The values of a CN and of an AN call with moneyness $k>0$ admit the MB representations:
    \begin{equation}
        \label{eq:digital-ts-one-sided}
        \begin{cases}
        \displaystyle
            C_\CN^-
            =
            \blue{e^{-rT}-}
            \frac{e^{(\gamma-r)T}}{\beta}
            \int\limits_{\bm{c}+\mi\R^2} 
            \frac{\GP{\blue{U^{-}}\bs}{\blue{U^{-}}\bs}}{s_2}
            \blue{k}^{-s_2}
            (aT)^{s_1/\beta}\lambda^{s_1-s_2}
            \frac{\D\bs}{(\mi2\pi)^2},\\
            \displaystyle
            C_\AN^-
            =
            \blue{Ke^{k-(r+\zeta) T}-}
            \frac{Ke^{k+(\gamma-r)T }}{\beta}
            \int\limits_{\bm{c}+\mi\R^2} 
            \frac{\GP{\blue{U^{-}}\bs}{\blue{D^{-}}\bs}}{s_2}
            \blue{k}^{-s_2}
            (aT)^{s_1/\beta}(\lambda+1)^{s_1-s_2}
            \frac{\D\bs}{(\mi2\pi)^2},\\
        \end{cases}
    \end{equation}
    where:
    \begin{equation}
        \blue{U^{-}}
        :=
        \begin{pmatrix}
            -\beta^{-1} & 0\\
            -1 & 1
        \end{pmatrix},\quad
        \blue{D^{-}}
        :=
        \begin{pmatrix}
            -1 & 0
        \end{pmatrix}
        ,
    \end{equation}
    and $\bm{c}\in \{\bm{x}\in\R^2|~x_1<0,~x_2>0\}$.
\end{proposition}
\begin{proof}
    We only prove the result for the CN call option since the proof is identical for the AN call option (up to a modification of $\lambda$ and multiplicative constants). Combining \eqref{eq:digital-option-def} and \eqref{eq:density-Mellin-TS} yields:
    \begin{equation}
        C_\CN^{-} = \blue{e^{-rT}-}\frac{e^{(\gamma-r)T}}{\beta}
        \int\limits_{c_1+\mi \R} 
         \frac{\Gamma\left(-\frac{s}{\beta}\right)}{\Gamma(-s)}(aT)^{s/\beta}\lambda^{s}
         \Gamma_\mathrm{inc}(-s,\blue{k\lambda})
        \frac{\D s}{\mi 2 \pi}
    \end{equation}
    where $\Real(c_1)<0$ and $\Gamma_\mathrm{inc}$ is the upper incomplete Gamma function (see \cite[eq.8.2.2]{DLMF}). Introducing a Mellin-Barnes representation for $\Gamma_\mathrm{inc}$ (see \cite[eq. 8.6.11]{DLMF}) yields the desired result. 
\end{proof}

\subsection{Series expansion}
%For a negative moneyness $k<0$, a positive number $x>0$ and $n\in\N$, 
Let us start by defining the coefficient $c_n^{\blue{-}}(x_1;x_2)$ as follows:
for a tuple $(x_1,x_2)\in (-\infty,0)\times (0,\infty)$ and $n\in\N$, we
let:
\begin{equation}
    c_n^{\blue{-}}(x_1;x_2) := x_2^{\beta n}\Gamma_{\green{\mathrm{inc}}}(-n\beta,\blue{x_1x_2}).
\end{equation}

\begin{proposition}
    \label{prop:digital-series-one-sided-ts}
    The values of a CN and of an AN call $C_\CN^{\blue{-}}$ and $C_\AN^{\blue{-}}$ with moneyness $k>0$ admits the series representation:
        \begin{equation}
        \begin{cases}
            \displaystyle
            C_\CN^{\blue{-}} = e^{-rT}\left(\blue{1-e^{\gamma T}}\sum_{n=0}^\infty\frac{(-1)^{n}(aT)^{n}}{n!\Gamma(-n\beta)}c_n^{\blue{-}}(k;\lambda)\right),\\
            \displaystyle
            C_\AN^{\blue{-}} := Ke^{k-rT}
            \left(\blue{e^{-\zeta T}-e^{\gamma T}}
            \sum_{n=0}^\infty\frac{(-1)^{n}(aT)^{n}}{n!\Gamma(-n\beta)}c_n^{\blue{-}}(k;\lambda\blue{+}1)\right).
        \end{cases}
        \end{equation}
\end{proposition}

\begin{proof}
    % Again, we only do the proof for the cash-or-nothing option. In the formalism of \cite{zhdanov1998}and \cite[appendix C]{aguilarkirkby2022robust}, we have $\bm{\Delta} = (-\beta^{-1},1)$. It can be shown that the two set of poles in the admissible spaces are obtained by solving the systems:
    The characteristic vector $\bm{\Delta}$ associated to the 
        MB integral \eqref{eq:digital-ts-one-sided} is $\bm{\Delta}=(-\beta^{-1},1)$.
        % The multidimensional Jordan residue lemma \cite{passare1994multidimensional} 
        % allows to express 
        \blue{Following \eqref{eq:residue_thm}, we can express} the integral
        \eqref{eq:digital-ts-one-sided} as a sum of residues associated to the poles of the integrand located in the admissible subspace $\Pi_{\bm{\Delta}}$ as defined in \eqref{eq:Pi_def}.
    % Using the formalism of \cite{passare1994multidimensional, zhdanov1998} in the $\mathbb{C}^2$ case, we denote the characteristic vector associated to the MB integrals \eqref{eq:digital-ts-one-sided} by $\bm{\Delta}=(-\beta^{-1},1)$.
    %     The multidimensional Jordan residue lemma \cite{passare1994multidimensional} allows to express \eqref{eq:digital-ts-one-sided} as a sum of residues associated to the poles of the integrand located in the subspace
    %     $\{ s\in\mathbb{C}^2|~ \Re[\langle \bm{\Delta} , s \rangle] <  \langle \bm{\Delta} , \bm{c} \rangle \}$. 
        In this subspace, there are two sets of poles that are solutions of the two linear systems:
        \begin{equation}
            \begin{cases}
                U^{\blue{-}}\bm{s} = -\bm{n},\\
                U_{\{1\}}^{\blue{-}}\bm{s} = -n_1~\text{and}~\langle (0,1),\bm{s}\rangle = 0,
            \end{cases}
        \end{equation}
    % \begin{enumerate}
    %     \item $U^+\bm{s} = -\bm{n}$,
    %     \item $U_{\{1\}}^+\bm{s} = -n_1$ and $s_2 = 0$,
    % \end{enumerate}
    where $\bm{n}:=(n_1,n_2)\in\N^2$. 
    % The poles are then explicitly localized in $\bm{s} = (n_1\beta,-n_2+n_1\beta)$ and $\bm{s} = (n_1\beta,0)$. 
    %Computing the residue at these locations yields:
    These two systems have unique solutions, respectively given by:
    \begin{equation}
        \begin{cases}
            (\mathrm{P}1)_{\TS_-}
             \quad \bm{s} = (n_1\beta,-n_2+n_1\beta),\\
            (\mathrm{P}2)_{\TS_-}
            \quad
            \bm{s} = (n_1\beta,0).\\
        \end{cases}
    \end{equation}
    Computing and summing the residues at these poles yields, in the CN case:
    \begin{equation}
        C_\CN^- = \blue{e^{-rT}}\left(1-
        e^{\gamma T}\left[\sum_{\bm{n}\in\N^2}\frac{(-1)^{n_{1+2}}(aT)^{n_1}\lambda^{n_2}\blue{k}^{-n_1\beta+n_2}}{n_{1,2}!\Gamma(-n_1\beta)(n_1\beta-n_2)}
        +
        \sum_{n_1\in\N}
        \frac{(-1)^{n_{1}}}{n_{1}!}(aT)^{n_1}\lambda^{\beta n_1}
        \right]\right)
    \end{equation}
    that can be further simplified to get:
    \begin{equation}
        C_\CN^- = \blue{\blue{e^{-rT}}}\left(\blue{1-}
        e^{\gamma T}
        \sum_{n_1\in\N}
        \frac{(-1)^{n_{1}}(aT)^{n_1}\lambda^{\beta n_1}}{n_{1}!\Gamma(-n_1\beta)}
        \left[
        \Gamma(-n_1\beta)
        -
        \sum_{n_2\in\N}\frac{(-1)^{n_{2}}(\blue{k}\lambda)^{-n_1\beta+n_2}}{n_{2}!\Gamma(-n_1\beta)(n_2-n_1\beta)}
        \right]\right).
    \end{equation}
    Recognizing the series expansion of the lower incomplete gamma function (see \cite[eq. 8.7.3]{DLMF}) yields the desired result. Similar computation can be performed for the AN option.
\end{proof}

% Using the payoff relation \eqref{eq:payoff_Eureopean}, the price of a European option with negative moneyness $k$ can be written as:
As usual, it follows from \eqref{eq:payoff_Eureopean} that the price of the European claim $C_\EUR^{\blue{-}}$ can be written as $C_\text{Eur}^{\blue{-}} = C_\text{AN}^{\blue{-}} - KC_\text{CN}^{\blue{-}}$ and we have:
% The price of a European option with negative moneyness can be then recovered with the payoff relation \eqref{eq:payoff_Eureopean} and we have:
\begin{equation}
    C_{\EUR}^{\blue{-}} = 
    \blue{S_0e^{-qT}-Ke^{-rT}}
    -Ke^{(\gamma-r)T}
    \sum_{n\in\N}
    \frac{(-1)^{n}(aT)^{n}}{n!\Gamma(-n\beta)}c_{\EUR,n}^{\blue{-}}(k;\lambda)
\end{equation}
where $c_{\EUR,n}^{\blue{-}}(k;\lambda) := e^kc_n^{\blue{-}}(k;\blue{\lambda+1})-c_n^{\blue{-}}(k;\lambda)$.

\section{Numerical results}
\label{sec:num-results}
In this section, we provide numerical checks as well as a detailed performance study of the Mellin pricing series we obtained for the general double-side TS model, and also for specific cases that are of practical importance as they recover popular models (bilateral Gamma, Variance Gamma, spectrally one-sided processes...); implementation is publicly available in the repository {\sffamily TS-Pricing} (written in {\sffamily Python}).
As previously mentioned, one of the major advantages of TS distributions is the availability of an explicit characteristic function, allowing for the use of Fourier based methods for option pricing, \textit{e.g.} Lewis-Lipton \cite{Lewis01}, Carr-Madan \cite{Carr99}, Hilbert \cite{phelan2019hilbert}, COS \cite{fang2009novel}, PROJ \cite{kirkby2015efficient} and making these methods ideal challengers to assess the accuracy and efficiency of our Mellin pricing series. 

Throughout this section, we will use the following set of parameters:
\begin{equation}
    \label{eq:numerical-params}
    \begin{cases}
        \alpha_+ = \blue{0.44},\\
        \beta_+ = 0.1 + e/10,\\
        \lambda_+ =\blue{1.4},\\
        \alpha_- = \blue{0.35} ,\\
        \beta_- = 0.5-\pi/100,\\
        \lambda_- = \blue{0.4}.
    \end{cases}
\end{equation}
% where the values of $\beta_\pm$ are set to be irrational. 
The risk-free rate $r$ and dividend yield $q$ will be set to the values $r=0.02$ and $q=0.05$. 
All the experiences have been conducted on a personal laptop with CPU Intel Core i5-1021u-1.60GHz and operating system Ubuntu 24.04.

%%%%
\subsection{Numerical checks}
%%%%
The first set of tests that we conduct is aimed at checking the validity of the pricing formulas we have obtained. We consider the most general case (i.e., the double-sided TS model), and we compare the European call prices obtained via the Mellin series in proposition \ref{prop:eur-series-gts} to the prices obtained via the PROJ method; of course the choice of the PROJ method is arbitrary and any other Fourier related pricing method could be chosen as a reference, however we deliberately choose to compare to PROJ prices because the accuracy and efficiency of this method over other Fourier methods has been well evidenced in the literature across all vanilla and exotic options (see e.g. \cite{kirkby2015efficient,kirkby2020asian,kirkby2017bermudan}). We perform these comparisons for various spot price $S_0$ and maturity $T$, and display the results in figure \ref{fig:different-S0}: it is clear that the series prices perfectly match the PROJ prices across all moneyness. Let us note that we choose to perform the summations up to order $n_1=n_2=n_3=80$, in the series prices but, actually and as will be demonstrated in the next subsections, a truncation to a far lower order would ensure excellent numerical accuracy in most situations.

% for different values for different moneyness (varying the initial process value $S_0$ and the option maturity $T$). In figure \ref{fig:different-S0}, we compare our expansions to PROJ method for the double-sided model for European options. The graph show a  perfect matching between the two procedure. 

\begin{figure}
    \centering
    % Première ligne
    \begin{subfigure}[b]{0.48\textwidth}
        \centering
        \includegraphics[width=\textwidth]{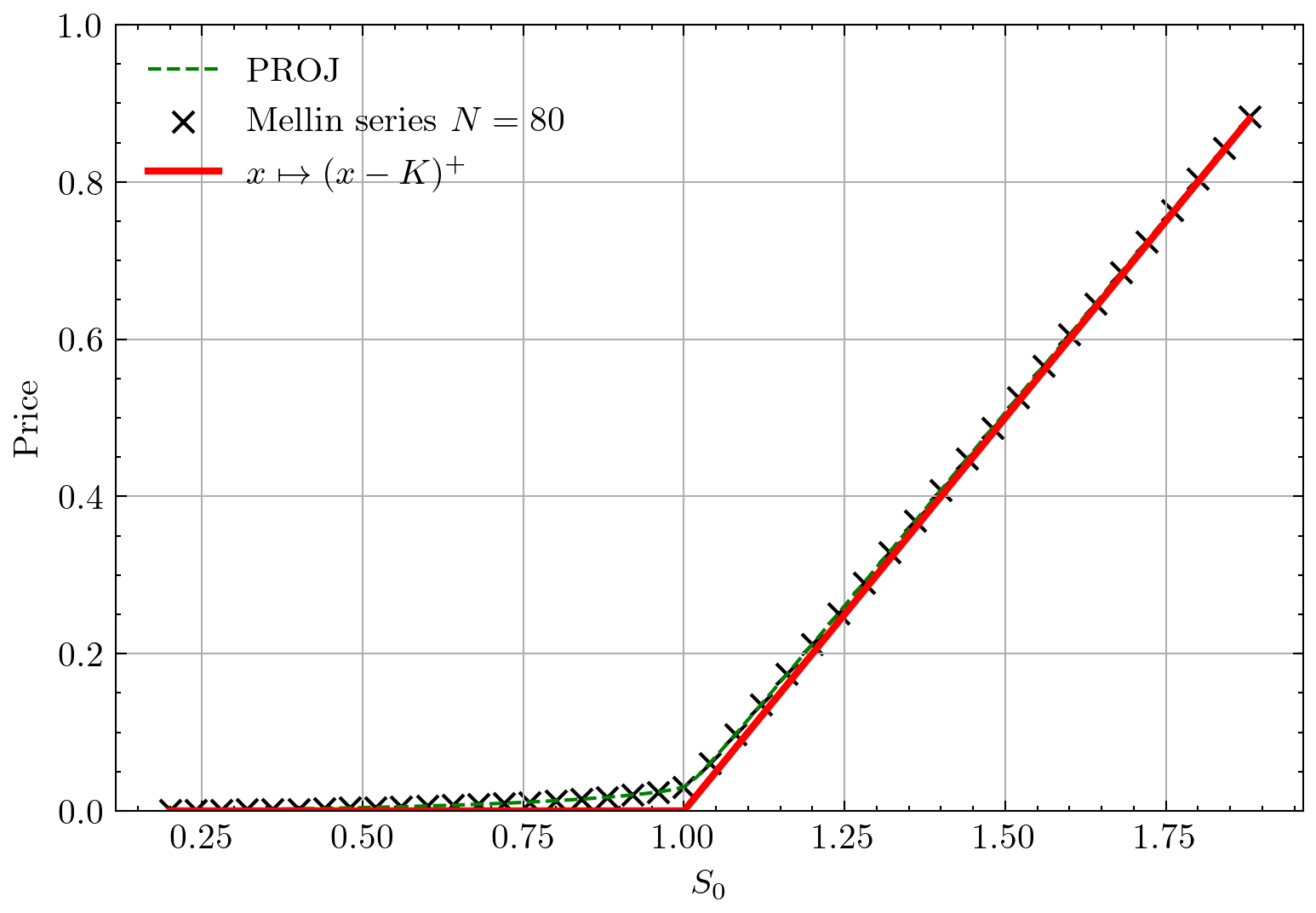}
        \caption{$T=0.05$}
    \end{subfigure}
    \hfill
    \begin{subfigure}[b]{0.48\textwidth}
        \centering
        \includegraphics[width=\textwidth]{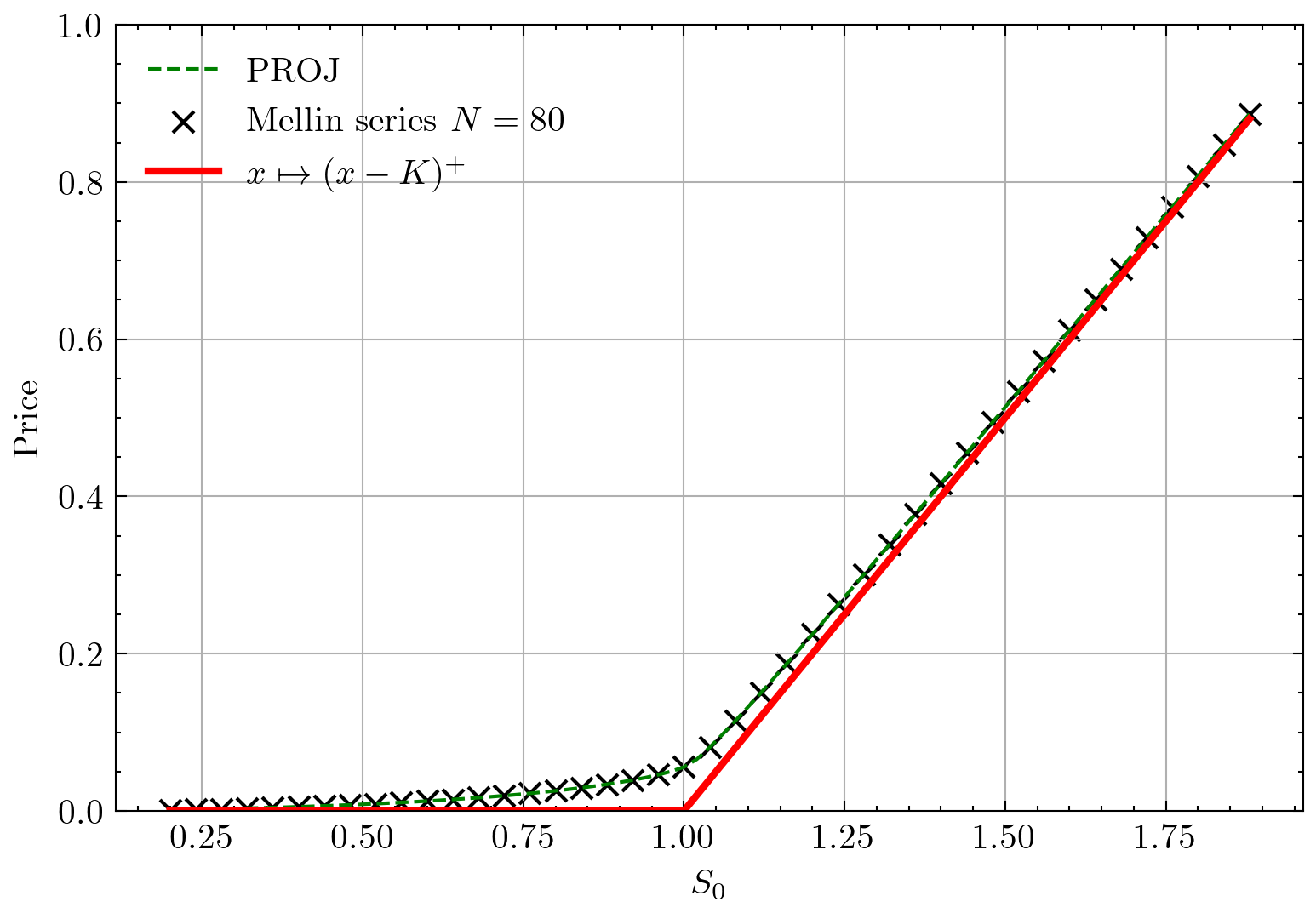}
        \caption{$T=0.1$}
    \end{subfigure}
    
    % Deuxième ligne
    \vspace{0.5cm} % Espace entre les lignes
    \begin{subfigure}[b]{0.48\textwidth}
        \centering
        \includegraphics[width=\textwidth]{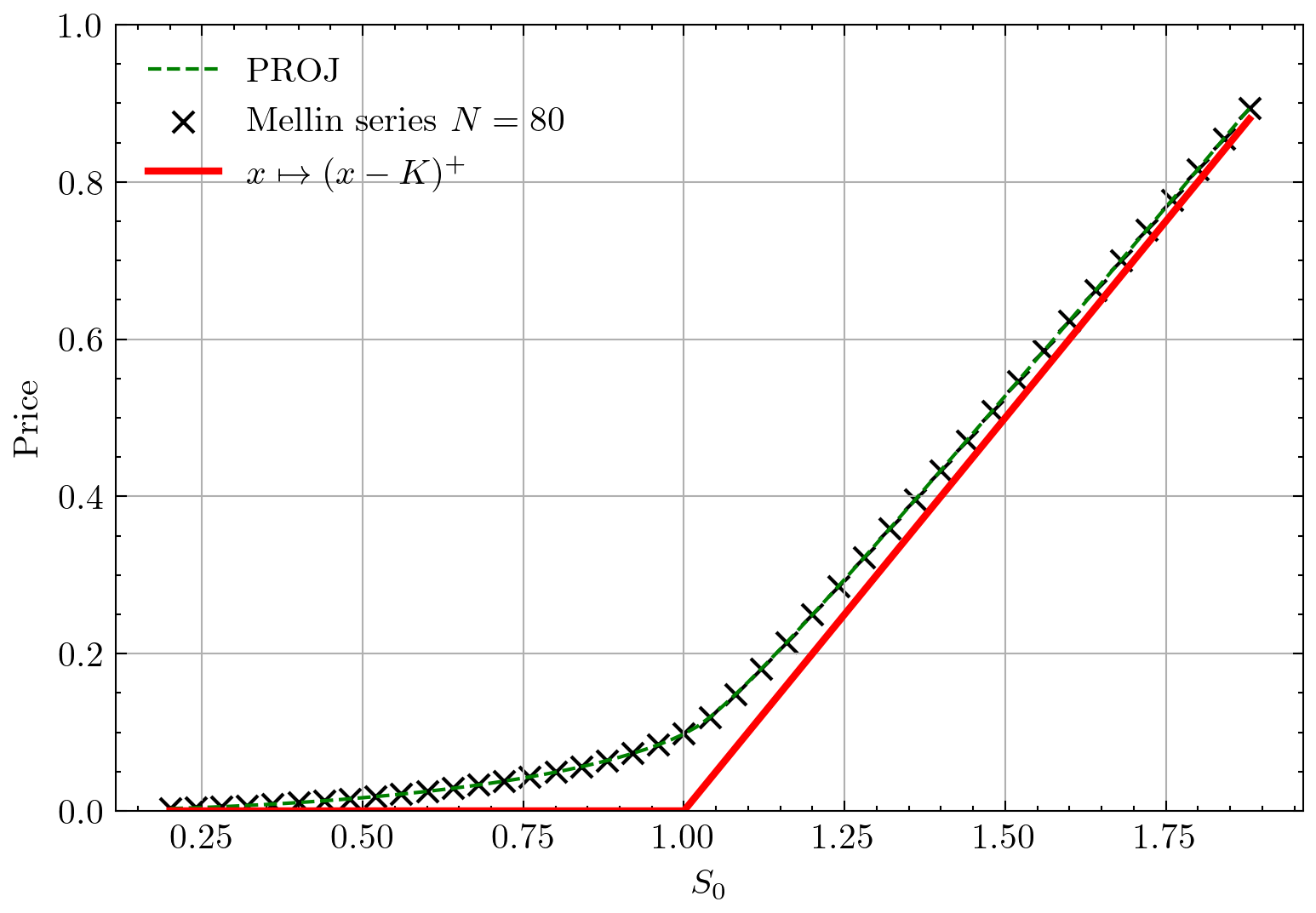}
        \caption{$T=0.2$}
    \end{subfigure}
    \hfill
    \begin{subfigure}[b]{0.48\textwidth}
        \centering
        \includegraphics[width=\textwidth]{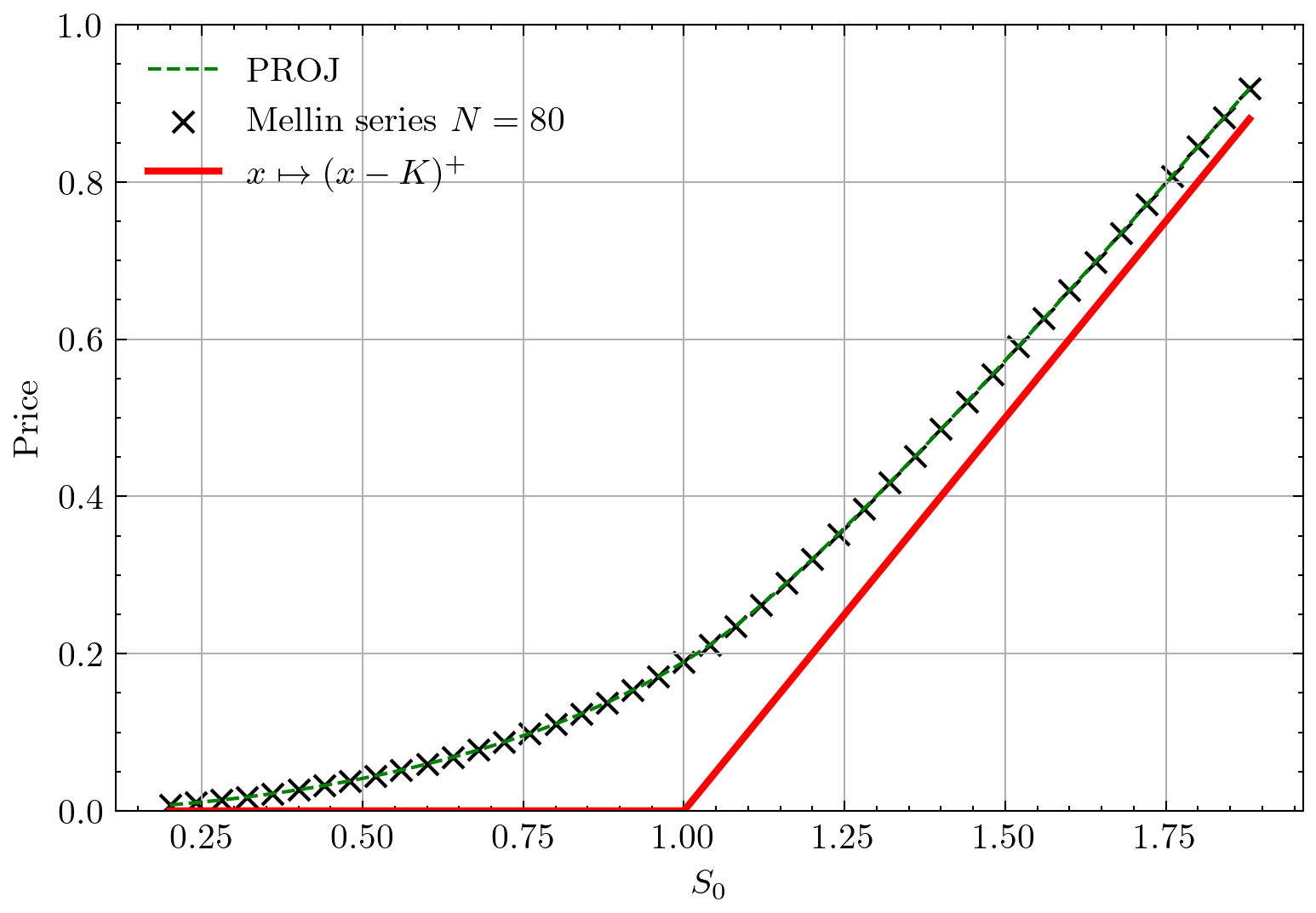}
        \caption{$T=0.5$}
    \end{subfigure}
    
    \caption{Comparison between the call pricing formula in the TS model (proposition \ref{prop:eur-series-gts}) and the  PROJ method. We let $S_0$ vary in $[0.2,1.9]$ and $K=1$, and we choose various maturities ($T\in\{0.05,0.1,0.2,0.5\}$). }
    \label{fig:different-S0}
\end{figure}

%%%%
\subsection{Performance study}
%%%%

% {
% \color{black}
In this section, we want to demonstrate the practical tractability and efficiency of the series formula we obtained, by displaying the computation time needed to attain a predetermined precision, and by comparing it with the PROJ computation time.
\blue{To set a reference price for given $S_0$, $K$ and $T$, we perform the Mellin expansion with a large summation index (typically $n_i$ up to $60$) and check that it agrees, up to machine precision, to the price computed with the PROJ method with large truncation parameter $\alpha>50$ and large discretization parameter $N_{\mathrm{PROJ}}$ (typically $2^{20}$); this price is then considered the reference price for this contract. Given another computed price (with different truncation etc.) for the same contract, the absolute precision is defined as the absolute value of the difference between this price and the reference price.}
We perform this study for the most general double-sided TS model and some of its particular cases (BG and one-sided TS), demonstrating a very competitive convergence speed when compared to state-of-the-art Fourier pricing; as we will see, the convergence is particularly accelerated in the BG and one-sided TS cases.

First, we note that \blue{our formulas are actually closed analytical formulas that can be directly implemented up to any chosen truncation order (common to the multiple series if any); this truncation may be seen as a tuning parameter required by the method, however it is in our opinion a far lesser concern than the determination of the tuning parameter(s), or hyperparameter(s), involved in other Fourier pricing technique, that have to be chosen in a clever way and can be delicate to select.} 
% a major advantage of our pricing method is the total absence of a parameter to set, while such a parameter (or even several) is in general mandatory to implement Fourier pricing methods and can be delicate to select. 
Indeed, in Fourier pricing, the tuning parameter can be a control parameter, conditioned by square integrability in the case of Carr-Madan pricing (it is known as a damping parameter in this case and is often denoted by $\alpha$), or be related to truncation intervals (for instance, in the case of the PROJ pricing it is related to the integration support size of the projection component). Even if some heuristic rules can be used to choose this parameter, notably some moment-based rules \green{(see \cite[eq.47]{kirkby2015efficient})}, these heuristics fail when it comes to high precision pricing, and manual increase of this parameter must therefore be performed, as illustrated in tables \ref{tab:time-double-sided-ts}, \ref{tab:time-bg} and \ref{tab:time-one-sided-ts}. In particular, in the double TS case (table \ref{tab:time-double-sided-ts}), it is clear that neither the ``default" value of $\alpha$ computed with the heuristic, nor several other attempts for this paremeter ($\alpha\in\{25,30,40\}$), allow to achieve a $10^{-9}$ \blue{absolute precision}, even by increasing the discretization parameter $N_{\text{PROJ}}$ (up to $N_{\text{PROJ}}=2^{20}$) and the number of functions in the Riesz basis - actually a manual increase to $\alpha=50$ is necessary to attain this precision; in contrast, the series formula \eqref{eur-series-gts} allows for a direct computation of prices without having to worry about such parametrization, which of course is a very desirable feature in practice as. \green{In tables \ref{tab:time-double-sided-ts}, \ref{tab:time-bg} and \ref{tab:time-one-sided-ts} we also display the minimal summation order $N$ such that the desired precision is reached by the Mellin series, and we see, in particular, that only 19 terms are needed to reach a $10^{-9}$ precision in the one-sided TS model.} We note, however, that, once the appropriate truncation parameter is chosen, the PROJ method remains approximately 100 times faster than the Mellin pricing formula \eqref{eur-series-gts} \green{in the general TS case} (due to the presence of 3 summation indices in this case); the series display nonetheless a very satisfactory performance, with only 0.09 second needed to reach a $10^{-9}$ precision.

% One of the main advantage of our pricing method is the absence of any parameter to set. For all the previous pricing techniques methods mentioned before, at least one hyperparameter is necessary. For PROJ method, the hyperparameter is the integration support size $\alpha$ of the projection component. 

% Even if some heuristic rules to chose this parameter according to the distribution moments, this heuristic fails to reach high precision and manual increase of this parameter is then needed.  This fact is illustrated in table \ref{tab:time-double-sided-ts}. It shows that for the default value of $\alpha$ (computed with the heuristic) and $\alpha\in\{25,30,40\}$, the method cannot reach precision highest than $10^{-8}$, whatever the size of the integration step (denoted by the $\times$ symbol). A manual increase to $\alpha=50$ is needed to reach $10^{-10}$ precision.

% A part from this, one note that PROJ is about 100 times faster that our series expansion. This is largely due to the fact that formula of proposition \ref{prop:eur-series-gts} is a sum over 3 indexes. 

\begin{table}
    \centering
    \scriptsize
    \begin{tabular}{p{1.2cm} p{1.2cm} p{1.2cm} p{1.4cm} p{1.2cm} p{1.2cm} p{1.2cm} p{1.2cm} p{1.2cm}}
    \toprule[2pt]
    Absolute precision & Mellin series& $\green{N}$  & PROJ (default $\alpha\green{= 20.33}$) & PROJ ($\alpha=25$) & PROJ ($\alpha=30$) & PROJ ($\alpha=40$) & PROJ ($\alpha=50$) \\
    \midrule
    \midrule
    1e-5 & 0.012719 &\green{33}& 0.000385 & 0.000327 & 0.000452 & 0.000559 & 0.000418 \\
    1e-6 & 0.022713 &\green{41}& $\times$ & 0.000629 & 0.000628 & 0.000882 & 0.000871 \\
    1e-7 & 0.043266 &\green{49}& $\times$ & $\times$ & 0.000939 & 0.000930 & 0.000872 \\
    1e-8 & 0.100635 &\green{58}& $\times$ & $\times$ & $\times$ & 0.000911 & 0.001459 \\
    1e-9 & 0.094498 &\green{66}& $\times$ & $\times$ & $\times$ & $\times$ & 0.001461 \\
    \bottomrule
    \end{tabular}
    \caption{Minimal computation time (in seconds) for PROJ and Mellin \green{European call option pricing} \eqref{eur-series-gts} in the double-sided TS model, for several predetermined levels of accuracy. We use $S_0=1$, $K=1.5$ and $T=1.2$. ``$\times$" means that the \blue{absolute} precision could not be reached, even for very large $N_{\text{PROJ}}=2^{20}$.}
    \label{tab:time-double-sided-ts}
\end{table}

Under bilateral Gamma and one-sided tempered stable processes, the pricing series for the European call are obtained from the digital option prices provided in propositions \ref{prop:series-digital-BG} and \ref{prop:digital-series-one-sided-ts} respectively, and are performed over a single index only (instead of three summation indices in the general TS case). This allows for considerably faster computation, as shown in tables \ref{tab:time-bg} and \ref{tab:time-one-sided-ts}. For example, in the bilateral Gamma case, if high precision ($10^{-8}\sim10^{-9}$) is desired, the Mellin pricing series is about 10 times faster than PROJ (even after choosing a hyperparameter $\alpha=50$); a similar overperformance is observed for the one-sided TS model.

% \begin{table}[h]
%     \centering
%     \scriptsize
%     \begin{tabular}{p{.1cm} p{1cm} p{1.4cm} p{1.2cm} p{1.2cm} p{1.2cm} p{1.2cm}}
%     \toprule[2pt]
%     & Precision & Mellin series & PROJ (default $\alpha$) & PROJ ($\alpha=25$) & PROJ ($\alpha=30$) & PROJ ($\alpha=40$) \\
%     \midrule
%     \midrule
%     \multirow{5}{*}{\rotatebox{90}{TS}} 
%     & 1e-5 & 0.000226 & 0.000958 & 0.000620 & 0.000626 & 0.001008 \\
%     & 1e-6 & 0.000221 & 0.000943 & $\times$ & 0.000653 & 0.000968 \\
%     & 1e-7 & 0.000218 & 0.001647 & $\times$ & $\times$ & 0.001592 \\
%     & 1e-8 & 0.000244 & $\times$ & $\times$ & $\times$ & $\times$ \\
%     & 1e-9 & 0.000236 & $\times$ & $\times$ & $\times$ & $\times$ \\
%     \midrule
%     \midrule
%     \multirow{5}{*}{\rotatebox{90}{TS one-sided}} 
%     & 1e-5 & 0.000086 & 0.000422 & 0.000780 & 0.000404 & 0.000395 \\
%     & 1e-6 & 0.000086 & 0.000396 & $\times$ & 0.000401 & 0.000477 \\
%     & 1e-7 & 0.000084 & 0.000521 & $\times$ & $\times$ & 0.000476 \\
%     & 1e-8 & 0.000085 & $\times$ & $\times$ & $\times$ & $\times$ \\
%     & 1e-9 & 0.000086 & $\times$ & $\times$ & $\times$ & $\times$ \\
%     \bottomrule
%     \end{tabular}
%     \caption{Temps de calcul minimal (en secondes) pour PROJ et Mellin dans les modèles TS et TS one-sided.}
%     \label{tab:time-rotated} 
% \end{table}

\begin{table}
    \centering
    \scriptsize
    \begin{tabular}{p{1.2cm} p{1.2cm} p{1.2cm} p{1.4cm} p{1.4cm} p{1.4cm} p{1.2cm} p{1.4cm} }
    \toprule[2pt]
    Asbolute precision & Mellin series &$\green{N}$  & PROJ (default $\alpha\green{= 35.85}$) & PROJ ($\alpha=25$) & PROJ ($\alpha=30$) & PROJ ($\alpha=40$) & PROJ ($\alpha=50$) \\
    \midrule
    \midrule
    1e-5 & 0.000205 &\green{32}& 0.000964 & 0.000632 & 0.000622 & 0.001008 & 0.001161 \\
1e-6 & 0.000215 &\green{41}& 0.000939 & $\times$ & 0.000709 & 0.001129 & 0.001123 \\
1e-7 & 0.000215 &\green{49}& 0.001842 & $\times$  & $\times$  & 0.001593 & 0.001535 \\
1e-8 & 0.000232 &\green{58}& $\times$  & $\times$  & $\times$  & $\times$  & 0.001591 \\
1e-9 & 0.000240 &\green{67}& $\times$  & $\times$  & $\times$  & $\times$  & 0.002852 \\
    \bottomrule
    \end{tabular}
    \caption{Minimal computation time (in seconds) for PROJ and Mellin European \blue{call} option pricing in the bilateral Gamma model, using the digital Mellin series formula in proposition \ref{prop:digital-series-one-sided-ts}, for several predetermined levels of accuracy. We use $S_0=1$, $K=1.5$ and $T=1.2$. ``$\times$" means that the \blue{absolute} precision could not be reached, even for very large $N_{\text{PROJ}}=2^{20}$. }
    \label{tab:time-bg}
\end{table}

% \begin{table}
%     \centering
%     \scriptsize
%     \begin{tabular}{p{1.2cm} p{1.2cm} p{1.4cm} p{1.4cm} p{1.4cm} p{1.2cm} p{1.4cm} }
%     \toprule[2pt]
%     Precision & Mellin series  & PROJ (default $\alpha$) & PROJ ($\alpha=7.5$) & PROJ ($\alpha=8.5$) & PROJ ($\alpha=10$) & PROJ ($\alpha=12.5$) \\
%     \midrule
%     \midrule
%     1e-5 & 0.000086 & 0.000422 & 0.000780 & 0.000404 & 0.000395 & 0.000469 \\
%     1e-6 & 0.000086 & 0.000396 & $\times$ & 0.000401 & 0.000477 & 0.000475 \\
%     1e-7 & 0.000084 & 0.000521 & $\times$ & $\times$ & 0.000476 & 0.000471 \\
%     1e-8 & 0.000085 & $\times$ & $\times$ & $\times$ & $\times$ & 0.000483 \\
%     1e-9 & 0.000086 & $\times$ & $\times$ & $\times$ & $\times$ & 0.000621 \\
%     \bottomrule
%     \end{tabular}
%     \caption{Minimal computation time (in seconds) for PROJ and Mellin European option pricing in the one-sided tempered stable model, using the digital Mellin series formula in proposition \ref{prop:digital-series-one-sided-ts}, for several predetermined levels of accuracy. We use $S_0=1$, $K=1.5$ and $T=1.2$. ``$\times$" means that the precision could not be reached, even for very large $N_{\text{PROJ}}=2^{20}$.}
%     \label{tab:time-one-sided-ts}
% \end{table}
{\begin{table}\color{black}
    \centering
    \scriptsize
    \begin{tabular}{p{1.2cm} p{1.2cm} p{1.2cm} p{1.4cm} p{1.4cm} p{1.4cm} p{1.2cm} p{1.4cm} }
    \toprule[2pt]
    Absolute precision & Mellin series  &$\green{N}$  & PROJ (default $\alpha\green{= 19.81}$) & PROJ ($\alpha=20$) & PROJ ($\alpha=25$) & PROJ ($\alpha=35$) & PROJ ($\alpha=40$) \\
    \midrule
    \midrule
    1e-5 & 0.000153&\green{11} & 0.000684 & 0.000788 & 0.000672 & 0.001000 & 0.000995 \\
    1e-6 & 0.000171 &\green{13}& $\times$ & $\times$ & 0.000999 & 0.001016 & 0.001037 \\
    1e-7 & 0.000154&\green{13} & $\times$ & $\times$ & $\times$ & 0.000989 & 0.001133 \\
    1e-8 & 0.000174&\green{17} & $\times$ & $\times$ & $\times$ & 0.016174 & 0.001007 \\
    1e-9 & 0.000153 &\green{19}& $\times$ & $\times$ & $\times$ & $\times$ & 0.001426 \\
    \bottomrule
    \end{tabular}
    \caption{\color{black}Minimal computation time (in seconds) for PROJ and Mellin European \blue{call} option pricing in the one-sided TS model (negative jumps), using the digital Mellin series formula in proposition \ref{prop:digital-series-one-sided-ts}, for several predetermined levels of accuracy. We use $S_0=1$, $K=0.5$ and $T=1.2$. ``$\times$" means that the \blue{absolute} precision could not be reached, even for very large $N_{\text{PROJ}}=2^{20}$. }
    \label{tab:time-one-sided-ts}
\end{table}}
This accelerated convergence behavior in the bilateral Gamma and one-sided TS cases is also clearly visible in figure \ref{fig:comparison-3-series}, where relative error (in dotted lines) and computation time (in plain lines) are displayed as functions of the maximal summation index $N$, for both these models as well as for the most general double-sided TS model. In terms of error, we can see that the one-sided pricing formula reaches a $10^{-10}$ precision after 20 summation terms only, and roughly reaches machine precision after 25 terms, with a total computation time of $10^{-5}\sim 10^{-4}$ seconds only (in line with the results of table \ref{tab:time-one-sided-ts}). For the bilateral Gamma and double-sided TS cases, the precision increases progressively and, actually, almost linearly as a function of $N$, which constitutes a surprising and very interesting behavior, as it allows to reliably control the accuracy (at least from an empirical point of view) of the pricing results; we also note that, for these two models, we can reach a $10^{-10}$ precision for $N=80$ approximately, but the overall computation time need to reach this precision remains around $10^{-4}$ seconds for the bilateral Gamma model (which, again, is extremely fast), while it is of order $10^{-1}\sim1$ seconds for the general doublie-sided case. Of course and as already mentioned, this is because only $N$ terms are needed to perform the bilateral Gamma summations in proposition \ref{prop:series-digital-BG}, while $N^3$ terms are needed in the double-sided TS series \ref{eur-series-gts}.
% \blue{As a limitation of the method, however   Finally, we also note that for values of $\beta_pm$ close to the boundaries $0$ and $1$, the Mellin expansions have difficulty to converge. }

% A comparison between our three pricing formulas is shown in figure \ref{fig:comparison-3-series}.
% Regarding precision, we see that one-sided pricing formula allows to reach a $10^{-10}$ precision with only 20 summation terms. Under bilateral Gamma and double-sided process, the precision increases progressively and reach a $10^{-10}$ relative error in 80 terms. The computation time is almost constant and equal to approximately $10^{-5}\sim 10^{-4}$ seconds. For double-sided tempered stable process, the computation increases with the iteration number and reaches $10^{-1}\sim1$ seconds for 90 terms.  

\begin{figure}
    \centering
    \includegraphics[width=0.9\linewidth]{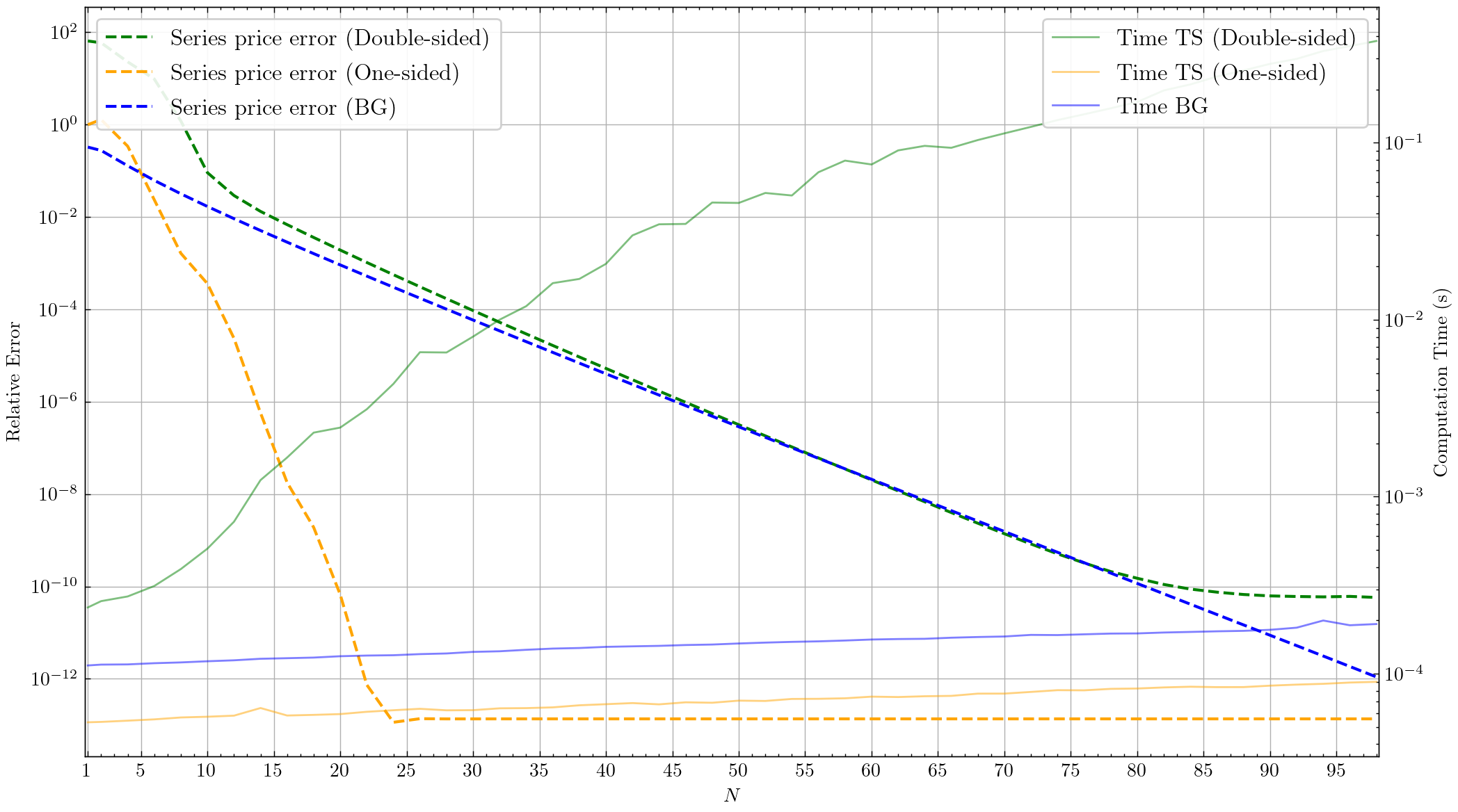}
    \caption{Relative error (dotted lines) and computation time (continuous lines) under one and double-sided tempered stable as well as bilateral Gamma process. We use $S_0=1$, $K=1.5$ and $T=1.2$.
    %\red{Machine precision = 16 digits. (donc $10^{-15}$ vu qu'il y a le premier chiffre),}
    }
    \label{fig:comparison-3-series}
\end{figure}

%%%%
\subsection{Comparison to other Fourier methods}
%%%%

Finally, we perform a final check by providing a short numerical comparison to other standard Fourier based pricing techniques, namely the PROJ \cite{kirkby2015efficient}, Lewis \cite{Lewis01}, Carr-Madan \cite{Carr99}, Gil-Pelaez \cite{GilPelaez} \blue{(whose implementation is detailed in Appendix \ref{sec:gil-Pelaez})} and Hilbert \cite{phelan2019hilbert} methods. We perform this check for several classical jump models covered by the TS family:

\begin{enumerate}
    \item TS, KoBoL and CGMY, that pertain to the double-sided TS family, and whose Mellin pricing series is given by \eqref{eur-series-gts};
    \item BG and VG, that are within the bilateral Gamma family (itself a subfamily of the TS class with $\beta_\pm=0$) and whose Mellin pricing series is given by proposition \ref{prop:series-digital-BG}; 
    \item TS$^-$, that is, the spectrally \blue{negative} TS distributions, that are also a special case of TS distributions (with a L\'evy measure supported by the \blue{negative} real axis, i.e. $\alpha_-=\lambda_-=\beta_-=0$) and whose Mellin pricing series is given by proposition \ref{prop:digital-series-one-sided-ts}.
\end{enumerate}
The results are summarized in table \ref{tab:comp-other-fourier} and feature very good agreement; without suprise, Lewis, PROJ and Mellin display best precision.

%  . All the prices are equals up to 2 decimal digits and the one that performs better are without surprise Lewis, PROJ and Mellin pricing series.   \footnote{\url{https://github.com/jkirkby3/fypy}}
% \red{\bfseries mentionner qu'on teste les modèles CGMY, KoboL et VG}

\begin{table}
    \centering
    \scriptsize
    \begin{tabular}{l l l l l l l}
    \toprule
    & \multicolumn{3}{c}{Double-sided TS} & One-sided TS & \multicolumn{2}{c}{Bilateral Gamma} \\
    \cmidrule(lr){2-4} \cmidrule(lr){5-5} \cmidrule(lr){6-7}
    Method & TS & KoBoL & CGMY & $\TS^+$ & BG & VG \\
    \midrule
    \midrule
    PROJ & 2.29678e-01 & 2.29685e-01 & 2.42660e-01 & 1.87698e-01 & 2.59919e-01 & 2.71774e-01 \\
    Lewis & 2.29685e-01 & 2.29685e-01 & 2.42660e-01 & 1.87698e-01 & 2.59919e-01 & 2.71774e-01 \\
    Carr-Madan & 2.29468e-01 & 2.29094e-01 & 2.42068e-01 & 1.87500e-01 & 2.59285e-01 & 2.71139e-01 \\
    Gil-Pelaez formula & 2.29685e-01 & 2.29685e-01 & 2.42660e-01 & 1.87698e-01 & 2.59929e-01 & 2.71774e-01 \\
    Hilbert & 2.29528e-01 & 2.29528e-01 & 2.42419e-01 & 1.87661e-01 & 2.59381e-01 & 2.71006e-01 \\
    Mellin series & 2.29685e-01 & 2.29775e-01 & 2.42660e-01 & 1.87698e-01 & 2.59919e-01 & 2.71774e-01 \\
    \bottomrule
    \end{tabular}
    \caption{\blue{Comparison of the prices retrieved} by Mellin series technique with other Fourier-related pricing techniques, for different standard models included in the TS class. We use $S_0=1$, $K=1.5$ and $T=1.2$. For Mellin series, we use $N=50$. Fourier parameters from the public PROJ repository\protect\footnote{\url{https://github.com/jkirkby3/fypy}}.}
    \label{tab:comp-other-fourier}
\end{table}

\blue{\subsection{Limitation of the method} As we have seen, Mellin expansion provides fast, accurate and simple method to price European options. Some limitations may however be discussed. First, the pricing formulas hold under assumption \ref{assum:beta-irra}; however, even if it may appear as a limitation, we showed in proposition \ref{prop:pointwise-convergence} that TS densities are continuous functions of their $\beta_\pm$ parameters. Hence, if $\beta_\pm$ is rational, one can approximate this density 
by adding small irrational perturbations $\varepsilon_\pm$ and then use proposition \ref{prop:series_density_GTS} with parameter $\beta_\pm+\varepsilon_\pm$.
% it is easy to construct an irrational sequence $\beta_{\pm,n}\to \beta_\pm$ and to apply proposition \ref{prop:series_density_GTS} to the corresponding sequence of densities: \ref{prop:pointwise-convergence} ensures that the limit will equal the value of the density with $\beta_\pm$ parameters. 
Similar approach can be carried out for prices, which is illustrated in figure \ref{fig:convergence}, where we approximate a European call price with $\beta_-=1/2$ by a sequence of European call prices with $\beta=1/2 - \pi/n$ and display the pricing error: we see that choosing 
$\varepsilon_- = \pi/10^{10}$ allows to obtain an approximated price that matches
the \green{reference} price up to machine precision.
% to
% for the corresponding prices, the same numerical approach can be adopted. 
% In figure \ref{fig:convergence},  we approximate the prices with parameter $\beta_-=1/2$ by a sequence of parameters of the form $\beta_{-,n}=1/2- \pi /n$; we can see that the price given by our Mellin series for $\beta_{-,n}$ converges without problem to the price given by the PROJ method for $\beta_{-}$. 
Another limitation we may mention is that, for $\beta_\pm$ close to the boundaries $(0,1)$, the expansion presents numerical instabilities whereas typically a numerical approach like PROJ method is more robust. This is also shown in figure \ref{fig:convergence}; we note, however, that the case $\beta_\pm = 0$ is actually the bilateral Gamma case, that has its own formulas (proposition \ref{prop:series-digital-BG}), and that the instability in the TS European call formula occurs only for $\beta_\pm$ close the boundaries of the $(0,1)$ interval.

% \begin{figure}[hbt]
%     \centering
    
%     \begin{subfigure}{0.45\textwidth}
%         \centering
%         \includegraphics[width=\linewidth]{SIAM Template (R1)/images/Price_conv.png}
%         \caption{Absolute error of TS European call price retrieved from Mellin Series \eqref{eur-series-gts} with $\beta_-=1/2-\pi/n$ vs. the price obtained by PROJ method with $\beta=1/2$.}
%         \label{fig:conv_mellin}
%     \end{subfigure}
%     \hfill
%     \begin{subfigure}{0.45\textwidth}
%         \centering
%         \includegraphics[width=\linewidth]{SIAM Template (R1)/images/varying_beta.png}
%         \caption{TS European call prices retrived by the Mellin series \eqref{eur-series-gts} for various $\beta_-$.}
%         \label{fig:conv_proj}
%     \end{subfigure}
    
%     \caption{Pricing results with $\beta_\pm$ rational or close to the boundaries $(0,1)$.}
%     \label{fig:convergence}
% \end{figure}

\begin{figure}[hbt]
    \centering
    
    \begin{subfigure}[t]{0.45\textwidth}
        \centering
        \includegraphics[width=\linewidth]{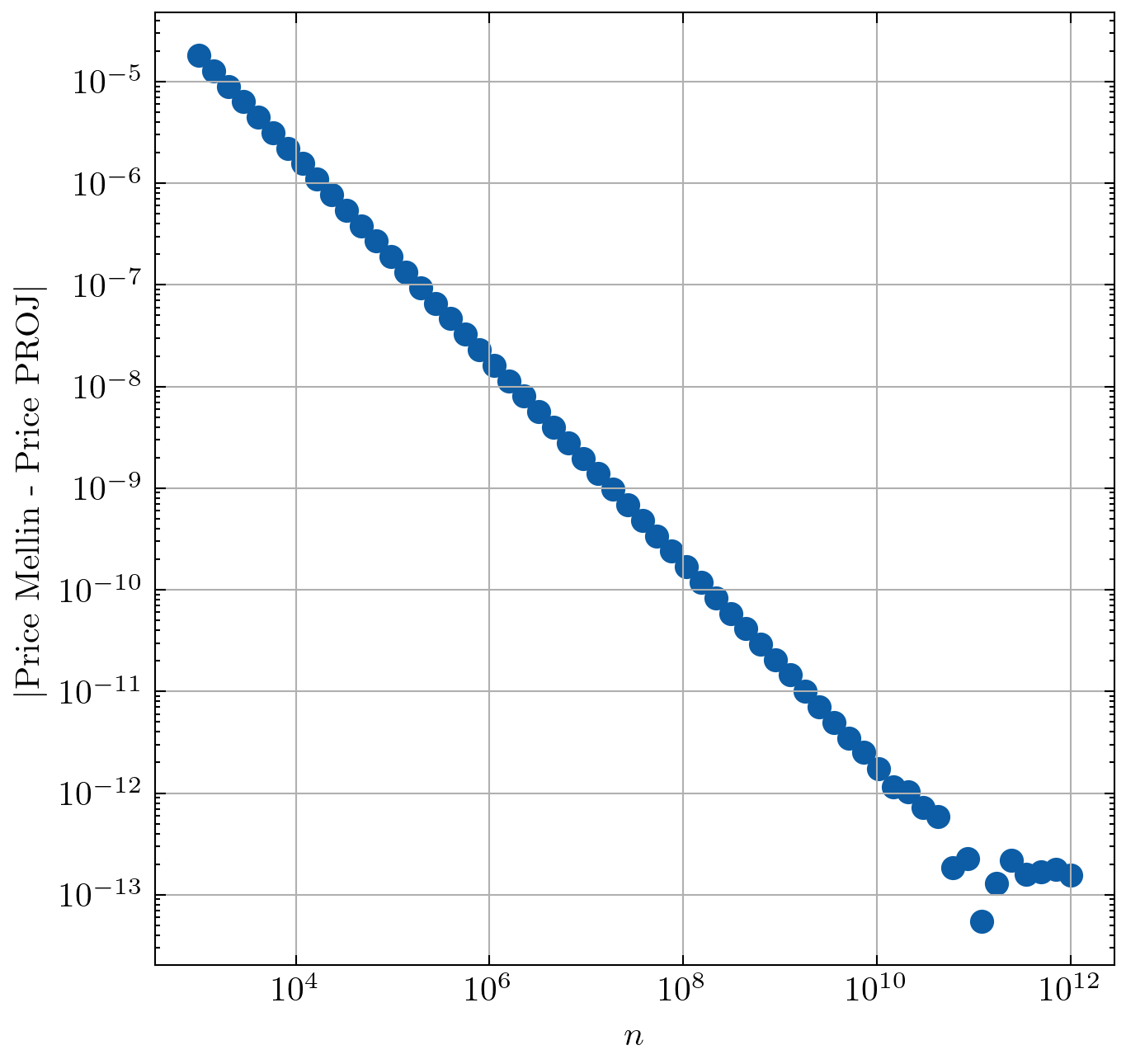}
        \caption{Absolute error of TS European call price retrieved from Mellin Series \eqref{eur-series-gts} with $\beta_-=1/2-\pi/n$ vs. the price obtained by PROJ method with $\beta=1/2$.}
        \label{fig:conv_mellin}
    \end{subfigure}
    \hfill
    \begin{subfigure}[t]{0.45\textwidth}
        \centering
        \includegraphics[width=\linewidth]{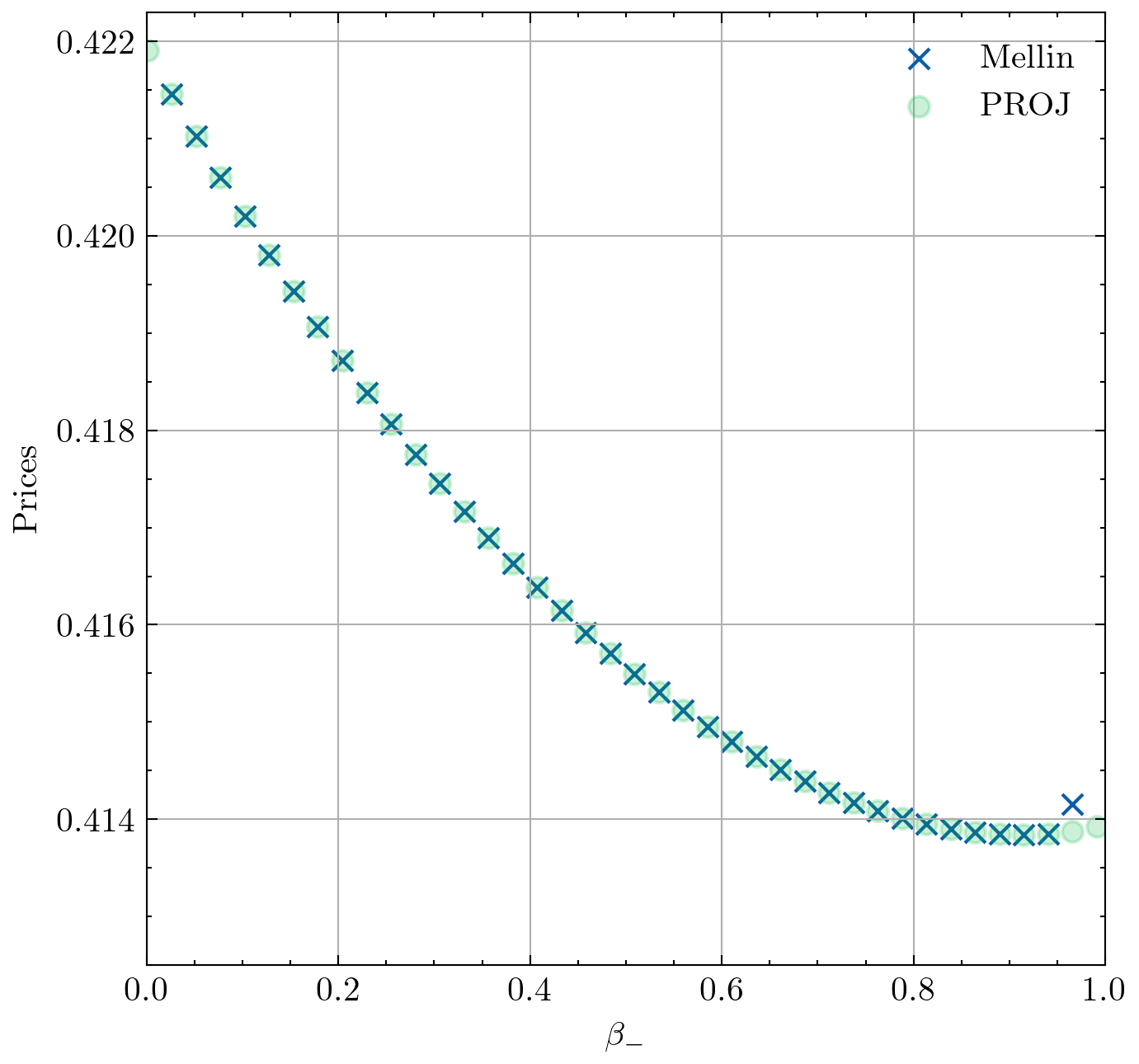}
        \caption{TS European call prices retrived by the Mellin series \eqref{eur-series-gts} for various $\beta_-$.}
        \label{fig:conv_proj}
    \end{subfigure}
    
    \caption{Pricing results with $\beta_-\in\mathbb Q$ (left) and for various $\beta_-\in (0,1)$ (right).}
    \label{fig:convergence}
\end{figure}
}

\section{Conclusion}
\label{sec:concl}
% \red{TO DO}

In this paper, we have proved Mellin-Barnes representations for TS density functions and derived their associated series expansions. We have also provided Mellin-Barnes representations for the prices of digital and European \blue{call} options when the log-price process is a TS process, and we were able to express these integrals as series expansions whose terms are straightforward to compute. These formulas become remarkably simple when considering particular cases such as bilateral Gamma or one-sided TS processes. 
Moreover, our technique \blue{does not involve a delicate selection of hyperparameter(s) like in standard Fourier pricing techniques, but only a choice of truncation order for the evaluation for the resulting analytical formulas.}
% has the major advantage of being hyperparameter free.

We have also presented a detailed numerical analysis of the pricing series expansions, including performance studies and comparison with other approaches. This analysis demonstrates that our method is competitive with state-of-the-art pricing techniques. All numerical implementations can be found in our public repository {\sffamily TS-Pricing} (written in {\sffamily Python}).

Future work should consider adapting our methodology to TS processes with $\beta_\pm \in(1,2)$, these processes having a different characteristic function \blue{and featuring infinite variation}. Other extensions should include a derivation of similar series expansions for additive processes and their particular cases \blue{such as Sato processes, that are self-similar additive processes introduced to the financial literature in \cite{carr2007self,eberlein2009sato}, and often feature a characteristic function that is known in closed form.}

\section{Acknowledgments}
\label{sec:ackn}

\blue{We thank two anonymous reviewers for insightful comments and suggestions, which helped us substantially improve the manuscript. We also thank Lorenzo Torricelli for bringing interesting references about positive TS densities to our attention, as well as Guillermo Navas-Palencia
 for his help regarding the simplification of \eqref{eq:atm-symvg} into \eqref{eq:navas2}.}

\bibliographystyle{siamplain}
\bibliography{biblio}

\appendix

% Last, we introduce a supplementary notation that will be useful for multiple discrete sums: given a multi-index $\bm{n}:=(n_1,...,n_p)$, we will denote for $1\leq p_1 < p_2\leq p$,  $(-1)^{n_{p_1+p_2}} := (-1)^{n_{p_1}+n_{p_1+1}+...+n_{p_2-1}+n_{p_2}}$ and $n_{p_1,p_2}! = \prod_{k=p_1}^{p_2}n_{k}!$.
\blue{
\section{Continuity property}
\begin{lemma}
    \label{lem:bound-integral}
    Let $\beta\in \green{(}0,1)$ and $\lambda>0$. For $u\in\R$, we have:
    \begin{equation}
        \int_0^\infty (\cos(ux)-1)x^{-\green{1}-\beta}e^{-\lambda x}\D x
        \leq
        g_\lambda(u)
        % \begin{cases}
        %     ...\quad \text{if}~|u|\geq 1\\
        %     ...\quad \text{if}~|u|\leq 1
        % \end{cases}
    \end{equation}
    where $g(\cdot;\beta,\lambda):\R\to\R$ is defined as:
    \green{
    \begin{equation}
        g(u;\beta,\lambda)
        =
        \begin{cases}
            \displaystyle
            0 \quad &\text{if}~|u|\leq 1,\\
            \displaystyle
            -\frac{11}{24(2-\beta)}e^{-\lambda}|u|^{\beta} \quad &\text{if}~|u|\geq 1.
        \end{cases}
    \end{equation}}
\end{lemma}
\begin{proof}
    We adapt to our case the proof of \cite[Lemma 4.13]{grabchak2016tempered}. 
    If $u=0$, the proof is easily checked. Let us assume that $u\neq 0$.
    Since $\cos(\cdot)\leq 1$, we have:
    \begin{equation}
        \int_0^\infty (\cos(ux)-1)x^{-\green{1}-\beta}e^{-\lambda x}\D x
        \leq
        \int_0^{\frac{1}{|u|}} 
        (\cos(ux)-1)x^{-\green{1}-\beta}e^{-\lambda x}\D x.
    \end{equation}
    where the integral in the right hand side converges since $(\cos(ux)-1)x^{-\green{1}-\beta}e^{-\lambda x}\sim_{x\to 0} - \frac{u^{2}}{2}x^{\green{1}-\beta}e^{-\lambda x}$ is integrable in $0$ as $\beta\in(0,1)$. 
    % further, the exponential term guarantees the integrability in $+\infty$.
    Since we have, from the alternating series estimation theorem, that
    \begin{equation}
        \forall s \in\R,\quad \cos(s)-1 \leq -\frac{s^2}{2}+\frac{s^4}{24},
    \end{equation}
    it follows that:
    \begin{equation}
            \forall |s|\leq 1,\quad \cos(s)-1 \leq -\frac{11}{24}s^2
    \end{equation}
    as $s^4<s^2$ on $(0,1)$. We then have the upper bound:
    \begin{equation}
    \label{eq:bound-cos-integral}
        \int_0^\infty (\cos(ux)-1)x^{-\green{1}-\beta}e^{-\lambda x}\D x
        \leq
        -\frac{11}{24}u^2\int_0^{\frac{1}{|u|}} 
        x^{\green{1}-\beta}e^{-\lambda x}\D x.
    \end{equation}
    % and, since $x^{-\beta}\geq |u|^{\beta}$ on $(0,1/|u|)$:
    % \begin{equation}
    %     \int_0^\infty (\cos(ux)-1)x^{-2-\beta}e^{-\lambda x}\D x
    %     \leq
    %     -\frac{11}{24}|u|^{2+\beta}\int_0^{\frac{1}{|u|}} 
    %     e^{-\lambda x}\D x
    %     =
    %     -\frac{11}{24 \lambda }|u|^{2+\beta}(1-e^{-\lambda/|u|}).
    % \end{equation}
    % Two cases arise:
    % \begin{itemize}
    %     \item if $|u|\geq 1$, we have $|u|^\beta\geq 1$ since $\beta \geq 0$ 
    %     and $|u|(1-e^{-\lambda/|u|}) \geq 1-e^{-\lambda}$ 
    %     since $x\mapsto x(1-e^{-\lambda/x})$ 
    %     is increasing on $[1,\infty)$. It then comes:
    %     \begin{equation}
    %         -\frac{11}{24\lambda}|u|^{2+\beta}(1-e^{-\lambda/|u|}) \leq
    %         -\frac{11}{24\lambda}|u|(1-e^{-\lambda}).
    %     \end{equation}
    %     \item if $|u|\leq 1$, we have $|u|^\beta\geq |u|$ since $\beta <1$ 
    %     and $(1-e^{-\lambda/|u|}) \geq 1-e^{-\lambda}$ 
    %     since $x\mapsto (1-e^{-\lambda/x})$ 
    %     is decreasing on $(0,1]$. It then comes:
    %     \begin{equation}
    %         -\frac{11}{24\lambda}|u|^{2+\beta}(1-e^{-\lambda/|u|}) \leq
    %         -\frac{11}{24\lambda}|u|^3(1-e^{-\lambda}).
    %     \end{equation}
    % \end{itemize}
    \green{
    Two cases arise:
    \begin{itemize}
        \item if $|u|\leq 1$, we simply bound \eqref{eq:bound-cos-integral} by 0;
        \item if $|u|\geq 1$, we have $e^{-\lambda x}\geq e^{-\lambda}$ for all $x\in (0,1/|u|)\subset(0,1)$ from which we bound:
        \begin{equation}
            -\frac{11}{24}u^2\int_0^{\frac{1}{|u|}} 
        x^{\green{1}-\beta}e^{-\lambda x}\D x \leq -\frac{11}{24}e^{-\lambda}u^2\int_0^{\frac{1}{|u|}} 
        x^{\green{1}-\beta}\D x = 
        -\frac{11}{24(2-\beta)}e^{-\lambda}|u|^{\beta}.
        \end{equation}
    \end{itemize}
    }
\end{proof}
}
\section{Computational details for the proofs of propositions \ref{prop:cn-series-gts} and \ref{prop:an-series-gts}}
\label{sec:pole}

In this appendix, we provide detailed computations of the residues associated to the poles (P1), (P2) and (P4) arising in the proofs of proposition \ref{prop:cn-series-gts} and \ref{prop:an-series-gts}.

% In appendices \ref{subsec:app-first-pole}, \ref{subsec:app-second-pole} and \ref{subsec:app-third-pole}, we provide detailed computations of the residues associated to the poles (P1), (P2) and (P4) in the proof of proposition \ref{prop:cn-series-gts}. In appendix \red{A.4 mais changer le label et la notation}

\subsection{Residues associated to (P1)}
\label{subsec:app-first-pole}

Let $S_1$ define the series of residues induced by the singularities in the subset (P1). We have:
\begin{multline}
S_1 :=
        e^{(\gamma-r)T}
        \sum_{\bm{n}\in\N^4}
        \frac{(-1)^{n_{1+4}}}{n_{1,4}!}
        \frac{
            \Gamma(1-n_1+\beta_+n_2+\beta_-n_3)
            \Gamma(n_1-\beta_-n_3)
        }{
        \Gamma(1+\beta_+n_2)
        \Gamma(-\beta_+n_2)
        \Gamma(-\beta_-n_3)
        (-n_1+\beta_+n_2+\beta_-n_3-n_4)
        }
        \\
        \times 
        (a_+T)^{n_2}(a_-T)^{n_3}\lambda_+^{n_4}
        \ubar{\lambda}^{n_1}
        (-k)^{n_1-\beta_+n_2-\beta_-n_3+n_4}.
    \end{multline}

Some algebra yields:
\begin{equation}
    \begin{aligned}
        S_1 
        &= \sum_{\bm{n}\in\N^4}
        \frac{(-1)^{n_{1+4}}}{n_{1,4}!}
        \frac{(-\beta_-n_3)_{n_1}}{(-n_1+\beta_+n_2 + \beta_-n_3 - n_4)}\\
        &\quad\quad\times\frac{\Gamma(1-n_1+\beta_+n_2+\beta_-n_3)}{\Gamma(1+\beta_+n_2)\Gamma(-\beta_+n_2)}
        (a_+T)^{n_2}(a_-T)^{n_3}  
        \ubar{\lambda}^{n_1} (-k)^{n_1-\beta_+n_2-\beta_-n_3+n_4}\lambda_+^{n_4}\\
        &= -\sum\limits_{\bm{n}\in\N^3}
        \frac{(-1)^{n_{1+3}}}{n_{1,3}!}
        (-\beta_-n_3)_{n_1}\frac{\Gamma(1-n_1+\beta_+n_2+\beta_-n_3)}{\Gamma(1+\beta_+n_2)\Gamma(-\beta_+n_2)}
        (a_+T)^{n_2}(a_-T)^{n_3}  
        \ubar{\lambda}^{n_1} \\
        &\quad\quad\times\lambda_+^{-n_1+\beta_+n_2+\beta_+n_3}\sum_{n_4\in\N}\frac{(-1)^{n_4}}{n_4!(n_1-\beta_+n_2 - \beta_-n_3 + n_4)}(-\lambda_+k)^{n_1-\beta_+n_2-\beta_-n_3+n_4}\\
        &= -\sum\limits_{\bm{n}\in\N^3}
        \frac{(-1)^{n_{1+3}}}{n_{1,3}!}
        (-\beta_-n_3)_{n_1}\frac{\Gamma(1-n_1+\beta_+n_2+\beta_-n_3)}{\Gamma(1+\beta_+n_2)\Gamma(-\beta_+n_2)}
        (a_+T)^{n_2}(a_-T)^{n_3}  
        \ubar{\lambda}^{n_1} \\
        &\quad\quad\times\lambda_+^{-n_1+\beta_+n_2+\beta_-n_3}\gamma_{\text{inc}}(n_1-\beta_+n_2-\beta_-n_3, -\lambda_+k)
    \end{aligned}
\end{equation}
where the last line is obtained by recognizing the series expansion for the incomplete gamma function (see \blue{\eqref{eq:series-incomp-series}}). It comes that:
    \begin{equation}
        S_1 = -e^{(\gamma-r)T}\sum_{\bm{n}\in\N^3}
        \frac{(-1)^{n_{1+3}}}{n_{1,3}!}
        (a_+T)^{n_2}(a_-T)^{n_3}
        c_{\bm{n}}^{(1)}(k;\lambda_+).
    \end{equation}

\subsection{Residues associated to (P2)}
\label{subsec:app-second-pole}
Let $S_2$ define the series of residues induced by the singularities in the subset (P2). We have:
\begin{multline}
        S_2 :=
        \sum_{\bm{n}\in\N^4}
        e^{(\gamma-r)T}\frac{(-1)^{n_{1+4}}}{n_{1,4}!}
        \frac{
            \Gamma(-1-n_1-\beta_+n_2-\beta_-n_3)
            \Gamma(1+n_1+\beta_+n_2)
        }{
        \Gamma(1+\beta_+n_2)
        \Gamma(-\beta_+n_2)
        \Gamma(-\beta_-n_3)
        (-1-n_1-n_4)
        }
        \\
        \times 
        (a_+T)^{n_2}(a_-T)^{n_3}\lambda_+^{n_4}
        \ubar{\lambda}^{1+n_1+\beta_+n_2+\beta_-n_3}
        (-k)^{1+n_1+n_4}.
    \end{multline}

Some algebra yields:
\begin{equation}
    \begin{aligned}
        S_2 &=  \sum_{\bm{n}\in\N^4}\frac{(-1)^{n_{1+4}}}{n_{1,4}!}
        \frac{(1+\beta_+n_2)_{n_1}}{-1-n_1-n_4}
        \frac{\Gamma(-1-n_1-\beta_+n_2-\beta_-n_3)}{\Gamma(-\beta_+ n_2)\Gamma(-\beta_-n_3)}
        (a_+T)^{n_2}(a_-T)^{n_3}\\
        &\quad\quad\times\ubar{\lambda}^{1+n_1+\beta_+n_2+\beta_-n_3} (-k)^{1+n_1+n_4}
        \lambda_+^{n_4}\\
        & = -\sum_{\bm{n}\in\N^3}\frac{(-1)^{n_{1+3}}}{n_{1,3}!}
        (1+\beta_+n_2)_{n_1}
        \frac{\Gamma(-1-n_1-\beta_+n_2-\beta_-n_3)}{\Gamma(-\beta_+ n_2)\Gamma(-\beta_-n_3)}
        (a_+T)^{n_2}(a_-T)^{n_3}\\
        &\quad\quad\times\ubar{\lambda}^{1+n_1+\beta_+n_2+\beta_-n_3} 
        \sum_{n_4\in\N}\frac{(-1)^{n_{4}}}{n_{4}!}(-k)^{1+n_1+n_4}\frac{1}{1+n_1+n_4}
        \lambda_+^{n_4}\\
        &=  -\sum_{\bm{n}\in\N^3}\frac{(-1)^{n_{1+3}}}{n_{1,3}!}
        (1+\beta_+n_2)_{n_1}
        \frac{\Gamma(-1-n_1-\beta_+n_2-\beta_-n_3)}{\Gamma(-\beta_+ n_2)\Gamma(-\beta_-n_3)}
        (a_+T)^{n_2}(a_-T)^{n_3}\\
        &\quad\quad\times\ubar{\lambda}^{1+n_1+\beta_+n_2+\beta_-n_3} 
        \lambda_+^{-1-n_1}
        \sum_{n_4\in \N}\frac{(-1)^{n_{4}}}{n_{4}!}(-\lambda k)^{1+n_1+n_4}\frac{1}{1+n_1+n_4}\\
        &=  -\sum_{\bm{n}\in\N^3}\frac{(-1)^{n_{1+3}}}{n_{1,3}!}
        (1+\beta_+n_2)_{n_1}
        \frac{\Gamma(-1-n_1-\beta_+n_2-\beta_-n_3)}{\Gamma(-\beta_+ n_2)\Gamma(-\beta_-n_3)}
        (a_+T)^{n_2}(a_-T)^{n_3}\\
        &\quad\quad\times\ubar{\lambda}^{1+n_1+\beta_+n_2+\beta_-n_3} 
        \lambda_+^{-1-n_1}\gamma_{\text{inc}}(1+n_1,-\lambda_+ k)
    \end{aligned}
\end{equation}
    where the last line is obtained by recognizing the series expansion for the lower incomplete gamma function (see \blue{\eqref{eq:series-incomp-series}}). 
    It comes that:
    \begin{equation}    
        S_2 = -e^{(\gamma-r)T}\sum_{\bm{n}\in\N^3}
        \frac{(-1)^{n_{1+3}}}{n_{1,3}!}
        (a_+T)^{n_2}(a_-T)^{n_3}
        c_{\bm{n}}^{(2)}(k;\lambda_+).
    \end{equation}

\subsection{Residues associated to (P4)}
\label{subsec:app-third-pole}
Let $S_4$ define the series of residues induced by the singularities in the subset (P4). We have:
\begin{multline}\label{S4_residue}
S_4 := \sum_{\bm{n}\in\N^3}
            e^{(\gamma-r)T}\frac{(-1)^{n_{1+3}}}{n_{1,3}!}
            \frac{
                \Gamma(n_1-\beta_-n_3)
                \Gamma(1+n_1+\beta_+n_2)
                \Gamma(-n_1-\beta_+n_2)
            }{
            \Gamma(1+\beta_+n_2)
            \Gamma(-\beta_+n_2)
            \Gamma(-\beta_-n_3)
            }\\
            \times 
            (a_+T)^{n_2}(a_-T)^{n_3}\lambda_+^{n_1+\beta_+n_2}
            \ubar{\lambda}^{-n_1+\beta_-n_3}.
        \end{multline}
        We first notice that this can be rewritten with Pochhamer's symbols:
        \begin{equation}
            S_4
            := 
            e^{(\gamma-r)T}\sum_{\bm{n}\in\N^3}
                        \frac{(-1)^{n_{2+3}}}{n_{1,3}!}
                        (-\beta_-n_3)_{n_1}
                        (a_+T)^{n_2}(a_-T)^{n_3}\lambda_+^{n_1+\beta_+n_2}
                        \ubar{\lambda}^{-n_1+\beta_-n_3}.
        \end{equation}
        since $\Gamma(-n_1-\beta_+n_2)/\Gamma(-\beta_+n_2) = (-1)^{n_1}/(1+\beta_+n_2)_{n_1}$. By first summing over $n_1$, one recognize the hypergeometric series $_1F_0(-\beta_+n_3;;\lambda_+/\underline{\lambda})$, whose sum is known in closed form ($_1F_0a;;z)=(1-z)^{-a}$) and we therefore have:
        \begin{equation}
            S_4
            := 
            e^{(\gamma-r)T}\sum_{(n_2,n_3)\in\N^2}
                        \frac{(-1)^{n_{2+3}}}{n_{2,3}!}
                        (a_+T)^{n_2}(a_-T)^{n_3}\lambda_+^{\beta_+n_2}
                        \ubar{\lambda}^{\beta_-n_3}
                        (1-\lambda_+/\ubar{\lambda})^{\beta_-n_3}
        \end{equation}
        Noticing that $(1-\lambda_+/\ubar{\lambda})^{\beta_+n_3} = \lambda_-^{\beta_-n_3}$ (definition \eqref{eq:constants-ts}) and that the two remaining series are simply exponential expansions, we have:
        \begin{equation}
            S_4 := e^{(\gamma-r)T}e^{-a_-T\lambda_-^{\beta_-}-a_+T\lambda_+^{\beta_+}}
            =e^{-rT}
            ,
        \end{equation}
        where the last equality stems from the definition of $\gamma$ in \eqref{eq:constants-ts}.
        % Recalling the definition of $\gamma$ in \eqref{eq:constants-ts}, the desired result follows and we have $S_4 = e^{-rT}$.

\subsection{Residues associated to (P4) in the AN case}
\label{subsec:app-fourth-pole}
In the AN case, the residues \eqref{S4_residue} induced by the (P4) subset become:
\begin{multline}
S_4^{\text{AN}} := \sum_{\bm{n}\in\N^3}
            Ke^{k+(\gamma-r)T}\frac{(-1)^{n_{1+3}}}{n_{1,3}!}
            \frac{
                \Gamma(n_1-\beta_-n_3)
                \Gamma(1+n_1+\beta_+n_2)
                \Gamma(-n_1-\beta_+n_2)
            }{
            \Gamma(1+\beta_+n_2)
            \Gamma(-\beta_+n_2)
            \Gamma(-\beta_-n_3)
            }\\
            \times 
            (a_+T)^{n_2}(a_-T)^{n_3}(\lambda_+-1)^{n_1+\beta_+n_2}
            \ubar{\lambda}^{-n_1+\beta_-n_3}.
        \end{multline}
    Performing the same computations as in section \ref{subsec:app-third-pole} yields:
    \begin{equation}
        S_4^{\text{AN}}
        := 
        Ke^{k+(\gamma-r)T}\sum_{(n_2,n_3)\in\N^2}
                    \frac{(-1)^{n_{2+3}}}{n_{2,3}!}
                    (a_+T)^{n_2}(a_-T)^{n_3}(\lambda_+-1)^{\beta_+n_2}
                    (\lambda_-+1)^{\beta_-n_3}.
    \end{equation}
    Recognizing the exponential series yields:
    \begin{equation}
        S_4^{\text{AN}}
        = 
        Ke^{k+(\gamma-r)T}e^{-a_+T(\lambda_+-1)^{\beta_+}-a_-T(\lambda_-+1)^{\beta_-}}
        = Ke^{k+(\zeta-r)T}
    \end{equation}
    where the last equatlity follows from the definition of $\gamma$ in \eqref{eq:constants-ts} and of the martingale adjustment $\zeta$ in \eqref{eq:def-zeta}.

\section{ATM option price under bilateral Gamma process}
    \label{appendix:atm-bg-kt}
    In \cite{kuchlerBG}, the authors give a closed formula for the ATM digital options in the BG model. Let us prove that the CN component in our ATM pricing formula \eqref{eq:atm-bg} is the same the one obtained by \cite{kuchlerBG} (same proof can be made for the AN option, and therefore for the European option). To this aim, we consider, as
    in \cite{kuchlerBG}, the risk-free rate and dividend yield to be $r=q=0$. With these settings, the formula obtained by \cite{kuchlerBG} for the CN call is:
    \begin{equation}
        \frac{\lambda_+^{\alpha_+}\lambda_-^{\alpha_-}\Gamma(\alpha_++\alpha_-)
        }{\Gamma(\alpha_+)\Gamma(1+\alpha_-)\lambda_+^{\alpha_++\alpha_-}}
        ~_2F_1(\alpha_++\alpha_-,\alpha_-;1+\alpha_-,-\lambda_-/\lambda_+).
    \end{equation}
    Introducing $\Lambda:=\lambda_-/\lambda_+$, this equation can be simplified in:
    \begin{equation}
        \frac{\Lambda^{\alpha_-}\Gamma(\alpha_++\alpha_-)
        ~_2F_1(\alpha_++\alpha_-,\alpha_-;1+\alpha_-,-\Lambda)}{\Gamma(\alpha_+)\Gamma(1+\alpha_-)}
    \end{equation}
    % and applying \cite[eq. 9.132.1]{Bateman} (with $\bm{\alpha} \leftrightarrow\alpha_-$, $\bm{\beta}\leftrightarrow\alpha_++\alpha_-$, $\bm{\gamma}\leftrightarrow1+\alpha_-$ and $\bm{z}\leftrightarrow-\Lambda$), 
    and applying \cite[eq. 15.8.3]{DLMF}
    we have:
   {\small
    \begin{multline}
        \label{eq:annex-interm-2f1}
       _2F_1(\alpha_++\alpha_-,\alpha_-;1+\alpha_-,-\Lambda)
       =
       \frac{\Gamma(\alpha_+)\Gamma(1+\alpha_-)}{\Gamma(\alpha_++\alpha_-)(1+\Lambda)^{\alpha_-}}
       ~_2F_1(\alpha_-,1-\alpha_+;1-\alpha_+;(1+\Lambda)^{-1})\\
       +\frac{\Gamma(-\alpha_+)\Gamma(1+\alpha_-)}{\Gamma(\alpha_-)\Gamma(1-\alpha_+)(1+\Lambda)^{\alpha_-+\alpha_+}}
       ~_2F_1(\alpha_++\alpha_-,1;1+\alpha_+;(1+\Lambda)^{-1}).
    \end{multline}}
    Substituting the identity:
    \begin{equation}
        _2F_1(\alpha_-,1-\alpha_+;1-\alpha_+;(1+\Lambda)^{-1})
        =
        ~_1F_0(\alpha_-,(1+\Lambda)^{-1})
        = \left(\frac{\Lambda}{1+\Lambda}\right)^{-\alpha_-}
    \end{equation}
    in \ref{eq:annex-interm-2f1} yields:
    \begin{multline}
       \frac{\Lambda^{\alpha_-}\Gamma(\alpha_+ + \alpha_-)}{\Gamma(1+\alpha_-)\Gamma(\alpha_+)}~_2F_1(\alpha_++\alpha_-,\alpha_-;1+\alpha_-,-\Lambda)
       =
       1\\
       +\frac{\Lambda^{\alpha_-}\Gamma(-\alpha_+)\Gamma(\alpha_++\alpha_-)}{\Gamma(\alpha_-)\Gamma(1-\alpha_+)\Gamma(\alpha_+)(1+\Lambda)^{\alpha_-+\alpha_+}}
       ~_2F_1(\alpha_++\alpha_-,1;1+\alpha_+;(1+\Lambda)^{-1}).
    \end{multline}
    Noticing that $\Gamma(-\alpha_+)/\Gamma(1-\alpha_+) = -\alpha_+$ and $\alpha_+\Gamma(\alpha_+) = \Gamma(1+\alpha_+)$, we have:
    \begin{multline}
       \frac{\Lambda^{\alpha_-}\Gamma(\alpha_+ + \alpha_-)}{\Gamma(1+\alpha_-)\Gamma(\alpha_+)}~_2F_1(\alpha_++\alpha_-,\alpha_-;1+\alpha_-,-\Lambda)
       =
       1\\
       -\frac{\Lambda^{\alpha_-}\Gamma(\alpha_++\alpha_-)}{\Gamma(\alpha_-)\Gamma(1+\alpha_+)(1+\Lambda)^{\alpha_-+\alpha_+}}
       ~_2F_1(\alpha_++\alpha_-,1;1+\alpha_+;(1+\Lambda)^{-1}).
    \end{multline}
    With the following relation:
    \begin{equation}
        \frac{\Lambda^{\alpha_-}}{(1+\Lambda)^{\alpha_++\alpha_-}}
        = \frac{\lambda_-^{\alpha_-}\lambda_+^{\alpha_+}}{(\lambda_++\lambda_-)^{\alpha_++\alpha_-}},
    \end{equation}
    the desired result follows.

\blue{
\section{Gil-Pelaez formula for European option pricing}
\label{sec:gil-Pelaez}
% \ref{sec:gil-peleaz}

At time $t=0$, the price of a European option of strike $K$ and maturity $T$ can be expressed with a change of numeraire (see \cite[Theorem 2]{el1995changes}):
\begin{equation}
    C_{\mathrm{Eur}} = S_0 \widetilde{\mathbb{P}}(S_T>K) - Ke^{-rT}\mathbb{P}(S_T>K)
\end{equation}
where $\mathbb{P}$ is a risk neutral measure, $r$ is the risk-free rate and $\widetilde{\mathbb{P}}$ is defined by:
\begin{equation}
    \frac{\D \widetilde{\mathbb{P}}}{\D \mathbb{P}}\bigg|_{\mathcal{F}_T}
    =
    \frac{S_T}{S_0e^{rT}}.
\end{equation}
If $S_T = S_0\exp(X_T)$, by the Gil-Pelaez formula (see \cite{gil1951note}), we have:
\begin{equation}
    \begin{cases}  
    \displaystyle
    \mathbb{P}(S_T>K) = \frac{1}{2}+\frac{1}{\pi}\int_{0}^\infty 
    \frac{e^{-iu\log (K/S_T)}\phi_{X_T}(u)}{\mi u}\D u,\\
    \displaystyle
    \widetilde{\mathbb{P}}(S_T>K) = \frac{1}{2}+\frac{1}{\pi}\int_{0}^\infty 
    \frac{e^{-iu\log (K/S_T)}\widetilde{\phi}_{X_T}(u)}{\mi u}\D u,
    \end{cases}
\end{equation}
where $\phi_{X_T}$ is the chf. of $X_T$ under $\mathbb{P}$ and $\widetilde{\phi}_{X_T}$ is the chf. of $X_T$ under $\widetilde{\mathbb{P}}$, the first is assumed to be known and the second can be expressed with the first one with the relation:
\begin{equation}
    \forall u\in\R, ~\widetilde{\phi}_{X_T}(u) = \frac{\phi_{X_T}(u-\mi)}{\phi_{X_T}(-\mi)}.
\end{equation}
Both integrals can thus be computed (at least numerically) as soon as the chf. $\phi_{X_T}$ is known. Our implementation calls the {\ttfamily fypy}\footnote{\ttfamily https://github.com/jkirkby3/fypy} package in which this integral is computed with the {\ttfamily scipy} package. }

\end{document}